\documentclass[pra,aps,superscriptaddress,twocolumn,letter,nopacs,nofootinbib,longbibliography,notitlepage]{revtex4-1}

\usepackage{graphicx,color,amsmath,amsfonts,enumerate,amsthm,amssymb,mathtools,enumitem,thmtools,hyperref,mathdots,enumitem,centernot,bm,soul,bbm,stmaryrd}
\usepackage[capitalise, noabbrev]{cleveref}

\usepackage{floatrow}
\usepackage[caption=false]{subfig}

\usepackage{multirow}

\usepackage{tikz,pgfplots}
\pgfplotsset{compat=1.15}
\hypersetup{colorlinks=true,linkcolor=blue,citecolor=blue,urlcolor=blue}

\usepackage{algorithm,algpseudocode}
\usepackage{algorithmicx}

\def\F{ {\cal F} }

\def\L{ {\cal L} }

\def\Rcal{ {\cal R} }

\def\>{\rangle}
\def\<{\langle}
\newcommand{\abs}[1]{\left| {#1} \right|} 
\newcommand{\ketbra}[2]{\ensuremath{\left|#1\right\rangle\!\!\left\langle#2\right|}}

\newcommand{\tr}[1]{\mathrm{Tr}\left( #1 \right)}

\newcommand{\iden}{\mathbbm{1}}

\let\vv\v
\renewcommand{\v}[1]{\ensuremath{\boldsymbol #1}}

\definecolor{ppblue}{RGB}{46,117,182}
\definecolor{ppred}{RGB}{197, 90, 17}

\newcommand{\Tgate}[1]{T_{#1}}
\newcommand{\Tket}[1]{|T^\dagger_{#1}\rangle}
\newcommand{\Tkets}{|T_{\v{\phi}}^\dagger\rangle}
\newcommand{\Tketsprime}{|T_{\v{\phi}'}^\dagger\rangle}
\newcommand{\Tperpket}[1]{|T^{\dagger\perp}_{#1}\rangle}
\newcommand{\Tbra}[1]{\langle T^\dagger_{#1}|}
\newcommand{\Tbras}{\langle T_{\v{\phi}}^\dagger|}

\newcommand{\rele}{\varepsilon}
\newcommand{\vname}{deterministic number}
\newcommand{\rname}{projector rank}
\mathchardef\mhyphen="2D
\newcommand{\hide}[1]{}
\newcommand{\bbc}{\mathbb{C}}
\newcommand{\bbn}{\mathbb{N}}
\newcommand{\bbr}{\mathbb{R}}
\newcommand{\bbz}{\mathbb{Z}}
\newcommand{\set}[1]{\left\{#1\right\}}
\newcommand{\bitstring}[1]{\left\{0,1\right\}^{#1}}
\newcommand{\bigpr}[1]{{\rm Pr}\left(#1\right)}
\newcommand{\pr}{{\rm Pr}}
\newcommand{\brackn}[1]{\left(#1\right)}
\newcommand{\twonorm}[1]{\left\lVert #1 \right\rVert_2}
\newcommand{\sonenorm}[1]{\left\lVert #1 \right\rVert_1}
\newcommand{\widebar}[1]{\mathop{\overline{#1}}}
\newcommand{\order}[1]{O\left(#1\right)}
\newcommand{\orderlog}[1]{\tilde{O}\left(#1\right)}
\newcommand{\density}[1]{\left\lvert #1 \right\rangle \!\! \left\langle #1 \right\rvert}
\newcommand{\innerbraket}[2]{\left\langle \left. #1 \right\vert  #2 \right\rangle}
\newcommand{\innerbrakets}[2]{\langle #1 \vert  #2 \rangle}
\newcommand{\expval}[2]{\underset{#1}{\mathbb{E}}\left[{#2}\right]}
\newcommand{\childvar}[2]{\textsc{RawEstim}(#1,#2)}
\newcommand{\child}{\textsc{RawEstim}}

\newcommand{\parent}{\textsc{Estimate}}

\newcommand{\main}{\textsc{Estimate}}

\newcommand{\mainruntime}{\textsc{run-time}}

\newcommand{\compute}{\textsc{Compute}}

\newcommand{\qcompress}{\textsc{Compress}}

\newcommand{\gateseq}{\textsc{GateSeq}}

\newcommand{\zxform}{\textsc{ZX-Form}}

\newcommand{\extract}{\textsc{ConstrainStabs}}

\newcommand{\chformvar}[3]{\textsc{CH-Form}(#1,#2,#3)}

\newcommand{\overlapsvar}[2]{\textsc{SumOverlaps}(#1,#2)}
\newcommand{\overlaps}{\textsc{SumOverlaps}}
\newcommand{\cextent}{\xi^*}
\newcommand{\ubex}{\xi^*}
\newcommand{\rset}{\left\{0,1,\ldots,\min~\left\{t,n-w\right\}\right\}}
\newcommand{\vset}{\left\{0,1,\ldots,w\right\}}

\newcommand{\gen}[1]{\langle #1 \rangle}
\newcommand{\genset}[2]{\mathcal{G}({#1}, {#2})}
\newcommand{\pauli}[1]{\mathcal{P}_{#1}}
\newcommand{\pfac}[2]{\abs{#1}_{#2}}
\newcommand{\pphase}[1]{\omega(#1)}
\newcommand{\rega}{\rm a}
\newcommand{\regb}{\rm b}
\newcommand{\regc}{\rm c}
\newcommand{\parentrounds}{K}
\newcommand{\modeltau}{\tau_{\mathrm{model}}}

\usepackage{color}

\definecolor{hpCol}{RGB}{150, 0, 150}
\definecolor{newtextCol}{RGB}{150, 20, 150}

\newcommand{\newtext}[1]{#1}

\newcommand{\estimate}{\textsc{Estimate}} 
\newcommand{\join}[2]{{{#1}\!:\!{#2}}}
\newcommand{\CXgate}[2]{{\operatorname{CX_{#1,#2}}}}
\newcommand{\CX}{{\operatorname{CX}}}
\newcommand{\CZgate}[2]{{\operatorname{CZ_{#1,#2}}}}
\newcommand{\Sgate}[1]{{\operatorname{S_{#1}}}}
\newcommand{\SWAPgate}[2]{{\operatorname{SWAP_{#1,#2}}}}

\newcommand{\ket}[1]{\left\lvert{#1}\right\rangle}
\newcommand{\bra}[1]{\left\langle{#1}\right\rvert}
\newcommand{\braket}[2]{\left\langle{#1}\middle|{#2}\right\rangle}

\theoremstyle{plain}
\newtheorem{thm}{Theorem}
\newtheorem{lem}[thm]{Lemma}

\newtheorem{cor}[thm]{Corollary}

\theoremstyle{definition}
\newtheorem{defn}[thm]{Definition}
\newtheorem{rmk}{Remark}

\begin{document}

% -------------------------------------------------------------
% TITLE & ABSTRACT
% -------------------------------------------------------------

\title{Fast estimation of outcome probabilities for quantum circuits}
\date{\today}

\author{Hakop Pashayan}
\affiliation{Institute for Quantum Computing and Department of Combinatorics and Optimization, University of Waterloo, ON, N2L 3G1 Canada}
\affiliation{Perimeter Institute for Theoretical Physics, Waterloo, ON, N2L 2Y5 Canada}

\author{Oliver Reardon-Smith}
\affiliation{Faculty of Physics, Astronomy and Applied Computer Science, Jagiellonian University, 30-348 Krak\'{o}w, Poland}

\author{Kamil Korzekwa}
\affiliation{Faculty of Physics, Astronomy and Applied Computer Science, Jagiellonian University, 30-348 Krak\'{o}w, Poland}

\author{Stephen D. Bartlett}
\affiliation{Centre for Engineered Quantum Systems, School of Physics, The University of Sydney, Sydney, NSW 2006, Australia}

\begin{abstract}
	We present two classical algorithms for the simulation of universal quantum circuits on $n$ qubits constructed from $c$ instances of Clifford gates and $t$ arbitrary-angle $Z$-rotation gates such as $T$ gates. Our algorithms complement each other by performing best in different parameter regimes. The $\main$ algorithm produces an additive precision estimate of the Born rule probability of a chosen measurement outcome with the only source of run-time inefficiency being a linear dependence on the stabilizer extent (which scales like $\approx 1.17^t$ for $T$ gates). Our algorithm is state-of-the-art for this task: as an example, in approximately $13$ hours (on a standard desktop computer), we estimated the Born rule probability to within an additive error of $0.03$, for a $50$-qubit, $60$ non-Clifford gate quantum circuit with more than $2000$ Clifford gates. Our second algorithm, $\compute$, calculates the probability of a chosen measurement outcome to machine precision with run-time $\order{2^{t-r} t}$ where $r$ is an efficiently computable, circuit-specific quantity.  With high probability, $r$ is very close to $\min \set{t, n-w}$ for random circuits with many Clifford gates, where $w$ is the number of measured qubits. $\compute$ can be effective in surprisingly challenging parameter regimes, e.g., we can randomly sample Clifford+$T$ circuits with $n=55$, $w=5$, $c=10^5$ and $t=80$ $T$ gates, and then compute the Born rule probability with a run-time consistently less than $10$ minutes using a single core of a standard desktop computer. We provide a C+Python implementation of our algorithms \newtext{and benchmark them using random circuits, the hidden shift algorithm and the quantum approximate optimization algorithm (QAOA).}
\end{abstract}

\maketitle  

% -------------------------------------------------------------
% SEC. I - INTRODUCTION
% -------------------------------------------------------------

\section{Introduction}
\label{sec:intro}

With the rapid advancement in experimental control over noisy intermediate-scale quantum (NISQ) systems~\cite{preskill2018quantum}, claims of quantum advantage~\cite{harrow2017quantum} have recently been made using several different platforms~\cite{google2019supremacy,zhong2020quantum}. In addition to the enormous challenges in building complex quantum devices that can exhibit quantum advantage, two important but difficult problems are: how to test if these devices are operating as intended, and how to make effective computational use of NISQ systems. Alongside the improvements to quantum computing hardware, innovative ideas continue to improve the methods for simulating these devices on classical computers. Classical simulators are cheaper, more accessible, more reliable and sometimes even faster than modern quantum computers, and so classical simulation algorithms continue to play a significant role in assessing the performance of quantum devices and testing the feasibility or performance of new proposals for NISQ device applications. In this work we present a suite of classical algorithms that are state-of-the-art for simulating quantum circuits. This suite significantly broadens and diversifies the size and type of quantum circuits that can be classically simulated within a feasible run-time. In particular, our algorithms can access regimes that will be important for the verification of NISQ devices and the assessment of NISQ proposals such as the quantum machine learning protocol of Ref.~\cite{Havlicek2019}. 

In addition to other more pragmatic goals, research into the classical simulation of quantum circuits is a means of studying and quantifying the distinction between the computational power of quantum and classical computers. Here, one aims to ``simulate'' certain properties of a family of quantum circuits using only classical means in order to upper bound the classical resource costs associated with the simulation task. Quantum computational power translates into the exponential time complexity of classically simulating arbitrary sequences of universal quantum circuits. The hardness of simulating universal quantum circuits should be contrasted with particular classes of quantum circuits that can be efficiently simulated classically. The celebrated and ubiquitous example is given by stabilizer circuits~\cite{gottesman1998heisenberg}. These consist of an $n$-qubit system initialized in a computational basis state with gates composed of $poly(n)$ elementary Clifford gates and measurements in the computational basis. In this framework, the restriction to a non-universal gate set does not allow the quantum system to explore the full richness of the quantum state space and permits a $poly(n)$ run-time classical simulator of this family of circuits. \newtext{Subsequently these classical simulators have been extended to a universal gate-set consisting of the Clifford gates complemented with an additional elementary gate (commonly the $T$ gate) which promotes the gate-set to universality~\cite{bravyi2016improved, bravyi2019simulation, seddon2020quantifying}.} The run-time of these simulators grows at most polynomially in all variables except $t$, the number of elementary non-Clifford gates. This is \newtext{among the most notable achievements} of modern classical simulators since they can not only efficiently simulate stabilizer circuits but have a run-time sensitivity to the degree of departure from stabilizer circuits. 

The task of classically ``simulating'' a quantum circuit can take several different forms~\cite{jozsa2013classical}. A so-called \emph{weak simulator} is a classical algorithm which returns samples drawn from the exact or approximate outcome probability distribution of a given quantum circuit; while a \emph{strong simulator} calculates or approximates these probabilities directly. It is important to further specify the degree of accuracy required of the strong or weak simulator as this has a dramatic impact on the computational complexity associated with the simulation task (see Ref.~\cite{pashayan2017sampling} for an extended discussion). For example, the ability to exactly compute a desired Born rule probability allows one to solve extremely difficult optimization and counting problems believed to be well beyond the reach of even ideal \emph{quantum} computers restricted to polynomial run-time (problems that are hard for the complexity class NP and even \#P~\cite{valiant1979sharpp}). It is strongly believed that quantum computers cannot even approximate an arbitrary Born rule probability $p$ to within an additive approximation error $\epsilon$ in a run-time that scales at most polylogarithmically in $1/\epsilon$ or polynomially in $\epsilon/p$ since these are also \#P-hard~\cite{goldberg2017complexity,Fujii2017commuting}. Nevertheless, many works have focused on such simulators for a restricted family of quantum circuits, most notably see Ref.~\cite{gottesman1998heisenberg, aaronson2004improved} for classical strong simulators of stabilizer circuits and Refs.~\cite{valiant2002quantum, terhal2002classical} for match-gate circuits. \newtext{Our $\compute$ algorithm (introduced shortly) also satisfies this strong notion of simulation. Specifically, it computes target Born rule probabilities to machine precision.
In contrast, ideal universal quantum computers can \emph{estimate} Born rule probabilities to within an additive error of $\pm \epsilon$ in a run-time that scales polynomially in $1/\epsilon$. This means that rather than outputting the target Born rule probability $p$, they output an estimate $\hat{p}$ that is with high probability in the interval $[p-\epsilon,p+\epsilon]$. Further improvement in accuracy requires additional run-time which scales exponentially in each additional digit of precision required of $\hat{p}$.} This simulation task is also computationally easier than weak simulation (using commonly employed notion of approximate sampling)~\cite{pashayan2017sampling}. Importantly, this task is also believed to be hard for classical computers~\cite{bernstein1997quantum} and has received limited attention. Our $\main$ algorithm targets this natural notion of simulation as it is within the capabilities of universal quantum computers and finds useful near term applications.

In this paper, we present an estimator of outcome probabilities for quantum circuits on $n$ qubits consisting of~$c$ Clifford gates and $t$ instances of $T_\phi$ gates, with $w$ qubits being measured. Our estimator is actually a pair of distinct algorithms, $\main$ and $\compute$, that work as part of a larger procedure, utilizing their respective performance advantages in complementary regimes, and is state-of-the-art for the task. A novel component of our algorithm is a tailored circuit analysis, $\qcompress$, which runs in polynomial time in $\{n,c,t\}$ and outputs a compressed representation of the circuit along with a parameter $r$ \newtext{that is a key driver of run-time.} For large values of $r$, the $\compute$ algorithm can compute the exact Born rule probabilities in feasible run-times (up to machine precision). More precisely, its run-time depends exponentially on  $(t-r)$ where $r$ can generally be as large as the minimum of $t$ and the number of unmeasured qubits $(n-w)$. In this way we identify circuits that are easy to simulate using $\compute$ i.e. circuits with small $(t-r)$. Thus, our first key contribution is the identification of a previously unknown class of quantum circuits that can be efficiently simulated classically, analogous to the case of Clifford circuits with small $T$ gate count. What is more, empirical evidence indicates that randomly generated circuits generically have $r$ close to its maximal allowed value, $r\approx \min \set{t, n-w}$. \newtext{Additionally, our algorithm permits the identification of an effective number of $T$ gates in the circuit which can improve run-times exponentially in the $T$-count reduction. We assessed the empirical run-time performance of $\compute$ using the hidden shift and the quantum approximate optimization algorithm (QAOA) benchmarks from Ref.~\cite{bravyi2019simulation}. We observed performance improvements of the order $10^3$ to $10^5$ compared to the prior state-of-the-art simulation methods~\cite{bravyi2019simulation}.}

Our second key contribution is the $\main$ algorithm, which is highly complementary to $\compute$, performing particularly well for rare and difficult to simulate circuits (those with small $r$), and significantly improving the run-time compared to the previous state-of-the-art. It is an additive precision estimator, whose run-time depends exponentially on $t$, and polynomially on $\{n,w,c,r,1/\epsilon\}$, where $\epsilon$ is the desired estimation error. The exponential dependence on the non-Clifford gates is quantified by a number $\xi^*$, called the \emph{stabilizer extent}, and is the same as in Ref.~\cite{bravyi2016improved,bravyi2019simulation}. However, by extending existing methods and developing a number of new techniques, we substantially improve the scaling of the polynomial prefactors allowing our estimation algorithm to perform many orders of magnitude faster than those of Refs.~\cite{bravyi2016improved,bravyi2019simulation} in certain practically relevant parameter regimes. In particular, we improve the dependence of the run-time on the probability $p$ being estimated. First, by expressing the target Born rule probability as the norm of the average of exponentially many vectors, we can employ Monte Carlo techniques. This allows us to exploit a concentration inequality for vectors, which has not previously been used in the simulation context.
This has the desired effect of reducing run-time for estimation tasks where the target Born rule probability is small or alternatively allowing smaller probability events to be estimated to higher accuracy in a given time. Additionally, this approach bridges the conceptual divide between two well known simulation techniques: stabilizer rank based simulators and quasi-probabilistic simulators.
Second, instead of using a naive approach employing the upper bound $p\leq 1$ to obtain an upper bound on the  estimatation error, we make use of a novel algorithm that iteratively learns tighter upper bounds on $p$ to substantially tighten our upper bound on estimation error, improving the run-time. As a result, in the regime where $\epsilon\leq p$, we improve the total run-time by a factor scaling as $p^{-3}$ compared to the best previously known algorithms~\cite{bravyi2016improved,bravyi2019simulation}.

Given the current technological landscape, our algorithms offer a feasible, reliable and accessible means to simulate intermediate scale quantum computations. We believe this is a timely and important contribution with key applications including the characterization and verification of NISQ devices and the appraisal of proposals for NISQ device applications.

The paper is structured as follows. In Section~\ref{sec:overview}, we first provide a brief review of Born rule probability estimation problem, and then present our results. These include the description of the three aforementioned classical algorithms, $\qcompress$, $\compute$ and $\main$, together with theorems detailing their run-times. In Section~\ref{sec:performance} we compare our results with the previous state-of-the-art and analyse the performance of our algorithms for quantum circuits in various parameter regimes. \newtext{This section also includes the results of the numerical simulations on random, hidden shift and QAOA circuits that we employed to benchmark our algorithms. We conclude with an outlook in Sec.~\ref{sec:outlook}.} The detailed proofs of our main theorems can be found in Appendices~\ref{app:compress}-\ref{app:parent}, while the Supplemental Material contains proofs of intermediate lemmas.

% -------------------------------------------------------------
% SECTION II - RESULTS
% -------------------------------------------------------------

\section{Results}
\label{sec:overview}

In this section, we formally state the problem that we solve, and we present our main results. These consist of two classical algorithms, $\compute$ and $\main$, that either exactly compute or estimate the outcome probability $p$ given the description of a quantum circuit. These make use of two auxiliary algorithms, $\qcompress$ and $\child$, the relation between them is presented in Fig.~\ref{fig:flowchart}. The C+Python implementation of our algorithms used to generate Figures~\ref{fig:compute-perf},~\ref{fig:optimize-perf},~\ref{fig:run-time-bounds},~\ref{fig:hidden-shift-2}, and~\ref{fig:qaoa-full-fig} and tables~\ref{tab:hidden-shift-comparison} and~\ref{tab:qaoa-comparison} can be found in Ref.~\cite{smith2020clifford}.

\begin{figure*}[t]
	\centering
        \includegraphics[width=\columnwidth]{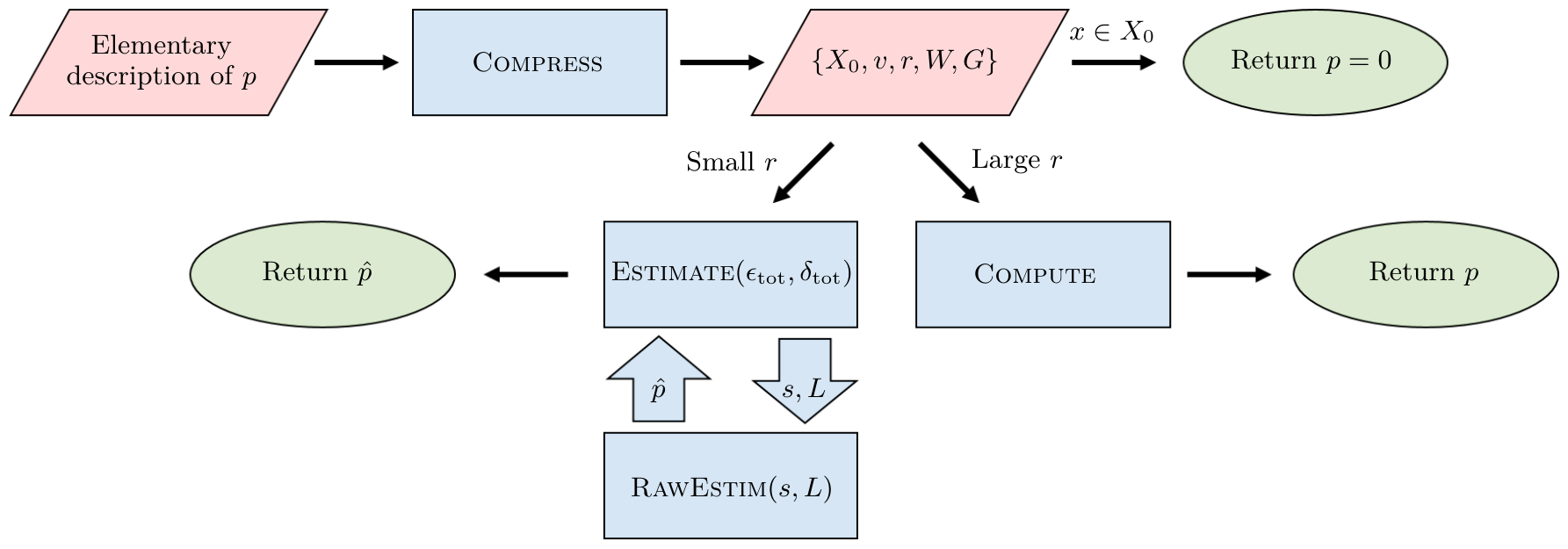}
	\caption{\label{fig:flowchart}\textbf{Flowchart for our main algorithms.} A component of the output of $\qcompress$ is the identification of $v$ measured qubits which have a deterministic outcome. Here, $X_0$ represents the set of measurement outcomes that are inconsistent with this deterministic outcome (also output by $\qcompress$).
	Owing to the $\compute$ algorithm's run-time having an inverse exponential dependence on $r$ (like $2^{t-r}$), the choice of preferred algorithm is informed by the size of $r$ (relative to $t$).
	The broad blue arrows indicate multiple calls by $\parent$ to $\child$ with different parameters $s$ and $L$, and multiple outputs of $\hat{p}$ that are fed back to $\parent$.   
 }
\end{figure*}

\subsection{Statement of problem}

Quantum circuits that initiate in a computational basis state, evolve under Clifford unitary transformations and are measured in the computational basis are efficiently classically simulable by the Gottesman-Knill theorem~\cite{gottesman1998heisenberg, aaronson2004improved}. The group of Clifford unitary transformations is generated by the gate-set $\set{S, H, CX}$ (and contains $CZ$):
\begin{equation}
    \begin{aligned}
	S&\equiv\begin{bmatrix}
		1 & 0\\ 
		0 & i 
	\end{bmatrix},&
	H&\equiv\frac{1}{\sqrt{2}}\begin{bmatrix}
		1 &~~1\\ 
		1 &-1 
	\end{bmatrix},\\ 
	CX&\equiv\begin{bmatrix}
		1 & 0 & 0 & 0\\
		0 & 1 & 0 & 0\\
		0 & 0 & 0 & 1\\
		0 & 0 & 1 & 0
	\end{bmatrix}, &
	CZ&\equiv\begin{bmatrix}
		1 & 0 & 0 & ~~0\\
		0 & 1 & 0 & ~~0\\
		0 & 0 & 1 & ~~0\\
		0 & 0 & 0 & -1
	\end{bmatrix}\text{.}
	\end{aligned}
\end{equation}
This gate-set is promoted to universality by the inclusion of a non-Clifford gate, \newtext{in particular the} diagonal gate
\begin{equation}
	\Tgate{\phi}\equiv\begin{bmatrix}
		1 & 0\\ 
		0 & e^{i\phi}
	\end{bmatrix}, 
\end{equation}
where $\phi \in (0,\pi/2)$ can be arbitrary.  A standard choice is the $T$ gate~\cite{boykin1999}, defined with $\phi=\pi/4$.

We consider a system composed of $n$ qubits initially prepared in the state $\ket{0}^{\otimes n}$. The system then evolves according to a unitary transformation $U$ to the final state \mbox{$U\ket{0}^{\otimes n}$}.  We consider circuits where $U$ is constructed via a sequence of elementary gates from the gate-set consisting of $S, H, CX, CZ$ and $\Tgate{\phi_j}$ gates, where each $\phi_j \in (0,\pi/2)$ can be arbitrary. We will call such a description of the circuit an \emph{elementary description} of $U$ and reserve the variables $c, h$ and $t$ to respectively denote the number of all Clifford gates, Hadamard gates and non-Clifford gates occurring in this description. Associated with the non-Clifford gates $\Tgate{\phi_j}$, for $j\in [t]$, are the non-stabilizer single qubit states $\Tket{\phi_j}$ and the product state $\Tkets$ defined by:
\begin{equation}
    \Tkets:=\Tket{\phi_1} \otimes \ldots \otimes\Tket{\phi_t},\qquad    \Tket{\phi_j} :=\Tgate{\phi_j}^\dagger H\ket{0}. 
\end{equation}
We use $\cextent$ to denote the quantity known as the \emph{stabilizer extent}~\cite{bravyi2019simulation} of the state $\Tkets$, which is formally defined in Eq.~\eqref{eq:extent} in Appendix~\ref{app:child}. For the moment, we simply note that for product states of qubit systems the stabilizer extent is multiplicative~\cite{bravyi2019simulation}, i.e., $\cextent$ is a product of stabiliser extents $\xi(\Tket{\phi_j})$ (so that $\cextent$ is a mild exponential in $t$), and that the extent of the single qubit state $\xi(\Tket{\phi})$ is a simple function of $\phi$ that is upper bounded by $\approx 1.17$. 

Given some ordered subset $\mathcal{J}\subseteq [n]$ of $w$ qubits to be measured in the computational basis and some outcome $x=(x_1, \ldots, x_w)\in \bitstring{w}$, our aim is to compute or estimate the probability $p(\mathcal{J},x)$ of observing the outcome $x$ when measuring the final state \mbox{$U\ket{0}^{\otimes n}$}. Without loss of generality, we will assume that the first $w$ qubits are measured and hence $\mathcal{J}=\set{1,\ldots, w}$. We refer to the first $w$ qubits as the \emph{measured} register `$\rega$' and the remaining $(n-w)$ qubits as the \emph{marginalized} register `$\regb$'.  Our central goal is to exactly or approximately compute the Born rule probability:
\begin{equation}
	\label{eq:target-prob}
	p:=\twonorm{\bra{x}_{\rega} U\ket{0}^{\otimes n}_{\rega\regb}}^2,
\end{equation}
\newtext{where the notation $\bra{0}_{\rega}$ should be interpreted as having an implicit identity matrix acting on the qubits in register $\regb$}. The Born rule probability in the above equation can be described by specifying $n\in \bbn$, $w\in [n]$, $x\in \bitstring{w}$ and an elementary description of an $n$-qubit unitary $U$. We will refer to this information as an \emph{elementary description} of $p$.

We note that estimating or exactly computing the Born rule probability for a single qubit computational basis measurement is sufficient to also allow estimation or computation (respectively) of expectation values of arbitrary tensor products of Pauli-operators. We give details of this construction in Supplemental Material Sec.~\ref{sec:supp-pauli-expectation-values}.

% -------------------------------------------------------------

\subsection{The {\normalfont\texorpdfstring{$\qcompress$}{Compress}} algorithm}

The $\qcompress$ algorithm is the starting point of our Born rule probability estimator.  It takes as input the elementary description of $p$, and the purpose of this algorithm is to efficiently transform the elementary description of the Born rule probability into an alternative form, where we can decide what type of estimator is most suitable.

The $\qcompress$ algorithm is composed of three steps, which we summarize here; a detailed description is presented in Appendix~\ref{app:compress}. In the first step, we use a ``reverse gadgetization'' of $\Tgate{\phi_j}$ gates (see Eq.~\eqref{eq:reverse-gadget}) to re-express the general circuit $U$ acting on $\ket{0}_{\rega\regb}^{\otimes n}$ as a Clifford circuit $V$ acting on \mbox{$\ket{0}_{\rega\regb}^{\otimes n}\otimes\ket{0}_{\regc}^{\otimes t}$}, with the $t$ ancillary qubits in register~`$\regc$' post-selected on the state $\Tkets$. Thus, we re-express the Born rule probability from Eq.~\eqref{eq:target-prob} as:
	\begin{equation}
		p=2^{t}\twonorm{\bra{x}_{\rega} \Tbras_{\regc}V\ket{0}_{\rega\regb\regc}^{\otimes n+t}}^2.
	\end{equation}
\newtext{Second, we express the probability $p$ in terms of the trace of the product of two projectors, $\Pi_G = V \ketbra{0}{0}_{\rega\regb\regc}^{\otimes n+t}V^\dagger$ and  $\ketbra{x}{x}_{\rega} \otimes I_{\regb} \otimes \Tkets \! \Tbras_{\regc}$, and identify constraints that a given stabilizer generator $g\in G$ has to satisfy in order to contribute non-trivially to the Born rule probability. By imposing these constraints, we are able to remove some stabilizer generators and qubits from consideration and} re-express the Born rule probability $p$ as:
	\begin{equation}
		\label{eq:p-for-compute}
		p=2^{v-w}  \Tbras \prod_{i=1}^{t-r}(I+g_i) \Tkets,
	\end{equation}
where $r\in \rset$ and $v\in \vset$ are circuit specific quantities (dependent on $V$) that are efficiently computable, and $\{g_i\}_{i=1}^{t-r}$ are $t$-qubit Pauli operators (generators of a stabiliser group).  Finally, by explicitly constructing a gate sequence based on a stabiliser generator matrix, we re-express $p$ as:
	\begin{equation}
		\label{eq:p-for-child}
		p=2^{t-r+v-w}\twonorm{\bra{0}^{\otimes t-r}W \Tkets}^2,
	\end{equation}
where $W$ is a $t$-qubit Clifford circuit of length $O(t^2)$ with $O(t)$ Hadamard gates.  

The performance of $\qcompress$ is captured by the following theorem: 

\begin{thm}[\texorpdfstring{$\qcompress$}{Compress} algorithm]
	\label{thm:compress}
	Given an elementary description of $p$, $\qcompress$ outputs the \emph{\vname} \mbox{$v\in \vset$}, specifies a size $v$ subset of the measured qubits that have deterministic outcomes and provides the measurement outcomes these qubits must produce. If the input $x$ is consistent with these deterministic outcomes, then the algorithm outputs the \emph{\rname} \mbox{$r\in \rset$} and an elementary description of the $t$-qubit Clifford unitary $W$ together with a set $G$ of $(t-r)$ Pauli operators $g_i$ on $t$ qubits such that:
	\begin{equation}
	\label{eq:compressed}
	\begin{split}
		p&=2^{t-r+v-w}\twonorm{\bra{0}^{\otimes t-r}W \Tkets}^2\\
		&=2^{v-w}  \Tbras \prod_{i=1}^{t-r}(I+g_i) \Tkets.
	\end{split}
	\end{equation}
	The run-time $\tau_{\qcompress}$ of the algorithm scales as
	\begin{equation}
		\label{eq:qcompress run-time}
		\tau_{\qcompress} = {\rm poly}(n,c,t).
	\end{equation}	
\end{thm}
\noindent The proof of this theorem is given in Appendix~\ref{sec:gate-seq}. 
If $x$ is not consistent with the deterministic outcomes specified by the $\qcompress$ algorithm then we immediately conclude that $p=0$ and we have efficiently calculated the target Born rule probability. Otherwise, we have two choices: either to use the $\compute$ or the $\main$ algorithm. In making this choice, the size of $r$ relative to $t$ will be important. The quantity $r$ is the binary logarithm of the rank of the stabilizer projector defined by $G$ but we will call it the \emph{\rname} for short. On the one hand, $r$ can be interpreted as a measure of quantum circuit's compressibility for the $\compute$ algorithm, since we will show that the only exponential component of the run-time of $\compute$ is a factor of $2^{t-r}$. On the other hand, $r$ can be interpreted as an incompressibility measure for the $\main$ algorithm, since we will show that its run-time scales as $r^3$, because $r$ will be the number of qubits appearing in the most computationally expensive step of $\main$.

\newtext{
\begin{rmk}[Reducing the $T$-count]\label{rmk}
An extension of Theorem~\ref{thm:compress} provides us with a polynomial-time algorithm to reduce the $T$-count $t$ to an effective $T$-count $t^\prime \leq t$. This procedure will result in the creation of ``primed'' counterparts to the variables $t,r$ and $\xi^*$ that satisfy the following properties: $\xi^\prime \leq \xi^*$, $t^\prime \leq t$, $r^\prime \leq r$ and $t^\prime - r^\prime \leq t -r$.

This procedure has the effect of improving the run-times of $\compute$, $\child$ and $\parent$ exponentially in the $T$-count reduction. Details of the procedure may be found in Appendix~\ref{sec:T-count reduction}. For the sake of simplifying notation we give run-times for the subsequent $\compute$, $\child$ and $\parent$ algorithms assuming the ``unprimed'' variables and note that they remain true if the unprimed variables are replaced with their primed counterparts. We specify the run-times for the general case in Appendix~\ref{sec:T-count reduction}. Although we do not expect substantial reductions in $\Tgate{}$ gate count for worst-case circuits, we see $t^\prime$ values substantially lower than $t$ driving the dramatic performance improvements of our algorithms on the hidden-shift and QAOA benchmark we consider in Sec.~\ref{sec:performance-existing-bounds}.
\end{rmk}}

\newtext{We note that our $\qcompress$ algorithm builds on a key idea from Ref.~\cite{bravyi2016improved}; that from a given $n+t$ qubit stabilizer projector $\Pi$ one can compute a $t$ qubit stabilizer project $\Pi'$ and an integer $u$ such that $\bra{0^n}\Pi\ket{0^n}=2^{-u}\Pi'$. Our analysis, however, goes substantially further. In particular, the algorithmic identification of deterministic outcomes, the reduction of the effective $T$-count and the identification of the parameter $r$ as a key driver of run-time are, to our knowledge, new.}

\subsection{The {\normalfont\texorpdfstring{$\compute$}{Compute}} algorithm}

If $\qcompress$ outputs a large value of $r$ relative to $t$ (in the sense that $(t-r)$ is small), the $\compute$ algorithm is likely to outperform our $\main$ algorithm.
The key idea behind the $\compute$ algorithm is that the target quantity can be calculated in run-time $\order{2^{t-r}t}$ by directly summing the $2^{t-r}$ terms appearing in the expansion of Eq.~\eqref{eq:p-for-compute}. The $\compute$ algorithm slightly improves on this by using a Gray code ordering to cycle through the terms in the sum with minimal effort. Full details of this algorithm are given in Appendix~\ref{app:compress}.

The performance of $\compute$ is captured by the following theorem:

\begin{thm}[\compute~algorithm]
	\label{thm:compute}
	Given the output of the $\qcompress$ algorithm, $\compute$ outputs $p$ (up to machine precision) in the run-time:
	\begin{align}\label{alg:run-time compute}
		\tau_\compute=\order{2^{t-r}t}.
	\end{align}
\end{thm}
\noindent The proof of this theorem is given in Appendix~\ref{app:compute_proof}.

\subsection{The {\normalfont\texorpdfstring{$\main$}{Estimate}} algorithm}

If the \rname~is too small and $\tau_\compute$ becomes infeasible, we may use our main result: the $\main$ algorithm.  This algorithm produces Born rule probability estimates satisfying a desired additive error and failure probability. Our $\main$ algorithm makes use of a crucial subroutine we call $\child$. This subroutine produces an estimate $\hat{p}$ of $p$ given run-time constraints specified by a pair of parameters $s$ and $L$. Optimal values for these parameters leading to estimates that satisfy a desired additive error and failure probability are determined by $\parent$. Here, we will first summarize the $\child$ subroutine (details of which are presented in Appendix~\ref{app:child}), and then briefly describe our $\main$ algorithm (with details in Appendix~\ref{app:parent}).

At its core, the $\child$ algorithm uses a concentration inequality (see Lemma~\ref{lem:hoeffding}) to bound the norm between a target vector $\ket{\mu}$ and a ``simulated'' approximation $\ket{\widebar{\psi}}$. The target quantity $p$ is directly related to the Euclidean norm of the target vector $\ket{\mu}$. Thus, an estimate of the Euclidean norm of the approximation vector $\ket{\widebar{\psi}}$ is used to compute an estimate of $p$. The approximation vector $\ket{\widebar{\psi}}$ is a uniform superposition of $s$ randomly sampled stabilizer states. The sample space of stabilizer states and the probability distribution over these is directly constructed from stabilizer decompositions of magic states $\Tkets$. \newtext{The expert reader may note that
$\ket{\widebar{\psi}}$ is an unbiased estimator of $\ket{\mu}$ constructed as a vector level analogue of quasi-probabilistic estimators (cf. \cite{pashayanthesis, pashayan2015estimating, seddon2020quantifying}) and our Lemma~\ref{lem:hoeffding} can been seen as a vector level analogue of the Hoeffding inequality employed in these works. $\ket{\widebar{\psi}}$ is also closely related to the random state $\ket{\Omega}$ in the sparsification lemma of Ref.~\cite{bravyi2019simulation} and our Lemma~\ref{lem:hoeffding} can been seen as an analogue of the sparsification tail bound of Ref.~\cite{bravyi2019simulation}. Thus, at a conceptual level, $\child$ unifies the quasi-probabilistic and stabilizer rank based approaches.} The $\child$ algorithm also uses a number of novel techniques to improve the run-time.

The $\child$ algorithm is composed of three steps, which are briefly summarized as follows.  
In the first step, we decompose the state $\Tkets$ appearing in Eq.~\eqref{eq:p-for-child} into a superposition of stabilizer states, thus re-expressing the Born rule probability $p$ as the length $\twonorm{\ket{\mu}}^2$ of the following vector:
	\begin{equation}
		\ket{\mu}=\sum_y q(y) \ket{\psi(y)}.
	\end{equation}
	Here, the sum is over all binary strings $y$ of length $t$, $q(y)$ is a product probability distribution and $\ket{\psi(y)}$ are unnormalised stabiliser states on $r$ qubits given by:
	\begin{equation}\label{eq:step3 psi y}
		\ket{\psi(y)}\propto \bra{0}^{\otimes t-r}W \ket{\tilde{y}},
	\end{equation}
	where $\ket{\tilde{y}}$ is a $t$-fold tensor product of single qubit stabiliser states with $y_j=0$ or $y_j=1$ meaning that qubit $j$ is in a stabiliser state $\ket{+}$ or $\ket{-i}$.
	We independently sample bit strings $y$, with probability $q(y)$, a total of $s$ times, each time returning an $r$-qubit stabilizer state $\ket{\psi_j}$ equal to $\ket{\psi(y)}$ for the sampled $y$ (the fast computation of $\ket{\psi(y)}$ is discussed in the next step). The uniform superposition of all $s$ sampled stabilizer states $\ket{\widebar{\psi}}$ is used to approximate $\ket{\mu}$. The distance between $\ket{\mu}$ and $\ket{\widebar{\psi}}$ for a given $s$ is sensitive to the lengths of $\ket{\psi(y)}$, which we upper-bound for all $y$ using the stabilizer extent: 
	\begin{equation}\label{eq:step3 max psi y}
		\max_{y}{\twonorm{\ket{\psi(y)}}^2}\leq \ubex.
	\end{equation}
	
In the second step, each sampled state $\ket{\psi_j}$ in the previous step is an unnormalised stabiliser state given by Eq.~\eqref{eq:step3 psi y}. We compute and represent these states in the phase sensitive CH form introduced in Ref.~\cite{bravyi2019simulation}. In order to obtain the needed CH forms of $\ket{\psi_j}$ we do the following. First, even before taking any samples, we pre-compute the CH form of $W\ket{\tilde{0}\dots\tilde{0}}$ using the phase-sensitive simulator of Ref.~\cite{bravyi2019simulation}. Then,	for each sampled $y$, we efficiently update the CH form of $W \ket{\tilde{0}\dots\tilde{0}}$ to get the CH form of $W \ket{\tilde{y}}$. Finally, we use a novel subroutine that efficiently yields the CH form of the post-selected state $\bra{0}^{\otimes t-r}W \ket{\tilde{y}}$, and so of $\ket{\psi(y)}$. The vector $\ket{\widebar{\psi}}$ is represented and stored as the CH forms of $\ket{\psi_j}$ for $j\in [s]$.
			
Finally, as the third step, we employ the fast norm estimation algorithm from Ref.~\cite{bravyi2019simulation} to estimate the norm of~$\ket{\widebar{\psi}}$. The square of the returned norm is the $\child$ algorithm's Born rule probability estimate $\hat{p}$.

The $\child$ algorithm's performance is characterized by the following theorem:
\begin{thm}[$\child$ algorithm]
	\label{thm:child alg}
	Given the output of the $\qcompress$ algorithm and two positive integers $s$ and $L$, $\child$ outputs an estimate $\hat{p}$ of the outcome probability $p$ such that for all $\epsilon_{\rm tot}>0$ and \mbox{$\epsilon \in (0,\epsilon_{\rm tot})$}:
	\begin{equation}\label{eq:thm3}
	\begin{aligned}
	\hspace{-2ex}\bigpr{\abs{\hat{p}-p}\geq \epsilon_{\rm tot}}\leq &2e^2 \exp\left(\frac{-s(\sqrt{p+\epsilon}-\sqrt{p})^2}{2(\sqrt{\ubex}+\sqrt{p})^2} \right)\\+&\exp\left(-\left(\frac{\epsilon_{\rm tot}-\epsilon}{p+\epsilon}\right)^2 L \right)=:\delta_{\mathrm{tot}}.
	\end{aligned}
	\end{equation}
	The run-time $\tau_{\child}$ of the algorithm scales as
	\begin{equation}
		\label{eq:child run-time}
		\tau_{\child} = \order{s t^3 + sL r^3}.
	\end{equation}	
\end{thm}
\noindent The proof of this theorem is presented in Appendix~\ref{app:child}.

Given the output of the $\qcompress$ algorithm and accuracy parameters \mbox{$\epsilon_{\rm tot}, \delta_{\rm tot}>0$}, $\main$ outputs an estimate $\hat{p}$ of the outcome probability $p$ such that:
	\begin{align}\label{thm:main}
		\bigpr{\abs{\hat{p}-p}\geq \epsilon_{\rm tot}}\leq \delta_{\rm tot}.
	\end{align}
The $\child$ algorithm is used as a subroutine of the $\main$ algorithm to achieve the desired error $\epsilon_{\rm tot}>0$ and failure probability $\delta_{\rm tot}>0$.  With the proper choice of input parameters $s$ and $L$, the $\child$ algorithm can achieve a desired failure probability $\delta_{\mathrm{tot}}$ of the estimate $\hat{p}$. However, this proper choice depends on the unknown quantity $p$ that we want to estimate. One could always make the conservative choice of $p=1$ in Eq.~\eqref{eq:thm3}, which will result in well-defined but highly suboptimal (too large) input parameters $s$ and $L$. In contrast, the run-time of our $\main$ algorithm takes advantage of improvements that become significant for small $p$. The $\main$ algorithm achieves this by calling the $\child$ subroutine multiple times, with different choices of $s$ and $L$. It starts with $s=s_0$ and $L=L_0$ so small that they cannot possibly satisfy the desired accuracy requirement. Then, at each step it chooses larger $s_k,L_k$ that lead to estimates $\hat{p}_k$, which are used to learn upper bounds on $p$ that decrease with each iteration. These, in turn, allow one to estimate sharper values of $s$ and $L$ to achieve the desired accuracy.

The run-time of $\main$, $\tau_{\main}$, has two distinct components we call the \emph{circuit-sensitive} and the \emph{circuit-insensitive} components. The circuit-sensitive component of $\tau_{\main}$ is associated with the total run-time over all calls to the $\child$ subroutine. The run-time of the $\child$ subroutine will approximately double in each subsequent call with the run-time of each round and the total number of rounds depending on circuit parameters (such as $t$) and accuracy parameters (such as $\epsilon_{\rm tot}$). Typically, this component constitutes the overwhelming majority of~$\tau_{\main}$. The circuit-insensitive component of~$\tau_{\main}$ arises from various numerical optimizations that are executed in each step of the $\main$ algorithm, e.g. to determine the choice of $s_k, L_k$ for each step $k$. The run-time of each such step is of order $\sim 1$ second (for a standard desktop computer) and it is insensitive to the various parameters that define the Born rule probability estimation task. The total number of steps is also small with more than $\sim 50$ steps being infeasible due to the exponential growth of the run-time of $\child$ in the step number $k$. For this reason, we treat the circuit-insensitive component of $\tau_{\main}$ as a fixed run-time cost.

Consistent with Eq.~\eqref{eq:child run-time}, we model the run-time of $\child$ as:
\begin{align}\label{eq:model child run-time}
    \modeltau(s,L):=c_1 s t^3 + c_2 s L r^3,
\end{align}
where $c_1, c_2$ are hardware specific positive constants (in units of seconds per elementary operation) that can be used to model the actual run-time of $\child$. The $\parent$ algorithm aims to minimize the quantity:
\begin{align}\label{eq:child cost}
	\mathcal{C}:=\sum_{k\in [\parentrounds]} \modeltau(s_k,L_k),
\end{align}
where $\parentrounds$ is the total number of times the $\child$ algorithm will be called and $s_k$, $L_k$ indicate the input parameters used on the $k^{th}$ call. We call $\mathcal{C}$ the \emph{run-time cost}; it represents our modelled circuit-sensitive component of the run-time of $\main$.

The run-time cost $\mathcal{C}$ is probabilistic and depends on the unknown $p$. Our $\mainruntime$ algorithm efficiently computes a probabilistic upper bound of $\mathcal{C}$ for any assumed $p$. This may be useful for informing expected run-times particularly when prior information about $p$ is known. Our $\main$ and $\mainruntime$ algorithms, together with related details, can be found in Appendix~\ref{app:parent}. We note that our $\main$ algorithm allows the user to fix the accuracy parameters, $\epsilon_{\rm tot}$ and $\delta_{\rm tot}$, for the price of moving their dependence on $p$ to $\tau_{\main}$.

% -------------------------------------------------------------
% SECTION III - DISCUSSION
% -------------------------------------------------------------

\section{Discussion of the performance of our algorithms}
\label{sec:performance}

In this section we first review the existing simulation algorithms, and then compare our results with them, demonstrating that our suite of algorithms offers state-of-the-art performance in Born rule probability estimation across a broad range of parameter regimes.

\subsection{Related research}
\label{sec:related}

Brute force simulation algorithms such as Schr\"{o}dinger-style~\cite{fatima2020faster}, Feynman-style~\cite{de2019massively, markov2008simulating, de2007massively} or hybrid simulators~\cite{markov2018quantum} offer high precision general purpose classical simulation capabilities for universal quantum circuits. However, such simulations can be extremely resource intensive for moderate circuit width (number of qubits $n \approx 40$) and/or depth. Alternatively, there exist efficiently classically simulable families of (non-universal) quantum circuits~\cite{gottesman1998heisenberg,bartlett2002efficient,terhal2002classical,aaronson2004improved,Jozsa2008matchgates}.  In particular, the Gottesman-Knill theorem makes it possible to classically simulate thousands of qubits with hundreds of thousands of gates provided that we restrict to so-called stabilizer circuits~\cite{gottesman1998heisenberg}.

Between these two extremes, Aaronson and Gottesman~\cite{aaronson2004improved} were the first to present  a classical simulation algorithm that is efficient for stabilizer circuits but can also simulate non-stabilizer circuits with a run-time cost that is exponential in the number of non-stabilizer gates (non-Clifford gates). A limitation of this work is that the run-time does not depend on the specifics of the additional non-stabilizer gates. Thus, their simulator pays a heavy run-time penalty for introducing a small number of non-stabilizer gates even if these are arbitrarily close to stabilizer gates.  Research to overcome this limitation falls into two broad categories:  Born rule probability estimators based on using a quasi-probabilistic representation of the density matrix~\cite{Rall2019, pashayan2017sampling, pashayanthesis,howard2017application,seddon2020quantifying,pashayan2015estimating, veitch2012negative, mari2012positive}, and pure-state sampling simulators~\cite{garcia2012efficient, garcia2014geometry, bravyi2016trading,bravyi2016improved}.  

In the pure state formalism, a number of works~\cite{garcia2012efficient, garcia2014geometry, bravyi2016trading} have culminated in two important simulation algorithms by Bravyi and Gosset (BG)~\cite{bravyi2016improved,bravyi2019simulation}. The first of these, which we refer to as the \textit{BG-estimation algorithm}, produces multiplicative precision estimates of Born rule probabilities. The second of these, which we refer to as the \textit{BG-sampling algorithm}, approximately samples from the outcome distribution of the quantum circuit. These algorithms exactly or approximately represent the initial quantum state by a linear combination of stabilizer states. The efficiently simulable circuits consist of initial states that are a superposition of at most polynomially many stabilizer states, together with Clifford gates and computational basis measurements. These circuits can be promoted to universality by allowing initial states to include many copies of a magic state. The run-times of the BG-estimation and BG-sampling algorithms depend linearly on the exact and approximate stabilizer rank of the initial quantum state respectively. Roughly, the exact (or approximate) stabilizer rank of a quantum state is the minimal number, $\chi$, of stabilizer states required such that this state can exactly (or approximately) be written as a linear combination of $\chi$ stabilizer states. Both algorithms have run-times that scale linearly in their respective stabilizer ranks and efficiently in all other circuit parameters, although some of the polynomial dependencies are nevertheless significant and can be prohibitive. Both the exact and approximate stabilizer ranks are computationally hard to compute even for product states although upper bounds exist for some important examples. Ref.~\cite{bravyi2019simulation} introduced a computationally better-behaved quantity $\xi$, called the stabilizer extent, and showed that the approximate stabilizer rank of an initial state $\ket{\psi}$ can be upper bounded by $\xi(\ket{\psi})/\epsilon^2$, where $\epsilon$ quantifies the degree of error in the approximation of $\ket{\psi}$. Ref.~\cite{bravyi2019simulation} also presented the \emph{sum over Cliffords} sampling algorithm: a new variant of the BG-sampling algorithm where non-Clifford gates are directly simulated by expressing them as a linear combination of Clifford gates. To compare this to our work, we consider the application of this technique to diagonal single qubit non-Clifford gates $\Tgate{\phi}$ inducing a $Z-$rotation of angle $\phi\in (0,\pi/4)$. For circuits composed of exactly $t$ uses of $\Tgate{\phi}$, the run-time of the sum over unitaries algorithm scales linearly in the stabilizer extent of the state \mbox{$\Tkets$}.

In the density matrix formalism, algorithms known as quasi-probabilistic simulators~\cite{pashayan2015estimating, howard2017application, pashayanthesis, seddon2020quantifying} produce additive precision estimates of Born rule probabilities. These algorithms represent the quantum density matrix as a linear combination of a preferred set of operators known as a \emph{frame}~\cite{pashayan2015estimating, Ferrie_2008}. Many frame choices have been considered including Weyl-Heisenberg displacement operators~\cite{Rall2019, pashayan2017sampling, pashayanthesis}, frames constructed from stabilizer states~\cite{howard2017application, seddon2020quantifying} and phase-point operators~\cite{pashayan2015estimating, veitch2012negative, mari2012positive} used in the construction of the discrete Wigner function~\cite{GrossHudson, Gibbons2004Wootters}. Particularly relevant to our work is the dyadic frame simulator of Seddon \textit{et al.}~\cite{seddon2020quantifying}.  In this simulator, density matrices are decomposed into a linear combination of stabilizer dyads: operators of the form $\ket{L}\!\!\bra{R}$ where $\ket{L}$ and $\ket{R}$ are pure stabilizer states. The efficiently simulable circuits consist of initial states that are tensor products of convex combination of stabilizer dyads, stabilizer preserving operations including Clifford gates and computational basis measurements. These circuits can be promoted to universality by allowing initial states to include many copies of a magic state: states that are not a convex combination of stabilizer dyads and can be used to teleport non-Clifford gates into the circuit. The degree to which the initial state's optimal linear decomposition into stabilizer dyads departs from a convex combination is quantified by the \emph{dyadic negativity}. The run-time of the dyadic frame simulators depends quadratically on the dyadic negativity. The dyadic negativity can in general be exponentially large and is the only source of run-time inefficiency. Nevertheless, in contrast to the Aaronson and Gottesman simulator, the dyadic frame simulator's run-time will be responsive to the level of deviation from the efficiently simulable operations. The dyadic frame simulator of Ref.~\cite{seddon2020quantifying} is the current state-of-the-art quasi-probabilistic simulator for simulating stabilizer circuits promoted to universality via magic state injection.

The mixed-state stabilizer rank simulator of Ref.~\cite{seddon2020quantifying} made further improvements to the BG-sampling algorithm by improving the run-time dependence on the error tolerance for the approximate sampling task and by generalizing the algorithm to the setting where initial states can be mixed states. The run-time of this improved algorithm scales linearly in a quantity known as the mixed state extent~\cite{seddon2020quantifying}. Ref.~\cite{seddon2020quantifying} also showed that for any $n-$qubit product states, its dyadic negativity, stabilizer extent and mixed state extent are all equal. This result allows one to compare performance across multiple simulation algorithms in the practically relevant setting where initial states are product states.

\subsection{Performance improvements}

As compared with the related BG-estimation algorithm~\cite{bravyi2016improved}, our $\compute$ algorithm exhibits \newtext{three} obvious benefits. \newtext{First, our $\qcompress$ algorithm, can significantly reduce the complexity of the circuit to be simulated}. Second, our algorithm is exact, while the one of Ref.~\cite{bravyi2016improved} runs with a failure probability $\delta$ and relative error $\rele$, and to improve these precision parameters one has to pay the price of longer run-times. Specifically, the run-time of that algorithm is given by $\order{2^{\beta t} t^3 \rele^{-2} \log(\delta^{-1})}$, where $\beta=(1/6)\log_2 7 \approx 0.47$. Comparing this with $\tau_{\compute}$, we see that the performance of our algorithm is better in certain parameter regimes when $(t-r)\leq \beta t$. As discussed above, this happens generically for random circuits when \mbox{$(1-\beta)t \leq n-w$}. \newtext{Since $\compute$ produces results that are exact (to machine precision), it is straightforward to employ it to compute expectation values of operators expressed as sums of Pauli operators.}

The discussion of the performance of $\main$ will be divided into three parts. First, we will discuss the crucial $\child$ subroutine and point out the run-time improvements over the existing Born rule estimation algorithms. Second, we will explain additional run-time improvements that arise from the $\main$ algorithm itself, i.e., from the adaptive choice of optimal input parameters $s$ and $L$ for the $\child$ subroutine. Finally, we will justify why we expect the total run-time of $\main$ to be closely related (to within 1-2 orders of magnitude) to the run-time of $\child$ with the optimal choice of parameters.

To analyse the performance of the $\child$ subroutine, we start by employing Eq.~\eqref{eq:thm3} to note that for arbitrary $\epsilon\in (0,\epsilon_{\mathrm{tot}})$ and $\delta\in (0,\delta_{\mathrm{tot}})$ the choice of parameters $s$ and $L$ satisfying
\begin{equation}\label{eq:optimal s and L}
    \begin{aligned}
	s&\geq \frac{2 (\sqrt{\ubex} +\sqrt{p})^2}{{\left(\sqrt{p+\epsilon}-\sqrt{p}\right)}^{2}} \log\left(\frac{2 e^2}{\delta}\right),\\ %
	L&\geq \brackn{\frac{p+\epsilon}{\epsilon_{\mathrm{tot}}-\epsilon}}^2 \log\left(\frac{1}{\delta_{\mathrm{tot}}-\delta}\right),
	\end{aligned}
\end{equation}
guarantees an estimate $\hat{p}$ with error smaller than $\epsilon_{\mathrm{tot}}$ and failure probability smaller than $\delta_{\mathrm{tot}}$. The meaningful parameter regime is given by $\epsilon_{\mathrm{tot}}\ll p$ (estimation error should be smaller than the estimated value) and $\ubex\gg 1$ (we want to simulate non-Clifford circuits, as Clifford ones are already efficiently simulable). Then, the two terms of the run-time $\tau_{\child}$ characterized by Eq.~\eqref{eq:child run-time} scale as
\begin{equation}
\begin{aligned}
\tau_\child^{(1)}&=\orderlog{\ubex t^3p\epsilon_{\mathrm{tot}}^{-2}},\\ \tau_\child^{(2)}&=\orderlog{\ubex r^3p^3\epsilon_{\mathrm{tot}}^{-4}},   
\end{aligned}
\end{equation}
where $\tilde{O}$ notation hides the logarithmic dependence on the failure probability $\delta_{\mathrm{tot}}$. Importantly, note that the relative error $\rele_{\mathrm{tot}}$ introduced by the additive error $\epsilon_{\mathrm{tot}}$ is given by $\rele_{\mathrm{tot}}=\epsilon_{\mathrm{tot}}/p$. Thus, the run-time only weakly depends on the additive error as $\order{\epsilon_{\mathrm{tot}}^{-1}}$ for both terms, with the remaining scaling dependent on the relative error as $\order{\rele_{\mathrm{tot}}^{-1}}$ and $\order{\rele_{\mathrm{tot}}^{-3}}$, respectively.

We first compare the performance of $\child$ with the results of Ref.~\cite{bravyi2019simulation}, where the authors provide a subroutine approximating Born rule probabilities to additive polynomial precision. It is based on the approximate stabiliser decomposition of magic states and on a novel fast norm estimation subroutine. First, one computes $k$-rank stabiliser decomposition taking $\order{kt^3}$ steps. The crucial Theorem~1 of Ref.~\cite{bravyi2019simulation} proves that by choosing $k\approx \cextent/\epsilon_1^2$, the additive error introduced in this step will be bounded by $\epsilon_1$. Next, one uses the fast norm estimation with a failure probability $\delta_{\mathrm{tot}}$ and a relative error $\rele_2$, which takes $\orderlog{kt^3\rele_2^{-2}}$ steps, and the run-time of this step dominates the total run-time. Note that the worst case total additive error $\epsilon_{\mathrm{tot}}$ can be lower-bounded by $\epsilon_1+p\rele_2$. Thus, the term $\order{\epsilon_1^{-2}\rele_2^{-2}}$ can be optimally replaced by $\order{p^2\epsilon^{-4}_{\mathrm{tot}}}$. Taking this into account, one gets that the total run-time is $\orderlog{\cextent t^3p^2\epsilon_{\mathrm{tot}}^{-4}}$. Combining the described algorithm with its variation, \emph{the sum over Cliffords method}~\cite{bravyi2019simulation}, one gets the run-time scaling as $\orderlog{\cextent \min \set{n^3,t^3}p^2\epsilon_{\mathrm{tot}}^{-4}}$. Comparing this with $\tau_\child^{(1)}$ and $\tau_\child^{(2)}$ (and noting that $r\leq t$, $r\leq n$, $p\leq 1$ and $\epsilon_{\mathrm{tot}}^{2}/p\leq 1$), we see that $\child$ compares favourably in almost all regimes. More precisely, there is a performance advantage scaling as $\orderlog{p\epsilon_{\mathrm{tot}}^{-2}\,\min\{1,(n/t)^3\}}$ and $\orderlog{p^{-1}\,\min\{(n/r)^3,(t/r)^3\}}$ for the two components of the run-time. 

\begin{figure*}[ht]
\centering
	\subfloat{\includegraphics[width=0.45\columnwidth]{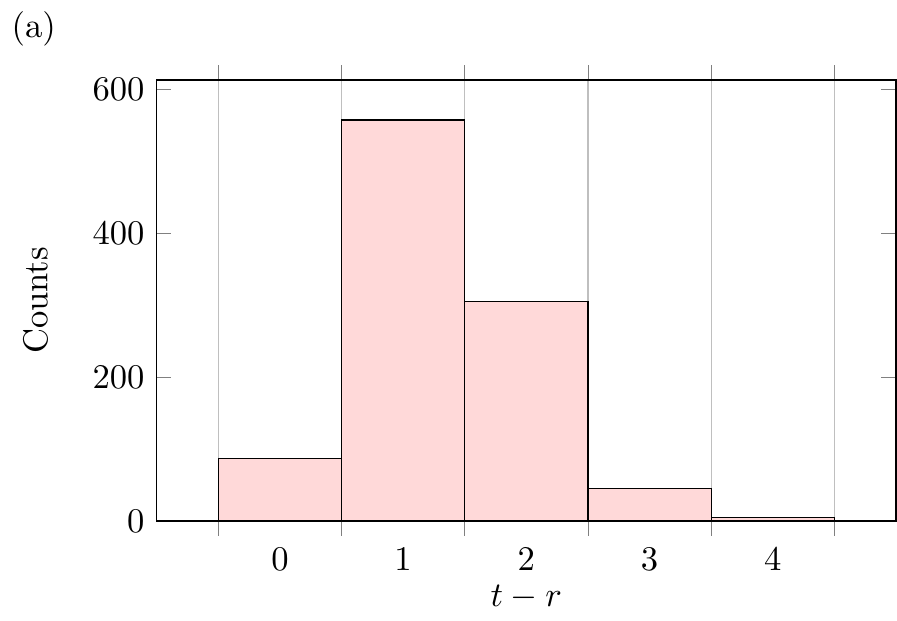}}\hspace{1cm}
	\subfloat{\includegraphics[width=0.45\columnwidth]{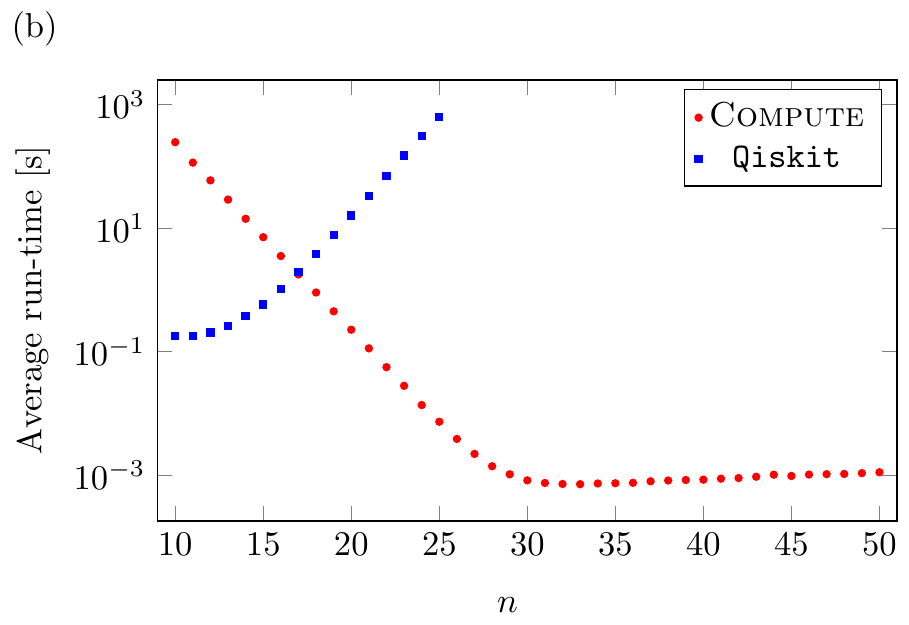}}
	\caption{\label{fig:compute-perf}\textbf{Performance of the} $\compute$ \textbf{algorithm for random circuits.} Random circuits are generated as follows: we generate $c$ Clifford gates acting on random qubits (equal probability of choosing $S$, $H$, $CX$ and $CZ$), we then replace a size $t$ random selection of these with $T$ gates. (a) The distribution of $(t-r)$ for $10^3$ random circuits with $n=100$ qubits, $c=10^5$ Clifford gates, $t=80$ $T$~gates and $w=20$ measured qubits. (b)~Average run-times for calculating the Born rule probability for random circuits with $n$ qubits, $c=10^3$ Clifford gates, $t=30$ $T$~gates and $w=10$ measured qubits, with the average taken over $10^2$ random circuits for each $n$. The red circles correspond to our $\compute$ algorithm (including the run-time needed to run $\qcompress$), while the blue squares correspond to classical state vector simulation framework of IBM's quantum programming suite \texttt{Qiskit}~\cite{aleksandrowicz2019qiskit}. Simulations were performed using a single core of a standard desktop computer.}
\end{figure*}

Next, we compare the performance of $\child$ with the results of Ref.~\cite{seddon2020quantifying}. We start by noting that the \emph{mixed-state stabilizer rank} simulator of Ref.~\cite{seddon2020quantifying} improved the run-time by a factor of up to $\epsilon_{\rm tot}^{-1}$ as compared to the sampling based simulation of Ref.~\cite{bravyi2019simulation}. This should be contrasted with our improvement factors of $p^{-1}$ and $\epsilon_{\rm tot}^{-2}/p$, and so, depending on the regime, the mixed-state stabilizer rank simulator could be better or worse than $\child$. However, it should be noted that the improvement in Ref.~\cite{seddon2020quantifying} applies specifically to the task of approximately sampling from the outcome distribution of a quantum circuit. Therefore, it is unclear how to attain such an improvement directly for the task of Born probability estimation (we note that one can attain Born probability estimates by using $\order{\epsilon_{\rm tot}^{-2}}$ samples but this invalidates the run-time advantage). Reference~\cite{seddon2020quantifying} also presents the \emph{dyadic frame simulator}. It performs exactly the same task as $\main$, i.e., it estimates a single Born rule probability with an additive error $\epsilon_{\mathrm{tot}}$, and we note that the dyadic frame simulator is more generally applicable as it is also suitable for mixed states. Ignoring the polynomial and logarithmic pre-factors, its dominant run-time scales as $\orderlog{\xi^{*2} \epsilon_{\mathrm{tot}}^{-2}}$. Therefore, we see that our $\child$ algorithm compares favorably, as it has a run-time advantage of $\cextent$ that is exponential in the number $t$ of non-Clifford gates.

We now discuss the second source of performance advantage that arises from the adaptive nature of the $\main$ algorithm. In order to produce a meaningful estimate, we require guarantees on its error $\epsilon_{\mathrm{tot}}$ and failure probability $\delta_{\mathrm{tot}}$. We note that neither $\child$ nor any of the above mentioned competing algorithms have such an accuracy guarantee, as in order to choose proper simulation parameters (like our $s$ and $L$), achieving given $\epsilon_{\mathrm{tot}}$ and $\delta_{\mathrm{tot}}$, one would need to know the unknown value of $p$. Thus, one is left to make a conservative choice of $p=1$ that kills any run-time advantage coming from the polynomial dependence on $p$. In contrast, our $\main$ algorithm is able to take advantage of this $p$ dependence. As a result, the run-time improvements related to the estimated probability and its error effectively scale as $\orderlog{p^{-1}\epsilon_{\mathrm{tot}}^{-2}}$ and $\orderlog{p^{-3}}$ (rather than the above-mentioned $\orderlog{p\epsilon_{\mathrm{tot}}^{-2}}$ and $\orderlog{p^{-1}}$). The run-time price of using $\main$, as compared to $\child$ with optimally chosen parameters $s$ and $L$, is a small circuit-insensitive overhead related to parameter optimisation, and an additional circuit-sensitive overhead arising from the fact that we make multiple calls to $\child$. The former one is so small that can be ignored, while we explain how to effectively upper-bound the latter one below. To conclude, $\main$ exhibits the following run-time improvements as compared to the run-time $\tau_{\rm previous}$ of the two methods of Ref.~\cite{bravyi2019simulation}:
\begin{equation}
\begin{aligned}
\frac{\tau_{\mathrm{previous}}}{\tau_\child^{(1)}}&=\orderlog{p^{-1}\epsilon_{\mathrm{tot}}^{-2}\,\min\{1,(n/t)^3\}},\\ \frac{\tau_{\mathrm{previous}}}{\tau_\child^{(2)}}&=\orderlog{p^{-3}\,\min\{(n/r)^3,(t/r)^3\}}.
\end{aligned}
\end{equation}

Finally, we explain why we expect that $\modeltau(s^*,L^*)$, with $(s^*,L^*)$ being the choice of parameters $s$ and $L$ optimized with respect to the unknown $p$, can act as a proxy for $\tau_\estimate$ in the regime where $p\geq \epsilon_{\mathrm{tot}}$. The $\main$ algorithm runs the $\child$ subroutine $K$ times. At each step $k$, the parameters $s_k$ and $L_k$ are chosen optimally with respect to $p_k^{\mathrm{UB}}$, an upper bound for $p$. It can be shown that in the final step, $p_K^{\mathrm{UB}}\leq p+2\epsilon_{\mathrm{tot}}$. Thus, in the regime where $p\geq \epsilon_{\mathrm{tot}}$, the final optimization is with respect to $p_K^{\mathrm{UB}}=O(p)$ with $\modeltau(s^*,L^*)$ having a cubic dependence on $p$. An additional source of discrepancy arises since the final step's optimisation uses a failure probability of $\delta_K=\frac{6}{\pi^2 K^2}\delta_{\mathrm{tot}}$ in contrast to $\delta_{\mathrm{tot}}$ used in determining $s^*$ and $L^*$. However, due to $\modeltau(s^*,L^*)$ having only a poly-logarithmic dependence on $\delta_{\mathrm{tot}}$, this also contributes a small run-time overhead to the final step's call to $\child$. Since the final call's cost is approximately half of the total run-time cost, we conclude that run-time of $\main$ should be close to the run-time of $\tau_\child$ when $p\geq \epsilon_{\mathrm{tot}}$. \newtext{As we will shortly see,} these expectations are indeed confirmed by our numerical analysis.

\subsection{Performance on random circuits}
\label{sec:performance-random}

\begin{figure*}[ht]
\subfloat{\includegraphics[width=0.45\columnwidth]{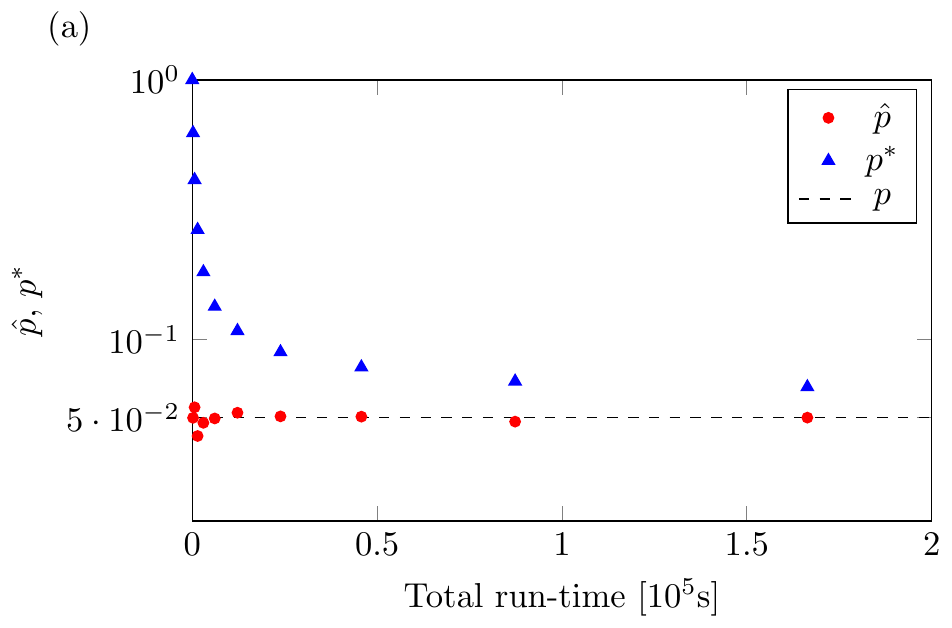}}
\hspace{1cm}
\subfloat{\includegraphics[width=0.45\columnwidth]{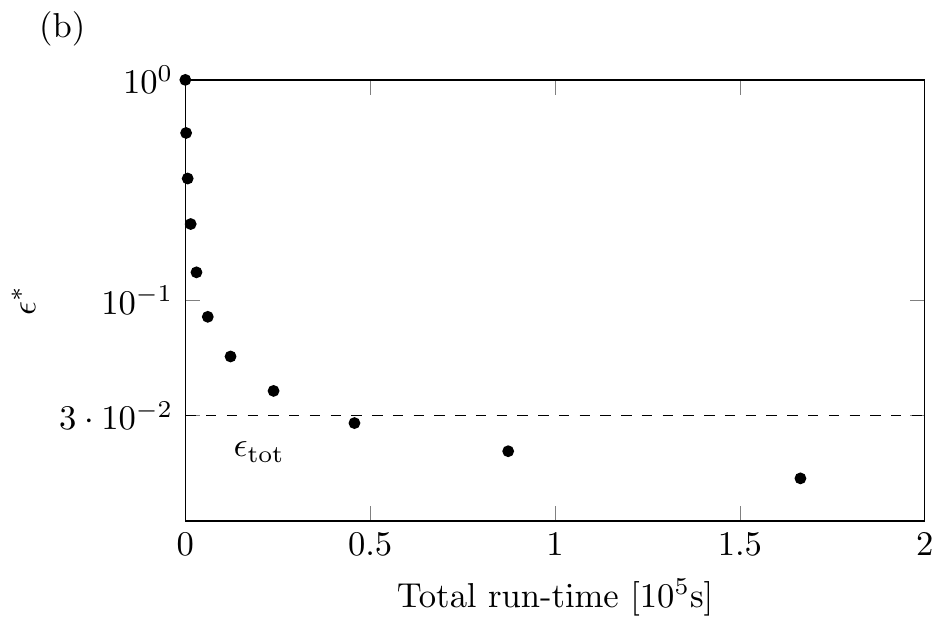}}	
	\caption{\label{fig:optimize-perf}\textbf{Performance of the} $\parent$ \textbf{algorithm.} An $n=50$ qubit, $t=60$ non-Clifford gate circuit of the form \mbox{$U U^\dagger V(p)$} as described in Sec.~\ref{sec:performance-random}. The unitary $U$ is randomly constructed as described in the caption of Fig.~\ref{fig:compute-perf}, and consists of $1000$ gates of which $26$ are non-Clifford $T_\theta$ gates. The unitary $V(p)$ acts non-trivially on the first $w=8$ qubits, which are then measured in the computational basis, leading to the the probability of the all-zero outcome $p=0.05$. The parameter $\theta$ is chosen such that the total circuit has stabiliser extent $\cextent\approx 3767$, equivalent to $52$ $T$ gates. For this circuit the value of the \rname~is $r=10$. The total run-time of approximately $1.6\times 10^5\mathrm{s}$ includes approximately $34\mathrm{s}$ of fixed overhead from the $\parent$ algorithm. (a)~The estimate $\hat{p}$ (red circles) and its upper bound $p^*$ (blue triangles) as a function of the total run-time. The dashed line indicates the chosen value of $p$. (b)~The upper-bound $\epsilon^*$ for the total estimation error as a function of the total run-time. The dashed line indicates the target error $\epsilon_{\mathrm{tot}}=0.03$, and a failure probability of $\delta_{\mathrm{tot}}=10^{-3}$ was used. The target error was obtained within $4.7\times 10^4 \mathrm{s}$}
\end{figure*}

The run-time of the $\compute$ algorithm depends exponentially on $(t-r)$, so its performance depends crucially on the value of $r$. In the case of random circuits where many Clifford gates are interleaved between each non-Clifford gate, our numerical investigations show that $r$ very strongly concentrates around the maximum allowed value of $\min\{t,n-w\}$, see Fig.~\hyperref[fig:compute-perf]{2a} for details. Thus, in certain parameter regimes, e.g., when $(n-w)\geq t$, the $\compute$ algorithm has a very quick run-time. In Fig.~\hyperref[fig:compute-perf]{2b} we present the comparison of the run-times between our $\compute$ algorithm and the IBM's \texttt{Qiskit} state vector simulator~\cite{aleksandrowicz2019qiskit}. While the run-times for the latter algorithm become infeasible on a standard desktop computer for $n>35$ (due to memory limitations), our algorithm can, within feasible run-times, compute the Born rule probabilities as long as the number of non-Clifford gates $t$ is not significantly larger than $(n-w)$. Thus, for random circuits it is not the total number of non-Clifford gates that makes our simulation infeasible, but rather the number of non-Clifford gates in excess of the number of unmeasured qubits. To illustrate this, we employed $\compute$ to obtain the Born rule probability of $1000$ random circuits with $n=55$, $w=5$, $c=10^5$ and $t=80$, the mean run-time was $459$ seconds with a maximum of $555$ seconds. For every one of these circuits, $r$ was found to take its maximal value $r=50$. Hence our $\compute$ run-time for such circuits is typical.

\begin{figure*}
\centering
\includegraphics[width=0.8\columnwidth]{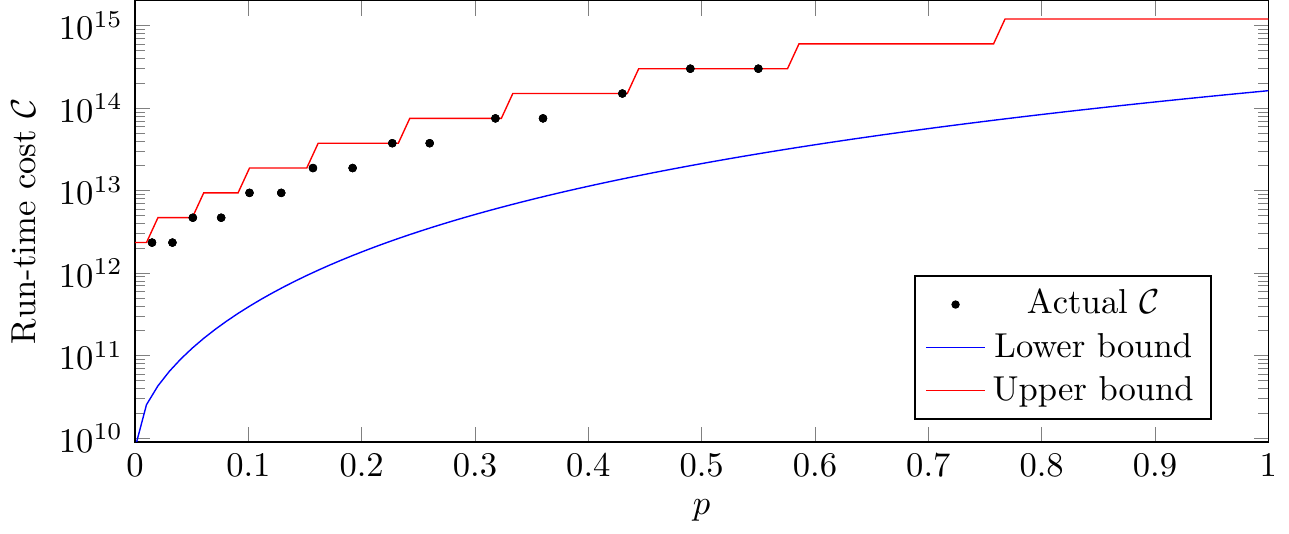}
\caption{\label{fig:run-time-bounds}\textbf{Estimate run-time cost and its bounds.} The run-time cost of the $\main$ algorithm with $\epsilon_{\rm tot}= 0.05$, $\delta_{\rm tot} = 10^{-3}$ for random circuits $UU^\dagger V(p)$. The black dots represent the total run-time cost $\mathcal{C}$, as defined in Eqs.~\eqref{eq:model child run-time}-\eqref{eq:child cost} using $c_1=c_2=1$. The circuits were acting on $n=40$ qubits and were composed of $t=40$ non-Clifford gates with the total stabilizer extent $\cextent\approx159$ (equivalent to $32$ $T$ gates) and $w=8$ measured qubits. For these circuits the value of the \rname~is $r=8$. Each black dot is in fact a cluster of $3$ independent $\main$ simulations that produce $\mathcal{C}$ values that are too close to resolve on this plot. All final Born rule probability estimates produced by $\main$ were within an additive error $0.2\epsilon_{\mathrm{tot}}$ of $p$. The top red line indicates the probabilistic upper bound (with failure probability less than $\delta_{\rm UB}=0.05$) for the total run-time cost $\mathcal{C}$ obtained with the efficient $\mainruntime$ algorithm. The bottom blue line indicates the lower bound on $\mathcal{C}$ obtained from $\tau_\child$ with the choice of parameters $s$ and $L$ being optimized using knowledge of the value of $p$. The run-time cost of the first data point with $p=0.015$ corresponds to actual computational time of $7$ minutes and the last with $p=0.55$ required approximately $10$ hours. }
\end{figure*}

\newtext{To numerically support our analysis of the performance of the $\main$ algorithm}, we use quantum circuits of the form $UU^\dagger V(p)$. Here $U$ is a random non-Clifford circuit composed of Clifford and $T_\theta$ gates, and $V(p)$ is a non-Clifford circuit that acts non-trivially on the first $w$ measured qubits as:
\begin{equation}
 \!\!\! {\begin{matrix}\includegraphics[clip,trim=0cm 0cm 0cm 0cm,width=0.2\linewidth]{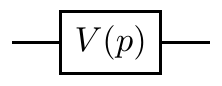}\end{matrix}}
  \!=\!\bigg(\!\!\begin{matrix}\includegraphics[clip,trim=0cm 0cm 0cm 0cm,width=0.4\linewidth]{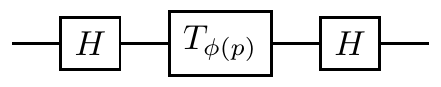}\end{matrix}\!\!\bigg)^{\otimes w}\!\!\!\otimes I^{\otimes(n-w)},
\end{equation}
\newtext{for a given $p$ the corresponding phase is given by
\begin{equation}
    \phi(p) = 2 \arccos p^{\frac{1}{2 w }}.
\end{equation}}
In this way we are able to generate random non-Clifford circuits with a chosen probability $p\in[0,1]$ of the all zero outcome controlled by the choice of parameter $\phi(p)$, and a stabilizer extent $\cextent$ that is made independent of $p$ by controlling $\theta$. Via this construction, we can numerically verify the actual run-time dependence of $\main$ on $p$ for a family of random circuits. The performance of $\main$ algorithm for $UU^\dagger V(p)$ circuits is illustrated in Fig.~\ref{fig:optimize-perf}.

\newtext{Finally,} using the $\mainruntime$ algorithm for random circuits of the described form $UU^\dagger V(p)$, we found the upper-bound of run-time cost $\mathcal{C}$ of the $\main$ algorithm as a function of $p$. We have also lower-bounded this cost by the run-time cost of $\child$ with the optimal choice of $s$ and $L$ (as we know the value of $p$ this can be easily done using Theorem~\ref{thm:child alg}). We present both bounds in Fig.~\ref{fig:run-time-bounds}, where it is clear that they differ by less than 2 orders of magnitude when $p\geq \epsilon_{\mathrm{tot}}$. To further strengthen our point, we have also run the $\main$ algorithm on circuits $UU^\dagger V(p)$ for a few chosen values of $p$, and also plotted the actual run-time costs in Fig.~\ref{fig:run-time-bounds}. This shows that, provided $p\geq \epsilon_{\mathrm{tot}}$, $\mathcal{C}$ is indeed close to the run-time cost of $\child$ with the choice of $s$ and $L$ being optimised using knowledge of the value of $p$.

\subsection{Performance on existing benchmarks}
\label{sec:performance-existing-bounds}

Numerous benchmarks have been used to assess the performance of classical simulators, see for example Refs.~\cite{garcia2014simulation, de2019massively,google2019supremacy,villalonga2019flexible}. To directly compare our algorithms with the prior state-of-the-art for Clifford+T simulation, we adopt the two benchmarks used in Refs.~\cite{bravyi2016improved} and~\cite{bravyi2019simulation}. The first of these is a simulation of an algorithm to solve a task known as the hidden-shift problem, introduced in Ref.~\cite{Rotteler2010hiddenshift}; while the second is an implementation of the quantum approximate optimization algorithm, (QAOA), developed by Farhi et al.~\cite{farhi2014-qaoa}. We apply the QAOA implementation to to solve a problem known as Max-E3LIN2 with bounded degree. Performance comparisons on those two benchmarks between our algorithms and those of Refs.~\cite{bravyi2016improved} and~\cite{bravyi2019simulation} are summarized in Table~\ref{tab:hidden-shift-comparison} (for the hidden shift problem) and Table~\ref{tab:qaoa-comparison} (for the QAOA).

\subsubsection{Hidden shift}

\begin{figure*}[hbtp]
\centering
	\sidesubfloat[a]{\includegraphics[width=0.42\columnwidth]{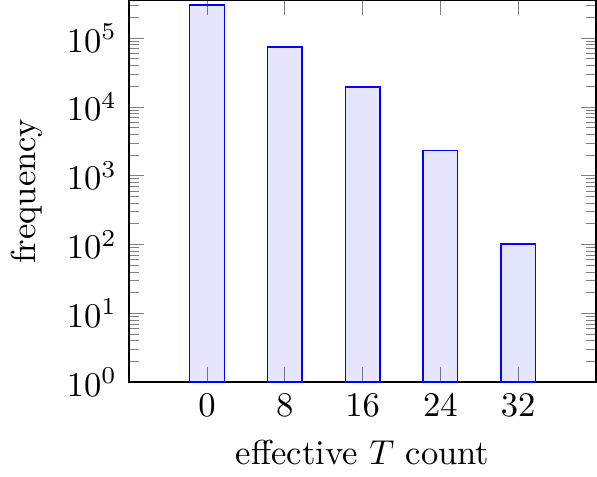}\label{fig:compress-hidden-shift-data}}\hspace{1cm}
	\sidesubfloat[b]{\includegraphics[width=0.42\columnwidth]{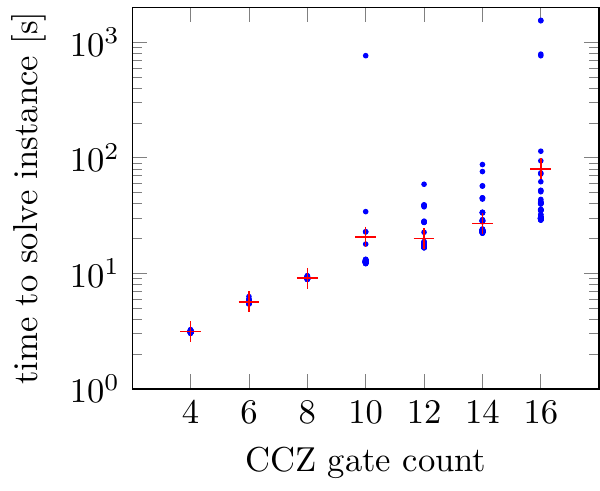}\label{fig:compute-hidden-shift-data}}

	\caption{\label{fig:hidden-shift-2} \newtext{\textbf{Performance of our algorithms for the hidden shift circuits.} On the left (a), we show the effective $T$-count $t^\prime$ for $10,000$ hidden-shift circuits each consisting of $40$ qubits, $8$ CCZ gates and $200$ diagonal Clifford gates per CCZ. Since we run $\compute$ once for each qubit, for each circuit there are $400,000$ data-points summarised in the histogram. The bar at $t=0$ consists of $304,307$ examples of which $200,000$ correspond to the measurements on the first $20$ qubits of each circuit, which are always deterministically solved by $\qcompress$. The timing data (b) refers to simulations performed using a single thread (no parallelism) on an desktop computer with an AMD Ryzen 9 3900X 12-Core Processor operating at approximately 2.2~GHz. For each indicated CCZ gate count we solve $100$ instances of the hidden-shift problem and plot the run-time for our algorithms (blue dots) and the means (red crosses). The longest run-time of $1,547$ seconds at $16$ CCZ gates should be compared to the run-time of approximately $10^5$ seconds for the algorithm of Refs.~\cite{bravyi2016improved} and \cite{bravyi2019simulation} for the same parameters. Each blue dot represents the time to solve a full hidden-shift instance, requiring 40 calls to our simulation algorithms. }}
\end{figure*}

\begin{table*}[hbtp]
\centering
\newtext{
\begin{tabular}{c|cccc}
     & qubit count & CCZ-count & $T$ gate count & runtime\\ \hline
     Ref.~\cite{bravyi2016improved} & 40 & 12 & 48 & ``several hours"\\
     Ref.~\cite{bravyi2019simulation} &  40 & 16 & n/a$^*$&roughly $10^5$s\\
     Compress+Compute & 40& 16 & 112 & min $29$s, mean $80$s, max $1,547$s \\
\end{tabular}
}
\caption{\label{tab:hidden-shift-comparison}\newtext{Direct comparison of our hidden shift results and those of prior works. In each paper referenced the authors provide a run-time for solving a single instance of the hidden shift problem. Our mean and max numbers are over the same sample of $100$ randomly generated bent functions that is plotted at $16$ CCZ gates in figure~\ref{fig:compute-hidden-shift-data}. Each instance was solved using a single thread (i.e. with no parallelization) on an AMD Ryzen 9 3900X 12-Core Processor operating at approximately 2.2 GHz. The n/a marked $^*$ occurs because authors employ direct sum-over-Cliffords decomposition of CCZ gates without first decomposing into $\Tgate{}$ gates. It should be noted that while the methods of Ref.~\cite{bravyi2019simulation} provide an approximate answer with some error probability, $\compute$ deterministically provides the exact answer.}}
\end{table*}

The circuits we use exactly match the benchmarks performed in Ref.~\cite{bravyi2019simulation}, with $n=40$ qubits, CCZ-counts between $4$ and $16$, and $200$ diagonal Clifford gates per CCZ. The circuits are constructed as described in Ref.~\cite{bravyi2016improved}, with the exception that we use a $7$ $T$ gate decomposition of the non-Clifford CCZ gate, whereas they use a decomposition employing only $4$ $T$ gates (however, for the price of complicating the circuit with additional ancillary qubits and intermediate measurements). Naively, the additional $3$ $T$ gates per CCZ should increase the cost of the $\compute$ algorithm by a factor of $2^{3k}$, where $k$ is the CCZ-count. Interestingly, this is not the case, as the $\qcompress$ algorithm brings the effective $T$-count, $t^\prime$, of the circuits below $4k$ in every case we have examined, including the $400,000$ data points at $k=8$ summarised in Fig.~\ref{fig:compress-hidden-shift-data}. 
In these data, we observe an average effective $T$-count of $2.4$ $T$ gates per circuit; down from $7k=56$ $T$ gates in the original circuit inputs to $\qcompress$. Since run-time is exponential in the effective $T$-count, our dramatic improvement on the run-time of Refs.~\cite{bravyi2016improved} and~\cite{bravyi2019simulation} is primarily due to the $T$-count reduction achieved by $\qcompress$. See Supplemental Material Sec.~\ref{sec:supp-hidden-shift-details}. for more details on the effective $T$-count.

Additionally, we observed that the effective $T$-count is a multiple of $8$ for every hidden-shift circuit we have examined. The structure of the circuits described in Ref.~\cite{bravyi2016improved} is such that the CCZ gates come in pairs, acting on different qubits. A possible explanation (that remains to be verified) is that the $\qcompress$ algorithm is reducing the $T$-count of each CCZ gate to exactly $4$ while removing matched pairs of CCZ gates.

Since the output Born-rule probability distribution of the hidden shift circuits is deterministic, we can reconstruct it perfectly by computing $n$ Born-rule probabilities, corresponding to single qubit measurements on each qubit. Hence, we run the $\qcompress$ algorithm $n=40$ times for each hidden-shift instance. In general, one would expect then to run the $\compute$ algorithm $40$ times; however, we observe that many of the cases are already solved (in polynomial time) by $\qcompress$ without having to call either of our exponential-time algorithms. As noted by the authors of Ref.~\cite{bravyi2016improved}, the first $\frac{n}{2}$ bits of the hidden-shift are always recoverable by a polynomial-time computation. The $\qcompress$ algorithm recovers not only these ``free" $\frac{n}{2}$ bits, but an additional fraction of the bit-string, depending on the problem instance. In the $10,000$ instances summarised in Fig.~\ref{fig:compress-hidden-shift-data}, $\qcompress$ recovered $304,307$ bits of the total $400,000$ hidden bits.

To provide data that can be directly compared to the results reported in Ref.~\cite{bravyi2019simulation}, we used our algorithms to solve 100 instances of the hidden shift problem for varying number of CCZ gates. In Fig.~\ref{fig:compute-hidden-shift-data}, we present the run-times and their means for these circuits. As we summarize in Table~\ref{tab:hidden-shift-comparison}, the run-time improvement factors vary from $10^2$ to $10^4$, with the mean of $10^3$.

\begin{figure*}[t]
  \centering
	\sidesubfloat[a]{\includegraphics[width=0.42\columnwidth,viewport=0 0 170 150]{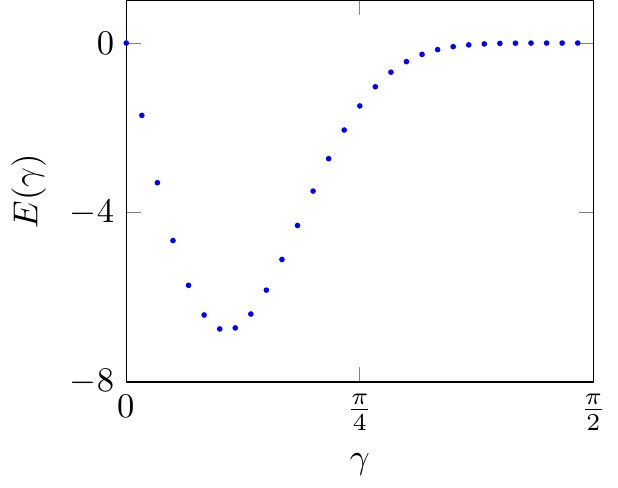}\label{fig:qaoa-1d}}~\hspace{1cm}~\sidesubfloat[b]{\includegraphics[width=0.5\columnwidth]{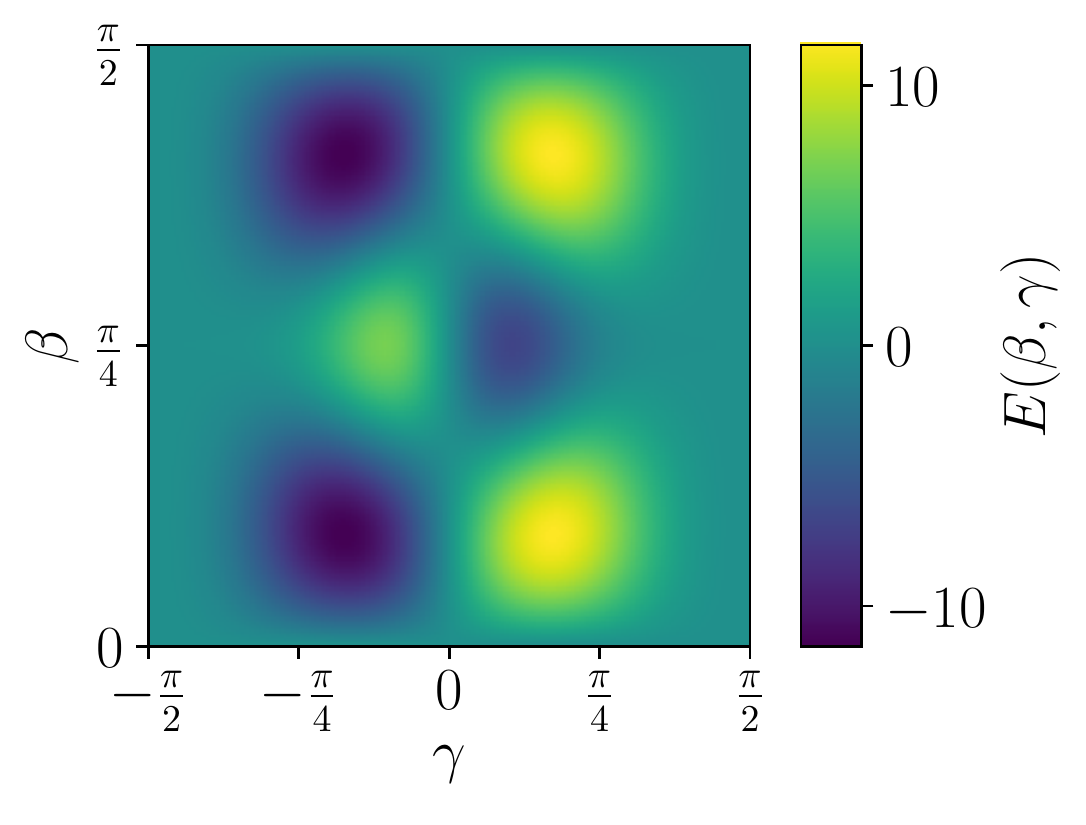}\label{fig:qaoa-2d}}
	\caption{\label{fig:qaoa-full-fig} \textbf{Performance of our algorithms for the QAOA circuits.} The graphs show the QAOA expectation value computed with the $\compute$ algorithm. (a) A slice is shown through the $\beta$, $\gamma$ plane with $\beta = \frac{\pi}{4}$. The authors of Ref.~\cite{bravyi2019simulation} restricted their benchmarks to this line, since doing so substantially reduces the number of non-Clifford gates in the circuit. The data-points are accurate to machine precision and were computed in $1.55$ seconds (total run-time of $\qcompress$ and $\compute$). In contrast, Ref.~\cite{bravyi2019simulation} reports a run-time of ``less than $3$ days''. (b) Unlike the results presented in Ref.~\cite{bravyi2019simulation} we are able to compute QAOA expectation values for points off the line $\beta=\frac{\pi}{4}$. This image was generated from $10,000$ data-points, comprising $100$ $\beta$ and $100$ $\gamma$ values. The total run-time for $\qcompress$ and $\compute$ to generate these data was $563$ seconds.}
\end{figure*}

\begin{table*}[htbp]
    \centering
    \begin{tabular}{c|cccccc}
         & qubit count & D & gamma values & beta values & non-Clifford gate count & runtime\\ \hline
        Ref.~\cite{bravyi2019simulation} & 50 & 4 & 31&1& 66 &  ``less than 3 days"\\
        Fig.~\ref{fig:qaoa-1d} & 50 & 4 & 31&1& 66 &  1.55 seconds\\
        Fig.~\ref{fig:qaoa-2d} & 50 & 4 & 100&100& 116$^{\dagger}$ & 563 seconds\\
        $200$ instances$^*$ & 50 & 4 & 31& 1 & 66 & mean 1.64 seconds, max 2.23 seconds
    \end{tabular}
    \caption{\label{tab:qaoa-comparison}Direct comparison of our QAOA (Max E3LIN2) results with prior works. For the row marked $^*$ mean and max numbers are over a sample of 200 randomly generated QAOA circuits for the Max-E3LIN2 problem. For each instance we use the same parameters as Ref.~\cite{bravyi2019simulation}: $\beta=\frac{\pi}{4}$ and 31 values of $\gamma$ evenly spaced in $\left[0, \frac{\pi}{2}\right]$. Each instance was solved using a single thread (i.e. with no parallelization) on an AMD Ryzen 9 3900X 12-Core Processor operating at approximately 2.2~GHz. It should be noted that while the methods of Ref.~\cite{bravyi2019simulation} provide an approximate answer (with some small probability of failure), $\compute$ deterministically provides the exact answer (to machine precision).\newline 
    $^{\dagger}$ The count of 116 non-Clifford gates is the correct value for the majority of the points plotted, however along certain lines (three horizontal lines where $\beta$ is a multiple of $\frac{\pi}{4}$ and five vertical lines where $\gamma$ is a multiple $\frac{\pi}{4}$) the non-Clifford gate count is lower.}   
\end{table*}

\subsubsection{Quantum approximate optimization algorithm}

The QAOA benchmarks performed in Ref.~\cite{bravyi2019simulation} consist of computing the following expectation values
\begin{align}
    E(\beta,\gamma) = \bra{\psi_{\beta\gamma}} C \ket{\psi_{\beta\gamma}},
\end{align}
where
\begin{align}
    &\!\!C =\!\!\!\!\!\!\!\! \sum_{1\leq u < v < w \leq n} \!\!\!\!\!\frac{d_{uvw}}{2} Z_u Z_v Z_w\mathrm{~~with~~}d_{uvw}\in \{-1,0,1\}\label{eqn:qaoa-c-operator}
\end{align}
and
\begin{align}
    &\ket{\psi_{\beta\gamma}} = e^{-\beta B} e^{-\gamma C} H^{\otimes n}\ket{0},
\end{align}
with $B$ being the so-called transverse field operator, \mbox{$B=\sum_{j=1}^n X_j$}. A problem instance is specified by a particular tensor~$d_{uvw}$. In the problem addressed in Ref.~\cite{bravyi2019simulation}, the \emph{degree} of the problem $D$ is chosen to be $4$, meaning that each qubit appears in at most $4$ non-zero terms in the sum in Eq.~\eqref{eqn:qaoa-c-operator}. To match the benchmarks of Ref.~\cite{bravyi2019simulation}, we use instances where all but one of the $n=50$ qubits appear in exactly $4$ terms. This choice gives a $C$ operator consisting of a sum of $66$~Pauli operators.

We explain in the Supplemental Material Sec.~\ref{sec:supp-pauli-expectation-values}. how our algorithms may be used to estimate or compute expectation values of Pauli operators with very low (polynomial) overhead. Since the Pauli operators are a basis for the space of self-adjoint operators on $\mathbb{C}^{2^n}$, one can in principle use our algorithms to compute expectation values of arbitrary self-adjoint operators. In general, the decomposition of a self-adjoint operator may require exponentially many Pauli operators, making this procedure infeasible. However, since the operator $C$ is a sum of $66$ Pauli operators, we only have to run our algorithms $66$ times to obtain a single QAOA expectation value.

In the benchmarks of Ref.~\cite{bravyi2019simulation}, the authors fixed $\beta=\frac{\pi}{4}$ and chose $31$ values of $\gamma$ in the interval $\left[0,\frac{\pi}{2}\right]$. With this choice of $\beta$, the number of non-Clifford gates in the circuit is reduced from $116$ to $66$, since $e^{-i \frac{\pi}{4} B} \propto \left(H S H\right)^{\otimes n}$, where the $\propto$ symbol hides an irrelevant global phase. The authors of Ref.~\cite{bravyi2019simulation} reported a run-time of ``less than 3 days'' for this simulation. We repeated this benchmark with our algorithms (see Fig.~\ref{fig:qaoa-1d}), and obtained a run-time of less than 2 seconds.
%The same choice of $\beta$ has been used in related research, for example by Farhi et al.~\cite{farhi2014-qaoa}. 
The performance of our algorithms is such that we can also easily compute QAOA expectation values $E(\beta,\gamma)$ for $\beta\neq \pi/4$. Figure~\ref{fig:qaoa-2d} shows the QAOA expectation values for 10 000 points evenly spaced in the region $(0,\pi/2)\times (-\pi/2,\pi/2)$. These 10 000 points, the majority of which (e.g. when $\beta$ and $\gamma$ are not multiples of $\pi/4$) would have been out of the reach of the prior state of the art simulation methods, were simulated in under 10 minutes on a desktop computer using our algorithms.

It is interesting to consider the reason our algorithms exhibit such a dramatic performance improvement. Taking a single Pauli operator from the sum in Eq.~\eqref{eqn:qaoa-c-operator} and considering the expectation value
\begin{align}
    E_{uvw}(\beta,\gamma) = \bra{+}^{\otimes n} e^{i\gamma C} e^{i\beta B} Z_u Z_v Z_w  e^{-i\beta B}e^{-i\gamma C}\ket{+}^{\otimes n},
\end{align}
it is possible to commute all but three of the $e^{-i\beta X_j}$ terms and all but $10$ of the $e^{-i\gamma Z_a Z_b Z_c}$ terms through the central $Z_u Z_v Z_w$ operator. Performing this computation ``by hand" the complexity of this circuit is not $66$ non-Clifford gates, but instead only $13$. We did \emph{not} perform this optimisation in the circuits we gave to our algorithms. However we observed empirically that the output of the $\qcompress$ algorithm had an effective $T$-count not greater than $13$, suggesting that $\qcompress$ is capable of noticing this optimisation.

% -------------------------------------------------------------
% SECTION IV - CONCLUSIONS AND OUTLOOK
% -------------------------------------------------------------

\section{Conclusions and outlook}
\label{sec:outlook}

We have developed state-of-the-art classical simulators for computing and estimating Born rule probabilities associated with universal quantum circuits.  We have made Python+C implementations of these simulators available~\cite{smith2020clifford}.
These simulators allow us to probe previously uncharted parameter regimes such as circuits with larger numbers of qubits and non-Clifford gates than was previously possible. Our results should find direct applications in the verification and validation of near-term quantum devices, and the evaluation of proposals for NISQ device applications.

Although we have tested the implementation of our algorithms on standard desktop hardware, the $\compute$ and $\child$ algorithms are \emph{embarrassingly parallel}, that is they may be divided into large numbers of sub-tasks such that the process addressing each sub-task may proceed without any dependency on, or communication with, any other process. This means that they can trivially be applied in a high performance computing context and so may be useful for very large simulations.

Through the use of our $\qcompress$ algorithm we were able to distill a complex circuit specification to a simpler form more amenable to the task of Born rule probability estimation. The circuit specific parameter, $r$, emerged as a key driver of run-time, with higher values of $r$ improving the run-time of $\compute$ and lower values of $r$ (often) improving the run-time of $\estimate$. Thus the \rname~$r$ is useful in identifying which simulator will be the fastest. In our work, the primary role served by $\compute$ has been to exclude all of the `high $r$ value' circuits from consideration, thus emphasising the performance advantages of our $\estimate$ algorithm over its alternatives. However, it remains an open question if and when the $\compute$ algorithm can be useful in its own right or has a genuinely interesting application. Indeed, in the extreme regime where $r=t$, $\compute$ outputs a Born rule probability consistent with the uniform distribution on all measured `non-deterministic' qubits. However, as $r$ moves away from $t$, perhaps the quantity $t-r$ (or other information contained in the stabilizer generating set $G$) can be viewed as a measure of departure from `non-uniformity'. This is broadly consistent with the outcome of our numerical analysis of high Clifford count randomly generated circuits, where we found that $r$ strongly concentrates near its maximum value of $\min \set{t, n-w}$. We leave the exploration of this narrative and the identification of other key drivers of outcome distribution structure to future work. 

This work has focused on the task of Born rule probability estimation without discussing the related task of approximately sampling from the quantum outcome distribution. Some of the techniques we have developed here may also be useful for achieving performance improvements for the task of approximate sampling.

We have restricted our attention to the simulation of ideal or noise-free quantum processes. Realistic implementations of quantum circuits are subject to noise the presence of which can significantly ease the computational cost of classical simulation. We leave open the generalization of our work to the mixed state formalism. We point out that for the task of approximate sampling, an analogous generalization (to the BG-sampling algorithm~\cite{bravyi2016improved}) was recently shown in Ref.~\cite{seddon2020quantifying} with additional performance gains being achieved as a consequence of this generalization.

% -------------------------------------------------------------
% ACKNOWLEDGEMENTS
% -------------------------------------------------------------

\section*{Acknowledgements}

HP acknowledges Marco Tomamichel for identifying an error in the statement of Lemma~\ref{lem:hoeffding} in an early draft; David Gosset for useful discussions regarding the CH-form; and Daniel Grier and Luke Schaeffer for useful discussions regarding the hardness of computing tight upper bounds associated with Lemma~\ref{lem:psi(y)_bound}.
Research at Perimeter Institute is supported in part by the Government of Canada through the Department of Innovation, Science and Economic Development Canada and by the Province of Ontario through the Ministry of Colleges and Universities. HP also acknowledges the support of the Natural Sciences and Engineering Research Council of Canada (NSERC) discovery grants [RGPIN-2019-04198] and [RGPIN-2018-05188]. KK and ORS acknowledge financial support by the Foundation for Polish Science through TEAM-NET project (contract no. POIR.04.04.00-00-17C1/18-00). This work is supported by the Australian Research Council (ARC) via the Centre of Excellence in Engineered Quantum Systems (EQuS) project number CE170100009.  Research was partially sponsored (SB) by the ARO and was accomplished under Grant Number: W911NF-21-1-0007. The views and conclusions contained in this document are those of the authors and should not be interpreted as representing the official policies, either expressed or implied, of ARO or the U.S. Government.

\appendix
\onecolumngrid
% -------------------------------------------------------------
% APPENDIX A - THE COMPRESS ALGORITHM
% -------------------------------------------------------------

\section{The {\normalfont\texorpdfstring{$\qcompress$}{Compress}} algorithm}
\label{app:compress}

\subsection{Step 1: Gadgetization}
\label{sec:gadget}

It is well known that a $T$ gate acting on a given qubit can be replaced by its gadgetized version~\cite{gottesman1999demonstrating,zhou2000methodology}. More precisely, one can prepare an ancillary qubit in a magic state
\begin{equation}
	\ket{T}=\frac{1}{\sqrt{2}}(\ket{0}+\exp(i\pi/4)\ket{1}),\label{eq: magic T state}
\end{equation}
couple it to the original qubit by a $CX$ gate (with the original qubit acting as the control) and measure in the computational basis. Then, if the outcome is $\ket{1}$, one also needs to apply a correction Clifford phase gate $S$ to the original qubit. The effect of the above procedure is the same as direct application of the $T$ gate to a given qubit. Diagrammatically, with circuits read from right to left throughout the paper,
\begin{equation} %
  \begin{matrix}\scalebox{-1}[1]{\includegraphics{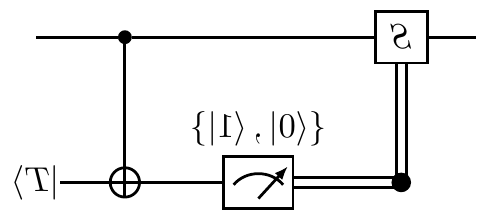}}\end{matrix}
  =\begin{matrix}
    \includegraphics{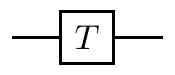}
  \end{matrix}.
\end{equation}
Here, we will employ an alternative construction that replaces the ancillary non-stabiliser $\ket{T}$ state with a non-Clifford measurement, and allows one to implement any diagonal $\Tgate{\phi}$ gate. Our reverse gadget is obtained as follows:
\begin{equation}
	\begin{matrix}\scalebox{-1}[1]{\includegraphics{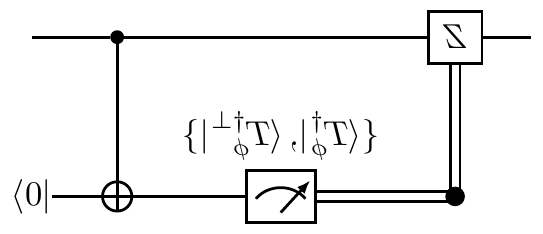}}\end{matrix}
	=\begin{matrix}
          \includegraphics{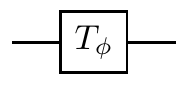}
	\end{matrix},
\end{equation}
with
\begin{equation}
	\Tket{\phi}=\frac{1}{\sqrt{2}}(\ket{0}+\exp(-i\phi)\ket{1}),\quad\Tperpket{\phi}=\frac{1}{\sqrt{2}}(\ket{0}-\exp(-i\phi)\ket{1}).
\end{equation}
Next, it is straightforward to show that in the above reverse gadget the measurement outcomes of the ancillary qubit are equally likely. Therefore, we can focus on the $\Tket{\phi}$ outcome (as no correction gates are then needed), and consider a simplified post-selected circuit:
\begin{equation}
\label{eq:reverse-gadget}
\begin{matrix}
  \includegraphics{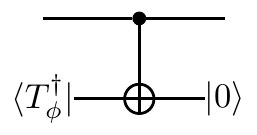}
\end{matrix}
=\begin{matrix}
  \includegraphics{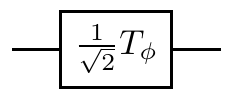}
\end{matrix}.
\end{equation}

\begin{figure}
\includegraphics[width=\columnwidth]{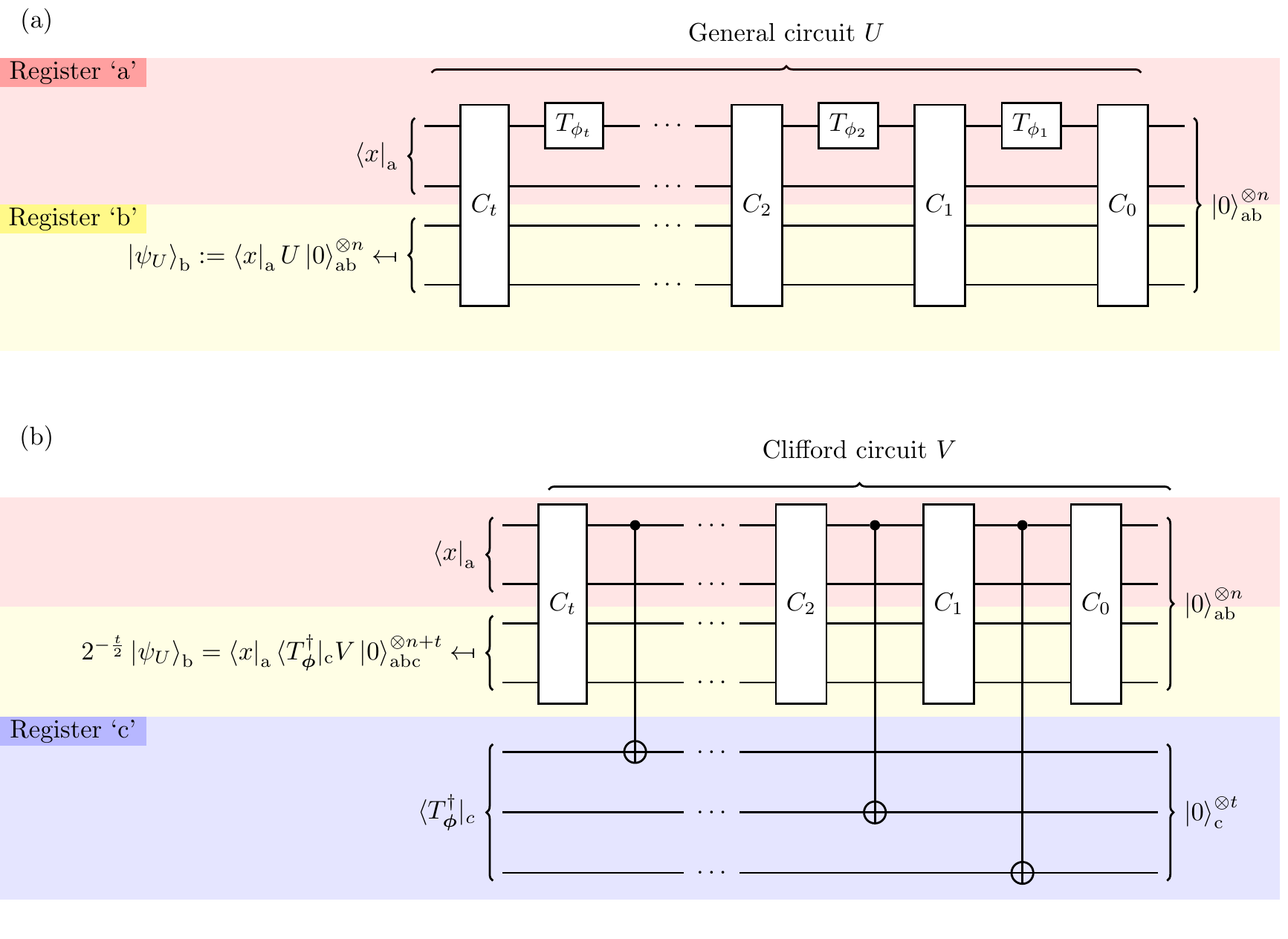}
	\caption{\textbf{Circuits and registers.} Circuit diagrams should be read from right to left. (a) General circuit $U$ composed of Clifford gates and $t$ diagonal $\Tgate{\phi_j}$ gates, post-selected on the $x$ outcome of $w$ qubits in register `$\rega$'. Gates $C_j$ consist of all Clifford gates appearing before the $j$-th diagonal gate $\Tgate{\phi_j}$. Note that non-Clifford diagonal gates act only on the first qubit for clarity of the figure and without loss of generality (since $SWAP$ gates are Clifford, each $\Tgate{\phi_j}$ gate can effectively act on any qubit). (b) Post-selected circuit obtained by reverse gadgetization of~$U$. The unitary $V$ is obtained from $U$ by replacing each $\Tgate{\phi_j}$ with a $CX$ gate between the original qubit and an ancillary qubit in register `$\regc$'. Qubits in register `$\regc$' are then post-selected on the non-stabiliser state $\Tkets$, while qubits in register `$\regb$' are post-selected on the same outcome as in the original $U$ circuit.}
	\label{fig:circuit}
\end{figure}
Now, for a circuit $U$ consisting of $c$ Clifford gates and $t$ diagonal gates $\{\Tgate{\phi_j}\}$, we can gadgetize each of the $t$ occurrences of the non-Clifford gate in the way described above. Hence, we can replace a general unitary circuit $U$ on $n$ qubits in a state $\ket{0}^{\otimes n}_{\rega\regb}$ by a Clifford circuit $V$ on $n+t$ qubits in a state $\ket{0}^{\otimes n+t}_{\rega\regb\regc}$, which is post-selected on $\Tkets$ outcome on the ancillary qubits in register `$\regc$'. The unitary $V$ is composed of $c+t$ Clifford gates: the original $c$ Clifford unitaries appearing in the decomposition of $U$ into Cliffords and non-Clifford gates, plus $t$ instances of $CX$ gates between computational and ancillary qubits arising from the reverse gadgetization of $\Tgate{\phi_i}$ gates. We illustrate this in Fig.~\ref{fig:circuit}, where we also present the division of all $n+t$ qubits into 3 registers: the measured register `$\rega$' that we post-select on the $\ket{x}$ outcome, the marginalised register `$\regb$', and register `$\regc$' consisting of the ancillary qubits that we post-select on the $\Tkets$ outcome. Due to the fact that all measurement outcomes in reverse gadgets are equally probable, such a post-selected circuit $V$ will realise $U$ up to a renormalization factor:
\begin{align}
	\label{eq:psi-U-gadget}
	U \ket{0}^{\otimes n}_{\rega\regb}=2^{t/2} \Tbras_{\regc} V \ket{0}_{\rega\regb\regc}^{\otimes n+t}.
\end{align}
The probability of observing outcome $x$ is thus given by
\begin{equation}
	\label{eq:born-prob}
	p= 2^t\twonorm{\bra{x}_{\rega} \Tbras_{\regc} V\ket{0}_{\rega\regb\regc}^{\otimes n+t}}^2.
\end{equation}
The process of constructing $V$ given an elementary description of $U$ obviously has a polynomial run-time \mbox{$\mathrm{poly}(n,c,t)$}.

% -------------------------------------------------------------

\subsection{Step 2: Constraining stabilisers}
\label{sec:constraining}

In this step we will use the stabilizer formalism introduced in Refs~\cite{gottesman1998heisenberg, aaronson2004improved} to rewrite the expression for $p$ given in Eq.~\eqref{eq:born-prob} in a simplified form. It will lead directly to the $\compute$ algorithm (see Appendix~\ref{app:compute_proof}), and will be further simplified in the next step before serving as an input to the $\child$ algorithm. Moreover, we will also extract the crucial parameters describing the circuit $V$: the \rname~$r$ and the \vname~$v$. The first one of these effectively characterizes how much the number $t$ of non-Clifford diagonal gates can be compressed, while the latter one is related to the number of outcomes with a zero probability.

Let us first briefly introduce some notation and recall standard techniques within the stabilizer formalism. An $n$-qubit Pauli operator, $P$, is any operator of the form \mbox{$\omega P_1 \otimes \ldots \otimes P_n$} where \mbox{$\omega\in \set{\pm 1, \pm i}$} and \mbox{$P_j \in \set{I, X, Y, Z}$} are single qubit Pauli operators. We denote the set of all $n$-qubit Pauli operators by $\pauli{n}$. For any $P\in\pauli{n}$ and $j\in [n]$, we use $\pfac{P}{j}$ to denote the $j^{\rm th}$ tensor factor $P_j$ and $\pphase{P}$ to denote the phase factor $\omega$. We will slightly abuse notation by using $\pfac{P}{\rega}$ to denote the sub-string of tensor factors associated with register~`$\rega$', i.e. $\pfac{P}{\rega}:=\otimes_{j\in [w]} P_j$ and similarly for $\pfac{P}{\regb}$ and $\pfac{P}{\regc}$. 

We say that $P\in \pauli{n}$ stabilizes an $n$-qubit quantum state $\ket{\psi}$ if and only if $P\ket{\psi}=\ket{\psi}$. The subset $S(\ket{\psi})\subset \pauli{n}$ consisting of all stabilizers of $\ket{\psi}$ is an Abelian group isomorphic to $\bbz_2^n$. This group can be non-uniquely represented by a generator set $G=\set{g_1,\ldots,g_n}\subset S(\ket{\psi})$ such that $S(\ket{\psi})=\gen{G}$. For $k\leq n$, $G=\set{g_1,\ldots,g_k}$ is an $n$-qubit, $k$-element generating set if and only if $g_i\in \pauli{n}$ for all $i\in [k]$, all pairs $g_i, g_j\in G$ commute and $G$ is independent, i.e. for all $i\in [k]$, $g_i \not\in \gen{G\setminus \set{g_i}}$. We denote the set of all $n$-qubit, $k$-element generating sets by $\genset{n}{k}$. For $G=\set{g_1,\ldots,g_k}\in \genset{n}{k}$ we define the associated projector:
\begin{subequations}
	\begin{align}
	\Pi_G&:=\prod_{i=1}^k \frac{I+g_i}{2}\\
	&=2^{-k}\sum_{g\in \gen{G}} g \label{eq:stab projector defn}.
	\end{align}
\end{subequations}

We can now state the critical lemma of this step. Its rigorous proof including the pseudo-code of the algorithm and associated sub-procedures can be found in Supplemental Material Sec.~\ref{sec:extract detailed}. Here, we will limit ourselves to a high level description of the main idea behind the proof. 

\begin{lem}[\extract~algorithm]
	\label{lem:extract}
	Given an elementary description of $p$, $\extract$ outputs \vname~\mbox{$v\in \vset$}, \rname~\mbox{$r\in \rset$}, a set $J=\set{j_1,\ldots,j_v}\subseteq [w]$, a bitstring $x'=(x_1',\ldots,x_v')$ and two generating sets $\tilde{G} \in \genset{n+t}{t-r+v}$ and  $G \in \genset{t}{t-r}$ such that:
	\begin{subequations} 
		\begin{align}
		{\rm Tr}_{\rega \regb}\left(V\density{0}^{\otimes n+t}_{\rega \regb \regc} V^{\dagger} \density{x}_{\rega}\right)&=2^{-n-r+v}{\rm Tr}_{\rega \regb}\left(\Pi_{\tilde{G}} \density{x}_{\rega}\right)\label{eq:extract long}\\
		&=2^{-r+v-w}\Pi_G,\label{eq:extract short}
		\end{align}
	\end{subequations} 
	and for all $k\in [v]$, $x_{j_k}\neq x_k'$ immediately implies $\Pi_G=0$. The run-time of the $\extract$ algorithm is polynomial in the relevant parameters:
	\begin{equation}
	\label{eq:congens run-time}
	\tau_{\extract}=\mathrm{poly}(c,n,t).
	\end{equation}
\end{lem}

Using Eq.~\eqref{eq:born-prob}, we note that Lemma~\ref{lem:extract} immediately implies that we can rewrite the Born rule probability $p$ in the following two ways:
\begin{subequations}
	\begin{align}
	p&=2^{-n+t-r+v} \tr{\Pi_{\tilde{G}} \density{x}_{\rega} \otimes I^{\otimes n-w}_{\regb} \otimes \Tkets \! \Tbras_{\regc}}\label{eq:congens long}\\
	&=2^{t-r+v-w} \tr{\Pi_{G} \Tkets \! \Tbras}=2^{v-w}  \Tbras\prod_{i=1}^{t-r}(I+g_i) \Tkets,\label{eq:congens short}
	\end{align}
\end{subequations} 
where in the last equality the product is over all $g_i\in G$. Moreover, since for all $k\in [v]$, $x_{j_k}\neq x_k'$ immediately implies $\Pi_G=0$, it also implies $p=0$. 

The high level description of the proof of Lemma~\ref{lem:extract} goes as follows. First, we rewrite $V\density{0}^{\otimes n+t}_{\rega \regb \regc} V^{\dagger}$ appearing on the left hand side of Eq.~\eqref{eq:extract long} as a stabilizer projector $\Pi_{\gen{G^{(0)}}}$ in the form of Eq.~\eqref{eq:stab projector defn}. Viewing $\Pi_{\gen{G^{(0)}}}$ as a sum over stabilizers, we note that to contribute non-trivially to the sum in Eq.~\eqref{eq:extract long}, a stabilizer must satisfy certain constraints. In particular, for a fixed $g\in \gen{G^{(0)}}$ to produce a non-zero contribution to the sum, it is necessary that:
\begin{itemize}
	\item Register `$\rega$' constraints: for all $j\in [w]$, $\pfac{g}{j}\in \set{I,Z}$,
	\item Register `$\regb$' constraints: for all $j\in [n-w]$, $\pfac{g}{w+j}= I$.
\end{itemize}
The generating set $\tilde{G}\in \genset{n+t}{t-r+v}$ is defined (and computed from $G^{(0)}$) such that the stabilizer group $\gen{\tilde{G}}$ contains $g\in \gen{G^{(0)}}$ if and only if $g$ satisfies all of these constraints. From this $(n+t)$-qubit stabilizer group, we compute a ``compressed'' generating set, $G$, of a $t$-qubit stabilizer group $\gen{G}$. The quantity $r$ is indirectly defined by:
\begin{align}
\abs{G}=t-r.\label{eqn:defining-r-indirectly}
\end{align}

The quantity $v$ is implicitly defined by the equation $|\tilde{G}|=t-r+v$. Together with related objects, $J$ and $x'$, the quantity $v$ is associated with the compression step, i.e., transforming $\tilde{G}$ into $G$. Here, each generator $g\in \tilde{G}$ is mapped to a Pauli $f_x(g)\in \pauli{t}$ where $f_x(g):=\pphase{g}\bra{x}\pfac{g}{\rega}\ket{x} \pfac{g}{\regc}$. The set $\set{f_x(g)|g\in \gen{\tilde{G}}}$ is a group but the set of Pauli operators $\set{f_x(g)|g\in {\tilde{G}}}$ may not be independent. That is, for a fixed $x\in \bitstring{w}$ and $g^*\neq I^{\otimes n+t}$, it is possible that $f_x(g^*)\in \set{\pm I^{\otimes t}}$. When $f_x(g^*)=-I^{\otimes t}$, the sum over $g\in \tilde{G}$ of $f_x(g)$ is zero. The objects $J$ and $x'$ specify the constraints on $x$ that ensure $-I^{\otimes t}\not\in \set{f_x(g)|g\in \gen{\tilde{G}}}$. When $f_x(g^*)=I^{\otimes t}$, the sum over $g\in \gen{\tilde{G}}$ of $f_x(g)$ contains duplicate sums over the group. The quantity $v$ is the minimal number of deletions to $\tilde{G}$ required to ensure the image under $f_x$ is an independent set.

% -------------------------------------------------------------

\subsection{Step 3: Gate sequence construction}
\label{sec:gate-seq}

So far we have replaced a general circuit $U$ with a post-selected Clifford circuit $V$ in Step~1, and then employed the stabilizer formalism in Step~2 to re-express the Born probability $p$ using a compressed stabiliser projector $\Pi_G$. Now, the final step is to go back from the compressed projector picture to a compressed unitary circuit $W$ built of Clifford gates. The aim of this step is summarised by the following lemma.

\begin{lem}[$\gateseq$ subroutine]
	\label{lem:gateseq}
	Given a stabilizer generator matrix ${G} \in \genset{t}{t-r}$, $\gateseq$ outputs an elementary description of a $t$-qubit Clifford unitary $W$ such that:
	\begin{equation}
		\label{eq:gateseq}
		\Pi_G=W^\dagger ( \density{0}^{\otimes t-r} \otimes I^{\otimes r}) W.
	\end{equation}
	The circuit $W$, consists of $\order{t^2}$ Clifford gates including at most $\order{t}$ Hadamard gates, and the run-time scaling of the algorithm is given by 
	\begin{equation}
		\tau_{\gateseq} = \mathrm{poly} (n,c,t).
	\end{equation}
\end{lem}
We note that $W$ can be interpreted as a unitary encoding of the stabilizer code.
The proof of the above lemma can be found in Supplemental Material Sec.~\ref{app:gate_seq}, and it is simply based on an explicit construction of a circuit $W$ out of elementary Clifford gates using the stabilizer formalism. 

Applying Lemma~\ref{lem:gateseq} to Eq.~\eqref{eq:congens short} we immediately get
\begin{align}
	p=2^{t-r+v-w}\twonorm{\bra{0}^{\otimes t-r}W \Tkets}^2,
\end{align}
which is precisely the main statement of Theorem~\ref{thm:compress} (with the second equality already proven in Eq.~\eqref{eq:congens short}). Moreover, since Steps~1~to~3 all required polynomial number of operations, the total run-time of the $\qcompress$ algorithm is \mbox{$\mathrm{poly}(n,c,t)$}, and so we have proven Theorem~\ref{thm:compress}.

\newtext{
\subsection{{\normalfont\texorpdfstring{$T$}{T}}-count reduction extension}
\label{sec:T-count reduction}
In this section, we describe an extension to Theorem~\ref{thm:compress} that allows us to compute an effective $T$-count $t'$ that may be lower than the $T$-count of the initial input circuit and can, in certain cases, result in a significant reduction in the run-times of our $\compute$, $\child$ and $\parent$ algorithms.

Having applied Lemma~\ref{lem:extract} to express the Born-rule probability in the form
\begin{equation}
    p = 2^{v-w}  \Tbras\prod_{i=1}^{t-r}(I+g_i) \Tkets,\tag{\ref{eq:congens short}}
\end{equation}
we note that we can also apply constraints to the register `$\regc$' qubits. First, note that the magic states comprising $\Tkets$ all lie on the equator of the Bloch sphere and, in particular, have $Z$-expectation-value $\Tbra{\phi_i} Z \Tket{\phi_i} = 0$. Given a particular register $\regc$ qubit $q$, we first check if any stabilizer generator contains a Pauli-$X$ or $Y$ operator on qubit $q$. If not, then we can multiply between the generators to obtain a new generating set with exactly one generator having a non-identity operator on $q$, assume this generator is $g_1$. Rewriting equation~\eqref{eq:congens short}, we obtain
\begin{align}
    p &= 2^{v-w}  \Tbras(I + g_1)\prod_{i=2}^{t-r}(I+g_i) \Tkets,\\
    &= 2^{v-w}  \Tbras\prod_{i=2}^{t-r}(I+g_i) \Tkets,
\end{align}
where the second equality follows since $\pfac{g_1}{q} = Z$, while $\pfac{g_i}{q} = I$, for $i\geq 2$. Following the removal of stabilizer $g_i$, we are left with a \emph{trivial qubit}; a qubit for which every remaining stabilizer in the generating set is the identity. Since the expectation value of the identity in any state is $1$, this qubit may \emph{also} be removed from our generating set. We repeat this procedure for each qubit until no more generators are removed. Note that the removal of a generator associated with one qubit can result in another qubit which previously contained, e.g., both $Z$ and $X$ generators, now only containing a $Z$ generator. Thus, if a generator is removed on any round of sweeping through each qubit, then another round must be performed, resulting in at most $t^2$ qubit checks. This gives a polynomial-time algorithm which reduces the number of generators from $t-r$ to $t^\prime -r^\prime$. This step defines the difference $t^\prime -r^\prime$, leaving the components $t^\prime$ and $r^\prime$ unspecified until the next step.

The removal of $(t-r)-(t^\prime -r^\prime)$ generators in the previous step will result in the creation of a matching number of trivial qubits. There may also be $u\geq 0$ other trivial qubits arising from the application of register $\rega$ and $\regb$ constraints. We remove these $(t-r)-(t^\prime -r^\prime)+u$ trivial qubits from the stabilizer table. This step does not change the number of stabilizer generators, leaving us with a stabilizer tableau of $t^\prime$ qubits and $t^\prime - r^\prime$ generators. Thus, the number of qubits at the end of this procedure defines $t^\prime$.

The string of generating sets that are produced by these manipulations are summarized below:
\begin{align*}
    G^{(0)}\in \genset{n+t}{n+t} \overset{(1)}{\rightarrow} \tilde{G}\in \genset{n+t}{t-r+v}\overset{(2)}{\rightarrow} G\in \genset{t}{t-r}\overset{(3)}{\rightarrow} \tilde{G}^\prime\in \genset{t}{t^\prime -r^\prime}\overset{(4)}{\rightarrow} G^\prime\in \genset{t^\prime}{t^\prime-r^\prime}
\end{align*}
where step $(1)$ corresponds to the application of register $\rega$ and $\regb$ constraints that remove generators, step $(2)$ corresponds to the first qubit removal step associated with the $f_x$ map, step $(3)$ corresponds to the application of register $\regc$ constraints and step $(4)$ corresponds to the second qubit removal step associated with removing additional trivial qubits. We note that the removal of these qubits produces a corresponding $t^\prime$ qubit magic state $\Tketsprime$ constructed from $\Tkets$ in the obvious way. We also note that the gate sequence construction step presented in Sec.~\ref{sec:gate-seq} can be applied at the level of the generating set $G'$ with all calculations following analogously.

We now show bounds on these primed variables as claimed in Remark~\ref{rmk}. It is clear that $t^\prime\leq t$ since no qubits were added in steps $(3)$ and $(4)$. By the multiplicativity of the stabilizer extent, it is clear that the removal of each magic qubit $\Tket{\phi}$ will result in a reduction of the stabilizer extend by a factor of $\xi(\Tket{\phi})^{-1}\leq 1$. Hence, $\xi^{\prime}\leq \xi^*$.  Additionally, we note that $r$ and $r^\prime$ can be defined as the difference between the number of qubits and the number of stabilizers immediately after steps $(2)$ and $(4)$, respectively. Since the number of qubits removed in step $(4)$ exceeds the number of generator removed in step $(3)$ by $u\geq 0$, we see that $r^\prime=r-u\leq r$. Finally, since $(t-r)-(t^\prime -r^\prime)\geq 0$ generators were removed in step $(3)$, it follows that $(t^\prime -r^\prime)\leq t-r$. 

For completeness, we note that in the general case where $t^\prime\neq t$, the run-times of our $\compute$ algorithm given in Eq.~\eqref{alg:run-time compute}, and $\child$ algorithm given in Eq.~\eqref{eq:child run-time} can be modified to:
	\begin{align}\label{alg:run-time general compute}
		\tau_\compute=\order{2^{t^\prime-r^\prime}t'}.
	\end{align}
and:
	\begin{align}
		\label{alg:run-time general child}
		\tau_{\child} = \order{s {t^\prime}^3 + sL {r^\prime}^3}.
	\end{align}
In addition, we note that for fixed precision parameters $\epsilon_{\rm tot}$ and $\delta_{\mathrm{tot}}$, an exponentially smaller parameter $s$ can be used since in Eq.~\eqref{eq:thm3}, we replace $\xi^*$, the stabilizer extent of $\Tkets$, by $\xi'$, the stabilizer extent of $\Tketsprime$. That is, for fixed precision parameters $\epsilon_{\rm tot}$ and $\delta_{\mathrm{tot}}$, we would now require $s$ and $L$ sufficiently large to satisfy:
\begin{align}\label{eq:general thm3}
	\bigpr{\abs{\hat{p}-p}\geq \epsilon_{\rm tot}}\leq 2e^2\exp\left(\frac{-s(\sqrt{p+\epsilon}-\sqrt{p})^2}{2(\sqrt{\xi'}+\sqrt{p})^2}\right)+\exp\left(-\left(\frac{\epsilon_{\rm tot}-\epsilon}{p+\epsilon}\right)^2 L \right)=:\delta_{\mathrm{tot}}.
\end{align}
%\begin{equation}
% 	\begin{aligned}\label{eq:general thm3}
% 	\hspace{-2ex}\bigpr{\abs{\hat{p}-p}\geq \epsilon_{\rm tot}}\leq &2e^2 \exp\left(\frac{-s(\sqrt{p+\epsilon}-\sqrt{p})^2}{2(\sqrt{\xi'}+\sqrt{p})^2} \right)+\exp\left(-\left(\frac{\epsilon_{\rm tot}-\epsilon}{p+\epsilon}\right)^2 L \right)=:\delta_{\mathrm{tot}}.
% 	\end{aligned}
%	\end{equation}
The run-time improvements to $\tau_{\child}$ flow through to the $\parent$ algorithm as expected.}

% -------------------------------------------------------------
% APPENDIX B - THE COMPUTE ALGORITHM
% -------------------------------------------------------------

\section{The {\normalfont\texorpdfstring{$\compute$}{Compute}} algorithm}
\label{app:compute_proof}

The $\compute$ algorithm computes $p$ from Eq.~\eqref{eq:congens short} by multiplying out the product into a sum of $2^{t-r}$ terms
\begin{align}
    p &= 2^{v-w}  \Tbras\prod_{i=1}^{t-r}(I+g_i) \Tkets= 2^{v-w} \sum_{g\in \langle G\rangle } \Tbras g \Tkets~\label{eq:congens sum form}
\end{align}
and directly evaluating each term in time $\order{t}$. The terms of the sum can be ordered such that for $j\in \set{1,\ldots, 2^{t-r}-1}$, the $j^{\rm th}$ term $2^{v-w}  \Tbras P_j \Tkets$ of the sum has the form $2^{v-w}  \Tbras g_{i_j} P_{j-1} \Tkets$, where $P_0$ is just the identity operator on every qubit. Here, $g_{i_j}$ is one of the $t-r$ stabilizer generators and for all $j$, the index $i_j\in [t-r]$ can be computed in time $\order{t}$. Multiplication of the length $t$ Pauli operators and evaluation of the expectation value both take time $\order{t}$. By only storing in running memory the partial sum up to the $j^{\rm th}$ term and the Pauli $P_j$, the algorithm iterates through all the terms with run-time $\order{2^{t-r}t}$.

We now establish that the generator index $i_j$ can be computed in the claimed run-time and that all the terms in the sum are included exactly once. The group $\langle G \rangle$ generated by the set of stabilizers $G$ is isomorphic to $\mathbb{Z}_2^{t-r}$. In particular, we identity the $i^\text{th}$ generator appearing in the stabilizer tableau with the bitstring that has a $1$ in position $i$ and all other bits equal to $0$. If the bistrings associated with two group elements differ by a single bit in position $i$ we may compute one from the other by multiplying by the $i^\text{th}$ stabilizer generator.
Thus, we require an enumeration of the bitstrings of length $t-r$ such that subsequent bitstrings in the enumeration differ by a single bit. The well-known reflected binary Gray code~\cite{BitnerBinaryReflectedGray1976}, $\operatorname{Gray}: \{0,1\hdots 2^{t-r} -1\}\to \mathbb{Z}_2^{t-r}$ has exactly this property and may be evaluated as
\begin{align}
    \operatorname{Gray}(j) = B(j) \oplus B\left(\left \lfloor \frac{j}{2}\right\rfloor\right),
\end{align}
where $B$ takes natural number $j$ to its usual binary representation, and $\oplus$ denotes element wise addition mod-$2$. Thus, $\operatorname{Gray}(j) \oplus \operatorname{Gray}(j-1)$ is a bitstring with unit Hamming weight. The position of `1' in this bitstring is the index $i_j$ and can be computed in run-time $\order{t}$ as claimed.

% -------------------------------------------------------------
% APPENDIX C - THE RAWESTIM ALGORITHM
% -------------------------------------------------------------

\section{The {\normalfont\texorpdfstring{$\child$}{RawEstim}} algorithm}
\label{app:child}

\subsection{Step 1: Stabiliser decomposition and sampling}
\label{sec:decomposition}

Each state $\Tket{\phi_j}$ appearing in $\Tkets$ can be decomposed into stabilizer states,
\begin{equation}
	\ket{\tilde{0}}:=\ket{+}=\frac{1}{\sqrt{2}}(\ket{0}+\ket{1}),\qquad \ket{\tilde{1}}:=\ket{-i}=\frac{1}{\sqrt{2}}(\ket{0}-i\ket{1}),
\end{equation}
as follows:
\begin{equation}
	\label{eq:T-decomp}
	\Tket{\phi_j}=\alpha_{\phi_j} \ket{\tilde{0}}+\alpha_{\phi_j}' \ket{\tilde{1}},
\end{equation}
where
\begin{equation}
	\alpha_{\phi_j}=\frac{i+e^{-i\phi_j}}{1+i}=e^{i\varphi_j}\sqrt{1-\sin\phi_j},\qquad \alpha_{\phi_j}'=\frac{1-e^{-i\phi_j}}{1+i}=e^{i\varphi_j'}\sqrt{1-\cos\phi_j},
\end{equation}
for some phases $\varphi_j, \varphi_j'$.

The above decomposition achieves the minimum defining the stabiliser extent $\xi$~\cite{bravyi2019simulation},
\begin{equation}
	\label{eq:extent}
	\xi(\ket{\psi}):=\min_{c} \left\{ \|c\|_1^2 ~\vline ~\ket{\psi}=\sum_j c_j \ket{\sigma_j},~~\ket{\sigma_j}~\text{is~a~stabiliser~state}\right\},
\end{equation}	
i.e.,
\begin{equation}
	\xi(\Tket{\phi_j})=(|\alpha_{\phi_j}|+|\alpha'_{\phi_j}|)^2=(\sqrt{1-\sin\phi_j}+\sqrt{1-\cos\phi_j})^2.
\end{equation}	
We choose this particular decomposition because it minimises the run-time of the algorithm: as we will shortly see, it scales in the square of the $l_1$-norm of the expansion coefficients. Moreover, as proven in Ref.~\cite{bravyi2019simulation}, the stabilizer extent for products of single-qubit states is multiplicative. Thus, denoting by $\cextent$ the total stabiliser extent of all states coming from reverse gadgetization of non-Clifford gates in $U$, we have
\begin{equation}
	\cextent:=\xi(\Tkets)=\prod_{j=1}^t\xi(\Tket{\phi_j}),
\end{equation} 
and so the optimal stabiliser decomposition of $\Tkets$ is simply obtained by decomposing each $\Tket{\phi_i}$ according to Eq.~\eqref{eq:T-decomp}. In Fig.~\ref{fig:extent}, we present the values of the stabiliser extent of $\Tket{\phi}$ as a function of $\phi$.

\begin{figure}[t]
\includegraphics[width=0.45\columnwidth]{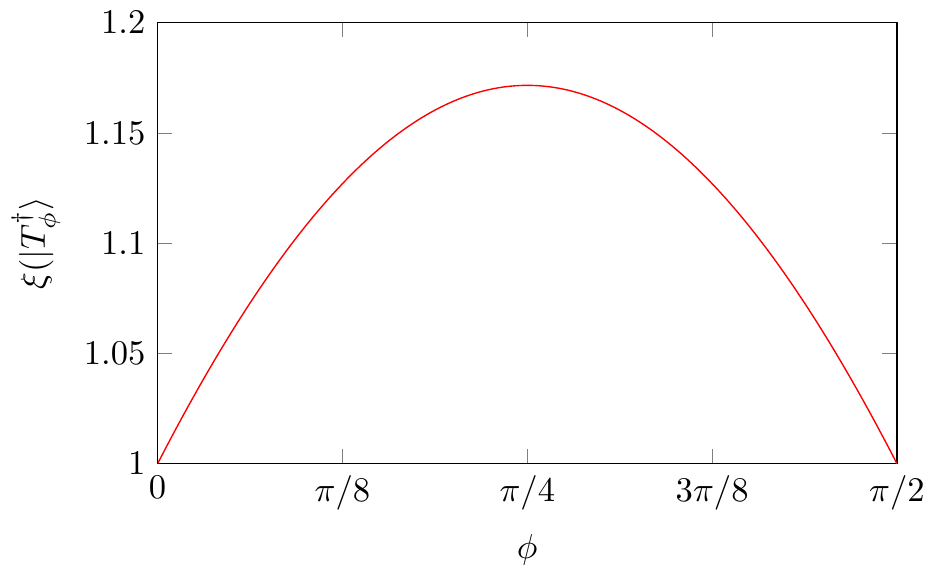}
	\caption{\label{fig:extent}\textbf{Stabiliser extent.} The values of the stabiliser extent $\xi$ of $\Tket{\phi}$ states as a function of $\phi$. Note that the maximum at $\phi=\pi/4$ is achieved for $2^\gamma$ with $\gamma \approx 0.228$ being exactly the same as in the exponential component of the run-time of the sampling algorithm presented in Ref.~\cite{bravyi2016improved}.}
\end{figure}

Using the optimal stabiliser decomposition, we can rewrite Eq.~\eqref{eq:compressed} as follows
\begin{subequations}
\begin{align}
	p&=2^{t-r+v-w} \twonorm{\sum_y \prod_{j=1}^t \alpha_{\phi_j}^{1-y_j}\alpha_{\phi_j}'^{y_j} \bra{0}^{\otimes t-r}W\ket{\tilde{y}}}^2\\
	&=\cextent\cdot 2^{t-r+v-w}  \twonorm{\sum_y \prod_{j=1}^t \frac{\alpha_{\phi_j}^{1-y_j}\alpha_{\phi_j}'^{y_j}}{|\alpha_{\phi_j}|+|\alpha'_{\phi_j}|} \bra{0}^{\otimes t-r}W\ket{\tilde{y}}}^2\\
	&=\cextent\cdot 2^{t-r+v-w} \twonorm{\sum_y q(y) \prod_{j=1}^te^{i\varphi_j(1-y_j)}e^{i\varphi'_j y_j} \bra{0}^{\otimes t-r}W\ket{\tilde{y}}}^2,
\end{align}
\end{subequations}
where $q(y)$ is a normalised product probability distribution,
\begin{equation}\label{eq: y distribution}
	q(y)=\prod_{j=1}^t q(y_j),\quad q(y_j)=\left\{
	\begin{array}{l}
		\frac{|\alpha_{\phi_j}|}{|\alpha_{\phi_j}|+|\alpha'_{\phi_j}|}\quad \mathrm{for~}y_j=0,\\
		\frac{|\alpha'_{\phi_j}|}{|\alpha_{\phi_j}|+|\alpha'_{\phi_j}|}\quad \mathrm{for~}y_j=1.
	\end{array}
	\right.
\end{equation}
Therefore, we can introduce the following unnormalised states:
\begin{equation}
	\label{eq:psi(y)}
	\ket{\psi(y)}:=\sqrt{\cextent} \cdot 2^{\frac{t-r+v-w}{2}} \prod_{j=1}^t e^{i\varphi_j(1-y_j)}e^{i\varphi'_j y_j} \bra{0}^{\otimes t-r}W\ket{\tilde{y}},
\end{equation}
and write $p$ as 
\begin{equation}
	p=\twonorm{\ket{\mu}}^2,\qquad \ket{\mu}:=\expval{Y\sim q}{\ket{\psi(Y)}}=\sum_y q(y) \ket{\psi(y)}.
\end{equation}

We thus see that the Born rule probability $p$ is given by the squared length of a vector $\ket{\mu}$ that is an expectation value over vectors $\ket{\psi(y)}$ distributed according to $q(y)$. The idea behind our algorithm is then to estimate this expectation value $\ket{\mu}$ using a mean $\ket{\widebar{\psi}}$ over $s$ samples:
\begin{equation}
	\label{eq:mean}
	\ket{\widebar{\psi}}=\frac{1}{s}\sum_{j=1}^s \ket{\psi_j},
\end{equation}
where each $\ket{\psi_j}$ takes the value $\ket{\psi(y)}$ with probability $q(y)$. More precisely, in order to obtain each sample we first generate a $t$-bit string $y$ bit by bit according to $q(y_j)$. This way we generate the state $\ket{\tilde{y}}$ with probability $q(y)$. We then evolve it by a Clifford $W$ and project on $\ket{0}^{\otimes t-r}$ to finally obtain $\ket{\psi(y)}$ with probability $q(y)$. The evolution and projection can be performed efficiently and we describe how to do it in the next step. Here, assuming that we have $s$ such samples, we bound the estimation error.

First, we note that by construction $\ket{\widebar{\psi}}$ is an unbiased estimator of $\ket{\mu}$. Next, we use the following lemma, the proof of which can be found in Supplemental Material Sec.~\ref{app:bound}, to upper-bound the norm of each $\ket{\psi(y)}$. 
\begin{lem}[Upper-bound for $\twonorm{\ket{\psi(y)}}^2$]
	\label{lem:psi(y)_bound}
	 For every elementary description of $p$, the corresponding vectors $\ket{\psi(y)}$ defined in Eq.~\eqref{eq:psi(y)} are unnormalised stabilizer states with the squared $l_2$-norm upper-bounded by the total stabiliser extent $\cextent$ of all states coming from reverse gadgetization of non-Clifford gates appearing in that elementary description:
	\begin{equation}
		\twonorm{\ket{\psi(y)}}^2\leq \cextent.
	\end{equation}
\end{lem}
It is very important to note that the above bound for $\twonorm{\ket{\psi(y)}}^2$ is general, i.e. independent of the particularities of a given quantum circuit. We do expect that stronger circuit-specific bounds can be efficiently computed, which would translate into improved run-times of the $\child$ algorithm. Now, the key technical tool that we will employ is the next lemma, proven in Supplemental Material Sec.~\ref{app:hoeffding}, which applies a concentration inequality for vector martingales given by Heyes~\cite{hayes2005large} to our setting.
\begin{lem}
	\label{lem:hoeffding}
	Let $N, s \in \bbn$ and $\set{\ket{\psi_j}}_{j\in [N]}$ be a set of $d$-dimensional vectors over $\bbc$ satisfying $\twonorm{\ket{\psi_j}}^2\leq m$. Moreover, let $q$ be a probability distribution over $[N]$ and define $\ket{\mu}$ as the $d$-dimensional vector over $\bbc$ that is the expectation of $\ket{\psi_X}$ with respect to the random variable $X$ with probability distribution $q$:
	\begin{align*}
		\ket{\mu}=\expval{X\sim q}{\ket{\psi_X}}=\sum_{j \in [N]} q_j \ket{\psi_j}\text{.}
	\end{align*}
	For $j\in [s]$, let $x_j\in [N]$ be independently sampled from the probability distribution $q$, and define a vector sample mean over $s$ samples by:
	\begin{align}\label{eq: mean phi}
		\ket{\widebar{\psi}}=\frac{1}{s}{\sum_{j=1}^s \ket{\psi_{x_j}}}\text{.}
	\end{align}
	Then, for all $\epsilon>0$:
	\begin{align}\label{hayes hoeffding twonorm}
		\bigpr{\twonorm{\ket{\widebar{\psi}}-\ket{\mu}}\geq \epsilon}\leq 2 e^2 \exp{\left(\frac{-{s}{\epsilon^2}}{2 (\sqrt{m}+\sqrt{p})^2}\right)}
	\end{align}
	and
	\begin{align}\label{hayes hoeffding onenorm}
		\bigpr{\sonenorm{\density{\widebar{\psi}}-\density{\mu}}\geq \epsilon}\leq 2 e^2 \exp{\left(\frac{-{s}{(\sqrt{p+\epsilon}-\sqrt{p})^2}}{2 (\sqrt{m}+\sqrt{p})^2}\right)}
	\end{align}
	where $p:=\twonorm{\ket{\mu}}^2$ and $\sonenorm{\cdot}$ is the Schatten 1-norm.
\end{lem}

The bound on the estimation error, leading to the exponential scaling of the run-time (measured by the number of steps $s$) with the total stabiliser extent, can now be given as a simple corollary of the above technical lemmas.

\begin{cor}[Upper-bound for estimation error] 
	\label{cor:bound}
	The mean vector $\ket{\widebar{\psi}}$ from Eq.~\eqref{eq:mean} satisfies
	\begin{align}
		\label{eq:bound}
		\bigpr{\abs{\twonorm{\ket{\widebar{\psi}}}^2-p}\geq \epsilon}\leq \delta, \quad \delta:=2 e^2 \exp{\left(\frac{-{s}{(\sqrt{p+\epsilon}-\sqrt{p})^2}}{2 (\sqrt{\cextent}+\sqrt{p})^2}\right)}.
	\end{align}
\end{cor}
\begin{proof}
	We first note that $\abs{\tr{A}}\leq \sonenorm{A}$ for any Hermitian operator $A$. This follows from the fact that $\tr{A}$ is the sum of the eigenvalues of $A$ while $\sonenorm{A}$ is the sum of the singular values of $A$. By applying this inequality to $A=\density{\widebar{\psi}}-\density{\mu}$, the result follows immediately from Eq.~\eqref{hayes hoeffding onenorm}, where $m$ can be replaced by $\cextent$ due to Lemma~\ref{lem:psi(y)_bound}.
\end{proof}

% -------------------------------------------------------------

\subsection{Step 2: State evolution}
\label{sec:evolution}

In the previous step we showed that by randomly sampling $s$ stabiliser states $\ket{\psi_j}$, each equal to $\ket{\psi(y)}$ with probability $q(y)$, and creating their uniform superposition $\ket{\widebar{\psi}}$, we can estimate $p$ by calculating $\twonorm{\ket{\widebar{\psi}}}^2$. Here, we will show how to efficiently obtain the description of each sampled state. It is clear from Eq.~\eqref{eq:psi(y)} that to find a given $\ket{\psi(y)}$ it is enough to find an efficient way of representing $\bra{0}^{\otimes t-r}W\ket{\tilde{y}}$ for every $y$. This step of the algorithm will consist of three parts: first, we will explain how to get $W\ket{\tilde{0}\dots\tilde{0}}$; then, how to modify this state to obtain $W\ket{\tilde{y}}$ for arbitrary $y$; and finally, how to perform post-selection to end up with $\bra{0}^{\otimes t-r}W\ket{\tilde{y}}$.

In the first part, we use the phase-sensitive Clifford simulator described in Ref.~\cite{bravyi2019simulation} to efficiently calculate the CH form of a $t$-qubit stabiliser state $W\ket{\tilde{0}\dots\tilde{0}}$. The CH form of a general $t$-qubit stablizer state $\ket{\sigma}$ can be described by a tuple \mbox{$\mathcal{T}(\sigma)=\{F,G,M,{\gamma},v,s,\omega\}$}. Here $F, G$ and $M$ are $t \times t$ binary matrices, $\gamma$ is a length $t$ vector with entries in $\mathbb{Z}_4$, ${v}$ and $s$ are binary vectors of length $t$, and $\omega$ is a complex number. Ref.~\cite{bravyi2019simulation} shows that for each gate $\Gamma\in \set{S, CX, CZ}$, the updated information $\mathcal{T}(\sigma')$ representing $\ket{\sigma'}=\Gamma \ket{\sigma}$ can be computed in $\order{t}$ elementary steps. Updates associated with each Hadamard gate can be computed in $\order{t^2}$ steps. Since $W$ is composed of $\order{t^2}$ elementary Clifford gates including $\order{t}$ Hadamard gates, in $\order{t^3}$ steps we can calculate the CH form of $W\ket{\tilde{0}\dots\tilde{0}}$ by updating the CH form of $\ket{\tilde{0}\dots\tilde{0}}$ step by step with every application of the elementary Clifford gates composing $W$. For completeness, we provide a more detailed introduction of the CH form in Supplemental Material Sec.~\ref{app:ch_form}. This first step can be performed as pre-computation, before any sampling of $y$ is started.

In the second part of this step, we employ a technique similar to Ref.~\cite{Qassim2019} in order to update the CH form of $W\ket{\tilde{0}\dots\tilde{0}}$ to get the CH form of $W\ket{\tilde{y}}$ after sampling a given $y$. If the $k^\text{th}$ bit of a bitstring $z$ is zero, and $y$ is the same bitstring with the $k^\text{th}$ bit set to one then $\ket{\tilde{y}} = S_k^3\ket{\tilde{z}}$. Hence $W \ket{\tilde{y}} = W S_k^3 W^\dagger W  \ket{\tilde{z}}$. In order to update the state $W\ket{\tilde{0}\dots\tilde{0}}$ to $W\ket{\tilde{y}}$, for arbitrary $y$ we therefore pre-compute the $t$ Clifford operators $W S_k^3 W^\dagger$. By writing $S^3_k = \frac{1}{\sqrt{2}}e^{-i\frac{\pi}{4}}\left(I + i Z_k\right)$ we can apply Lemma~4 of Ref.~\cite{bravyi2019simulation} to update a single bit of $y$ in time $\order{t^2}$. Transforming $W\ket{\tilde{0}\dots\tilde{0}}$ into $W\ket{\tilde{y}}$ for arbitrary $y$ therefore takes time $\order{t^3}$. \newtext{Although we could also obtain $W\ket{\tilde{y}}$ in time $\order{t^3}$ by starting from the state $\ket{0}$ and performing standard CH-form evolution we have observed that our method of updating the state from $W\ket{\tilde{0}\dots\tilde{0}}$ is faster in practice. This method also allows for further optimisations, since if a sampled bitstring $y$ is close to a previously sampled bitstring $y^\prime$ (in the sense of their sum having small Hamming-weight), then one may obtain $W\ket{\tilde{y}}$ by updating $W\ket{\tilde{y}^\prime}$, rather than starting afresh from $W\ket{\tilde{0}}$.}

In the third and final part, we need to transform the CH form of a $t$-qubit stabilizer state $W\ket{\tilde{y}}$ into an $r$-qubit stabiliser state $\bra{0}^{\otimes t-r}W\ket{\tilde{y}}$. The authors of Ref.~\cite{bravyi2019simulation} explained how, in $\order{t^2}$ steps, one can update the CH form of a given $t$-qubit state to simulate the action of a projector $\ketbra{0}{0}^{\otimes t-r}$. Surprisingly, despite the fact that the resulting unnormalised state, $\left(\ketbra{0}{0}^{\otimes t-r}\otimes I^{\otimes r}\right) W\ket{\tilde{y}}$, is a product state \newtext{(the first $t-r$ qubits are not entangled with the last $r$)}, it is  non-trivial to discard the $(t-r)$ measured qubits in the CH form description. This complication is related to the fact that the tuple $T(\sigma)$ corresponding to a stabiliser state $\ket{\sigma}$ is not unique. Thus, there exists a large number of equivalent tuples describing a given product state that do not admit a decomposition into two tuples representing each component of the tensor product. 

We have developed a new subroutine to address this issue. This result is summarised in the following lemma, the proof of which can be found in Supplemental Material Sec.~\ref{app:discard}.

\begin{lem}[Discarding systems in CH form]
	\label{lem:discard}
	Given a tuple describing the CH form of an $(n+1)$-qubit stabilizer state $\ket{0}\otimes\ket{\sigma}$, one can find the tuple describing the CH form of an $n$-qubit stabilizer state $\ket{\sigma}$ in $O(n^2)$ time.
\end{lem}

Using the above and given $y$, we can generate the CH form of a state $\ket{\psi(y)}$ in $\order{t^3}$ steps.

% -------------------------------------------------------------

\subsection{Step 3: Norm estimation}
\label{sec:norm}

We are now at the point that we have the description of a state $\ket{\widebar{\psi}}$ as a uniform superposition of $s$ stabiliser states $\ket{\psi_1},\dots\ket{\psi_s}$, and the squared $l_2$-norm of $\ket{\widebar{\psi}}$ is an estimate for the Born rule probability $p$. The goal of the final step is to find and estimate $\hat{p}$ for $\twonorm{\ket{\widebar{\psi}}}^2$ (so effectively for $p$), and bound the total estimation error by relating it to the run-time.

Given a vector $\ket{\widebar{\psi}}$ with a decomposition into an $s$-term linear combination of $r$-qubit stabilizer states $\ket{\psi_j}$ from Eq.~\eqref{eq:mean}, one can estimate the squared $l_2$-norm of $\ket{\widebar{\psi}}$ using a fast norm estimation algorithm by Bravyi and Gosset~\cite{bravyi2016improved}. As inputs, the algorithm is given the CH-forms of $\ket{\psi_j}$. Next, it generates $L$ randomly sampled $r$-qubit stabilizer states $\ket{\theta_1}, \ldots, \ket{\theta_L}$. The estimate $\hat{p}$ is then given by:
\begin{align}\label{eq:fast norm phat}
	\hat{p}=\frac{2^r}{s^2 L}\sum_{j=1}^L \abs{\sum_{k=1}^s \innerbraket{\theta_j}{\psi_k}}^2\text{.}
\end{align}
Each phase sensitive stabilizer inner product, $\innerbraket{\theta_j}{\psi_k}$, appearing above takes $\order{r^3}$ steps to evaluate, and so we need $\order{sLr^3}$ steps to evaluate $\hat{p}$. By choosing:
\begin{align}
	L=\left\lceil \tilde{\epsilon}^{-2}  \log \tilde{\delta}^{-1} \right\rceil\text{,}\label{eq: L bound}
\end{align}
we ensure that the estimate $\hat{p}$ has multiplicative precision, i.e., for any desired error level, $\tilde{\epsilon}>0$ and failure probability, $\tilde{\delta}>0$, we have
\begin{equation}
	\label{eq:norm_est}
	\operatorname{Pr}\left(\left|\hat{p}-\twonorm{\ket{\widebar{\psi}}}^2\right| \geq \tilde{\epsilon}\twonorm{\ket{\widebar{\psi}}}^2\right) \leq \tilde{\delta},
\end{equation}
with the run-time scaling as:
\begin{align}
	\order{s r^3 \ \tilde{\epsilon}^{-2}  \log \tilde{\delta}^{-1}}\text{.}\label{eq: fast norm estim run-time}
\end{align} 

We have now estimated the squared $l_2$-norm of $\ket{\mu}$, i.e. $p$, by an estimate of the squared $l_2$-norm of $\ket{\widebar{\psi}}$, i.e. $\hat{p}$. This introduced two sources of error. The first is due to the deviation between the two squared $l_2$-norms, and we have bounded this error in Corollary~\ref{cor:bound}. The second source of error is due to the deviation between the estimate $\hat{p}$ and the squared $l_2$-norm of $\ket{\widebar{\psi}}$. We now combine these two errors to show that, for an appropriate choice of $s$ and $L$, our estimate satisfies Eq.~\eqref{eq:thm3}. First, we can employ the triangle inequality to obtain
\begin{equation}
	|\hat{p}-p|=\left|\hat{p}-\twonorm{\ket{\widebar{\psi}}}^2+\twonorm{\ket{\widebar{\psi}}}^2-p\right|\leq \left|\hat{p}-\twonorm{\ket{\widebar{\psi}}}^2\right| + \left|\twonorm{\ket{\widebar{\psi}}}^2-p\right|.
\end{equation}
From Eq.~\eqref{eq:norm_est} we have that with probability larger than $1-\tilde{\delta}$ the following holds:
\begin{equation}
	\left|\hat{p}-\twonorm{\ket{\widebar{\psi}}}^2\right|\leq \tilde{\epsilon} \twonorm{\ket{\widebar{\psi}}}^2 \leq \tilde{\epsilon} \left(\left|\twonorm{\ket{\widebar{\psi}}}^2 -p\right|+p\right)=\tilde{\epsilon} \left(\epsilon+p\right).
\end{equation}
Then, from Eq.~\eqref{eq:bound} we get that with probability larger than $1-\delta$ we have 
\begin{align}
	\left|\twonorm{\ket{\widebar{\psi}}}^2 -p\right|\leq \epsilon.
\end{align}
Since both steps (computing the sample average vector $\ket{\widebar{\psi}}$ and computing $\hat{p}$ using fast norm estimation) are independent, we get that with probability larger than $(1-\delta)(1-\tilde{\delta})$ we have
\begin{equation}
	\abs{\hat{p}-p} \leq \tilde{\epsilon}(\epsilon+p)+\epsilon.
\end{equation}
We can thus write
\begin{equation}
	\label{eq:total-bound}
	\bigpr{\abs{\hat{p}-p} \geq \tilde{\epsilon}(\epsilon+p)+\epsilon}\leq \tilde{\delta}+\delta. 
\end{equation}

Introducing variables describing the total estimation error, $\epsilon_{\mathrm{tot}}>0$ and $\delta_{\mathrm{tot}}>0$, we want to find the bounds on the number of samples $s$ from Step~1 and on the number of repetitions $L$ from Step~3, so that the estimate $\hat{p}$ satisfies:
\begin{equation}
	\bigpr{\abs{\hat{p}-p} \geq \epsilon_{\mathrm{tot}}}\leq \delta_{\mathrm{tot}}\label{eq:target_accuracy}.
\end{equation}
Employing Eq.~\eqref{eq:total-bound}, together with Eqs.~\eqref{eq:bound}~and~\eqref{eq: L bound}, the above is satisfied whenever for any arbitrary choice of $\epsilon\in (0,\epsilon_{\mathrm{tot}})$ and $\delta\in (0,\delta_{\mathrm{tot}})$ we have
\begin{subequations}
\begin{align}
	s&\geq \frac{2 (\sqrt{\cextent} +\sqrt{p})^2}{{\left(\sqrt{p+\epsilon}-\sqrt{p}\right)}^{2}} \log~\brackn{\frac{2 e^2}{\delta}}\label{eq:number-of-samples},\\
	L&\geq \brackn{\frac{p+\epsilon}{\epsilon_{\mathrm{tot}}-\epsilon}}^2 \log\left(\frac{1}{\delta_{\mathrm{tot}}-\delta}\right) \label{eq: number of iterations of FNA}\text{.}
\end{align}
\end{subequations}

The output of the fast norm estimation algorithm $\hat{p}$ is the output of our $\child$ algorithm. The bound on the estimation error, Eq.~\eqref{eq:target_accuracy}, together with the bounds on $s$ and $L$, Eqs.~\eqref{eq:number-of-samples}~and~\eqref{eq: number of iterations of FNA}, are equivalent to the main statement of Theorem~\ref{thm:child alg}. To finish the proof, we need to show that the run-time is indeed as claimed in the theorem. To see this, recall that producing each of $s$ samples $\ket{\psi_j}$ (sampling $y$ in Step~1 and evolving the state in Step~2) takes $\order{t^3}$. Moreover, each sample has to go through $L$ repetitions of the norm estimation subroutine, with each repetition taking $\order{r^3}$ steps. Putting this all together, the run-time of the $\child$ algorithm is \mbox{$\order{s t^3 +sL r^3}$}, as claimed in Theorem~\ref{thm:child alg}.

We present the psuedocode for the $\child$ algorithm below.

%XXXXXXXXXXXXXXXXXXXXXXXXXXXXXXXXXXXXXXXXXXXXXXXXXXXXXXXXXXXXX
%CHILD
%XXXXXXXXXXXXXXXXXXXXXXXXXXXXXXXXXXXXXXXXXXXXXXXXXXXXXXXXXXXXX
\begin{algorithm}[H]
\caption{$\child$ outputs an estimate $\hat{p}$ as characterized by Theorem~\ref{thm:child alg}.}\label{algo:child}
\begin{algorithmic}[1]
\algrenewcommand\algorithmicrequire{\textbf{Input:}}
\algrenewcommand\algorithmicensure{\textbf{Output:}}
\Require Output of $\qcompress$ for an elementary description of $p$ and parameters $s,L\in \bbn$. 
\Ensure An estimate $\hat{p}$.
\Statex
\State $[r, v, W]\gets \qcompress(\mathcal{D})$ \Comment{$\mathcal{D}$ represents the elementary description of $p$.}
\State Compute $\chformvar{r}{W}{\tilde{0}}$\Comment{This is the CH-form of the state $\bra{0}^{\otimes t-r} W \ket{\tilde{0}}^{\otimes t}$.}
\For{$k\in [s]$}
	\State $y\gets Y$ where $Y \in \bitstring{t}$ is sampled according to the product distribution in Eq.~\eqref{eq: y distribution} 
	\State Compute $\ket{\psi_k} \gets \chformvar{r}{W}{\tilde{y}}$ \Comment{$\chformvar{r}{W}{\tilde{y}}$ is computed from $\chformvar{r}{W}{\tilde{0}}$ as per App.~\ref{sec:evolution}}.
\EndFor
%apply reg a constraints
\For{$j\in [L]$}
	\State $\ket{\theta_{j}}\gets $ random $r$-qubit equatorial state.
\EndFor
\State $\hat{p}\gets \overlapsvar{\set{\ket{\psi_k}}_{k\in [s]}}{\set{\ket{\theta_j}}_{j\in [L]}}$\Comment{The $\overlaps$ sub-procedure evaluates Eq.~\eqref{eq:fast norm phat}.}
\State \textbf{return} {$\hat{p}$}
\end{algorithmic}
\end{algorithm}
%XXXXXXXXXXXXXXXXXXXXXXXXXXXXXXXXXXXXXXXXXXXXXXXXXXXXXXXXXXXXXXXXX

Finally, we note that both $s$ and $L$ depend on the unknown quantity $p$ and increase with $p$. Thus if we require a bound on our estimate's failure probability, we can make a conservative choice by substituting $p=1$ into Eqs.~\eqref{eq:number-of-samples} and~\eqref{eq: number of iterations of FNA}. Instead of this naive approach, the $\parent$ algorithm allows us to significantly improve the run-time by making a less conservative choice of $p$, thus taking advantage of the substantial run-time improvements that occur for smaller $p$ values.

% -------------------------------------------------------------
% APPENDIX D - THE ESTIMATE ALGORITHM
% -------------------------------------------------------------

\section{The {\normalfont\texorpdfstring{$\main$}{Estimate}} algorithm}
\label{app:parent}

The role of the $\parent$ algorithm is to choose the parameters used in making repeated calls to the $\child$ algorithm with the goal of finally attaining a Born rule probability estimate $\hat{p}$ satisfying Eq.~\eqref{thm:main}. In what follows, we first present the overview of the algorithm and explain the key ideas behind it that crucially depend on the error function $\epsilon^*$. Then, we present a rigorous definition of $\epsilon^*$ and prove its properties necessary for the performance of $\parent$. Finally, we upper bound the run-time of $\parent$.

\subsection{Overview}

In each call to the $\child$ algorithm, the $\parent$ algorithm chooses the input parameters $s$ and $L$ based on the input parameters to $\parent$ (i.e., the desired level of additive error $\epsilon_{\mathrm{tot}}$ and the failure probability $\delta_{\mathrm{tot}}$), and the output Born rule probability estimates produced by $\child$ in prior rounds. Each round has a specified run-time budget $\mathcal{T}_k$, which limits the choice of parameters $s$ and $L$ by imposing $\modeltau(s,L)\leq \mathcal{T}_k$ on the run-time model defined in Eq.~\eqref{eq:model child run-time}. 
The budget allocation in the first step is $2\mathcal{T}_0$, where $\mathcal{T}_0$ is a budget that is insufficient (in the best case scenario of $p=0$) to satisfy Eq.~\eqref{thm:main}. Starting from this low initial run-time allocation, the budget doubles at each iteration. Thus, the run-time budget for each call of the $\child$ algorithm and the total run-time over all prior calls both grow exponentially in the round number.

Now, in each round $k$, given failure probability $\delta_k$ (yet to be specified), the algorithm computes the optimal choice of parameters $(s^*_k, L^*_k)$, such that they minimize the effective additive error of the output estimate from $\child$, while satisfying the run-time constraints, $\modeltau(s^*_k,L^*_k)\leq \mathcal{T}_k$, and the constraints on the failure probability (it has to be lower than $\delta_k$). This additive error depends on the unknown $p$ and will be denoted by the function $\epsilon^*(p,\delta_k,\mathcal{T}_k)$, with $\epsilon^*_k$ denoting its value in the $k^{\mathrm{th}}$ round. Thus, for the input parameters $(s^*_k,L^*_k)$, the Born rule probability estimate $\hat{p}_k$ obtained in the $k^{\mathrm{th}}$ call to $\child$ satisfies
\begin{align}
	\bigpr{\abs{\hat{p}_k-p}\geq \epsilon^{*}(p,\delta_{k},\mathcal{T}_k)}\leq \delta_{k}.\label{eq: failure prob using optimals}
\end{align}

As we prove in the next subsection, the function $\epsilon^*$ is monotonically increasing in $p$. Therefore, Eq.~\eqref{eq: failure prob using optimals} is still satisfied if we replace the unknown $p$ appearing in $\epsilon^*(p,\delta_k,\mathcal{T}_k)$ by a known upper bound. For that, we will be using a \emph{probabilistic} upper bound $p^*_k$ for $p$. Starting with a trivial upper bound $p^*_0=1$, we can use the estimates $\hat{p}_k$ to compute subsequent probabilistic upper bounds for $p$ in the following way: 
\begin{align}
    p^{*}_{k}:=\hat{p}_{k}+\epsilon^{*}(p^{*}_{k-1},\delta_{k},\mathcal{T}_k).
\end{align}
Thus, with each call to $\child$, the $\parent$ algorithm is able to learn tighter and tighter upper bounds for the target Born rule probability. The halting condition for $\parent$ is satisfied when $\epsilon^{*}(p^{*}_{k-1},\delta_{k},\mathcal{T}_k)\leq\epsilon_{\rm tot}$ is first satisfied.

Of course, each of the above probabilistic upper bounds can fail with some small probability. However, as we now show, a proper choice of $\delta_k$ guarantees that the desired accuracy parameters are satisfied by the final round's Born rule probability estimate. Note that $p^{*}_1$ is a probabilistic upper bound of $p$ with failure probability $\delta_1$. If $p^{*}_k$ is a probabilistic upper bound of $p$ with failure probability $\delta_1+\ldots +\delta_k$, then, by the union bound, $p^{*}_{k+1}$ is a probabilistic upper bound for $p$ with failure probability $\delta_1+\ldots+\delta_{k+1}$. For the choice $\delta_k=\frac{6}{\pi^2 k^2}\delta_{\rm tot}$, the infinite sum $\delta_1+\delta_2+\ldots$ converges to $\delta_{\rm tot}$, and hence the probability that at least one of the upper bounds, $p^{*}_{1}, p^{*}_{2} \ldots$, fails is at most $\delta_{\rm tot}$.

We provide the pseudocode of the $\parent$ algorithm below, while Fig.~\ref{fig:optimize-perf} shows the output and intermediate values of $\hat{p}, p^{*}$ and $\epsilon^*$ produced using the $\main$ algorithm. Note that some quantities appearing in the pseudocode will be rigorously introduced in the next subsection. In particular, $\textsc{OptParams}$ denotes the subroutine that finds the parameters $(s_k^*,L_k^*)$ which optimize the run-time cost while achieving the minimum error.

%XXXXXXXXXXXXXXXXXXXXXXXXXXXXXXXXXXXXXXXXXXXXXXXXXXXXXXXXXXXXX
%PARENT
%XXXXXXXXXXXXXXXXXXXXXXXXXXXXXXXXXXXXXXXXXXXXXXXXXXXXXXXXXXXXX
\begin{algorithm}[H]
\caption{$\parent$ returns the estimate $\hat{p}$ as characterized by Eq.~\eqref{thm:main}.}
\label{algo:parent}
\begin{algorithmic}[1]
\algrenewcommand\algorithmicrequire{\textbf{Input:}}
\algrenewcommand\algorithmicensure{\textbf{Output:}}
\Require  Output of $\qcompress$ for an elementary description of $p$ and accuracy parameters $\epsilon_{\rm tot}, \delta_{\rm tot}>0$.	
\Ensure An estimate $\hat{p}$.
\Statex
\State $[r, v, W]\gets \qcompress(\mathcal{D})$ \Comment{$\mathcal{D}$ represents the elementary description of $p$.}
\State $p^{*}_0\gets 1$\Comment{This is the running upper bound for the unknown $p$.}
\State ${\rm Exit}\gets 0$, $k\gets 0$
\State $\mathcal{T}_0\gets \modeltau(-\frac{2(\sqrt{\ubex}+1)^2}{\epsilon_{\rm tot}} \ln \frac{\delta_{\rm tot}}{2 e^2},1)$\label{alg:s0 L0 defin}\Comment{$\mathcal{T}_{0}$ is a run-time budget that is too small to satisfy Eq.~\eqref{thm:main} even assuming $p=0$.}
\vspace{0.05cm}
\While{${\rm Exit}=0$}
	\State $k\gets k+1$
	\State $(\eta,s,L^+)\gets \textsc{OptParams}(p^{*}_{k-1},\frac{6}{\pi^2 k^2}\delta_{\rm tot},2^{k}\mathcal{T}_0)$\Comment{In each round, we double the run-time budget.}
	\State $\epsilon^{*}_k\gets \epsilon'(p^{*}_{k-1},\frac{6}{\pi^2 k^2}\delta_{\rm tot},\eta,s,L^+)$\label{alg:eps dash}\Comment{Equivalently, $\epsilon^{*}_k\gets \epsilon^{*}(p^{*}_{k-1},\frac{6}{\pi^2 k^2}\delta_{\rm tot},2^{k}\mathcal{T}_0)$.}
	\If{$\epsilon^{*}_k\leq \epsilon_{\rm tot}$} 
		\State ${\rm Exit}\gets 1$
	\EndIf
	\State $\hat{p}_k\gets \childvar{s}{L^+ +L_{\min}(\frac{6}{\pi^2 k^2}\delta_{\rm tot},\eta)}$\label{alg:child call}
	\State $p^{*}_k\gets \max \set{0,\min \set{1, p^{*}_{k-1},\hat{p}_k+\epsilon^{*}_k}}$
\EndWhile
\State \textbf{return} {$\hat{p}_k$}
\end{algorithmic}
\end{algorithm}

%XXXXXXXXXXXXXXXXXXXXXXXXXXXXXXXXXXXXXXXXXXXXXXXXXXXXXXXXXXXXXXXXX

\subsection{Definition and properties of \texorpdfstring{$\epsilon^*$}{Epsilon-star}}

Consistent with Eq.~\eqref{eq:child run-time}, we model the run-time of $\child$ using Eq.~\eqref{eq:model child run-time} \newtext{which we repeat here for convenience:
\begin{equation}
    \tag{\ref{eq:model child run-time}}
    \modeltau(s,L):=c_1 s t^3 + c_2 s L r^3,
\end{equation}}
Hence the $\parent$ algorithm aims to minimize the run-time cost, $\mathcal{C}$, as defined in Eq.~\eqref{eq:child cost}. For $p\in [0,1]$, $\epsilon_{\rm tot}\in \bbr^+$, $\eta\in (0,1)$ and $s,L\in \bbn^+$, we define the function: 
\begin{align}\label{eq:delta dash}
	\delta'(p,\epsilon_{\rm tot},\eta,s,L):= 2e^2 {\rm exp}\left(\frac{-s(\sqrt{p+\eta \epsilon_{\rm tot}}-\sqrt{p})^2}{2(\sqrt{\ubex}+1)^2} \right)+\exp\left(-\left(\frac{(1-\eta)\epsilon_{\rm tot}}{p+\eta \epsilon_{\rm tot}}\right)^2 L \right).
\end{align}
Comparing to Eq.~\eqref{eq:thm3}, we note that $\childvar{s}{L}$ outputs an estimate $\hat{p}$ such that for all $\epsilon_{\rm tot}>0$ and \mbox{$\eta \in (0,1)$}:
\begin{align}\label{eq:child concentration}
	\bigpr{\abs{\hat{p}-p}\geq \epsilon_{\rm tot}} \leq \delta'(p,\epsilon_{\rm tot},\eta,s,L).
\end{align}
For fixed $p,\eta,s,L$, we want to view $\delta'$ as a function of $\epsilon_{\rm tot}$ and define its functional inverse. We will need this to be defined for all $\delta'>0$. By inspection of Eq.~\eqref{eq:delta dash}, we note that for $\delta_{\rm targ}>0$ close to zero and $L,\eta$ too small there does not exist $\epsilon'$ such that $\delta'(p,\epsilon',\eta,s,L,m)=\delta_{\rm targ}$. To resolve this technicality, we define a minimal $L$ value:
\begin{align}\label{eq:Lmin}
	L_{\min}(\delta,\eta):=\left\lceil {-\left(\frac{\eta}{(1-\eta)}\right)^2 \ln \delta}\right\rceil.		
\end{align}
To specify a well defined inverse $\epsilon'(p,\delta_{\rm targ},\eta,s,L)\in \bbr^+$ of the $\delta'$ function, let us define its domain
\begin{align}
	D=[0,1]\times(0,1)\times(0,1)\times \bbn^+ \times \bbn^+.
\end{align}	
By inspecting Eq.~\eqref{eq:delta dash}, it is clear that on $D$, there exists a well defined function	$\epsilon'$ that satisfies the following: for all $(p,\delta_{\rm targ},\eta,s,L^{+})\in D$, there exists $\epsilon_{\rm targ}=:\epsilon'(p,\delta_{\rm targ},\eta,s,L^{+})$ such that $\delta'(p,\epsilon_{\rm targ},\eta,s,L_{\min}(\delta_{\rm targ}, \eta) + L^{+})=\delta_{\rm targ}$. To see this, we just note that for $p\in [0,1]$ fixed, the function \mbox{$f(c)=\sqrt{p+c}-\sqrt{p}$} is strictly increasing and unbounded from above. 

\newtext{We now establish a property of the function $\epsilon'(p,\delta_{\rm targ},\eta,s,L^{+})$ that will be useful. Namely, we now show that this function is a monotonically increasing function of $p$. To see this we note that by the definition of the function $\epsilon'$, the evaluation of $\delta'(p,\epsilon'(p,\delta_{\rm targ},\eta,s,L^{+}),\eta,s,L_{\min}(\delta_{\rm targ}, \eta) + L^{+})$ is a constant. Thus, we have:
\begin{align}
	0&=d \delta'(p,\epsilon',\eta,s,L_{\min}(\delta,\eta))=\frac{\partial \delta'}{\partial p}dp+ \frac{\partial \delta'}{\partial \epsilon'}d \epsilon',
\end{align}
where we have omitted terms containing $d\eta$, $d s$, $d \delta_{\rm targ}$ and $d L^{+}$ as we are interested in the case were $\eta$, $s$, $\delta_{\rm targ}$ and $L^{+}$ are held constant.
Using Eq.~\eqref{eq:delta dash}, it is easy to verify that $\frac{\partial \delta'}{\partial p}\geq 0$ and $\frac{\partial \delta'}{\partial \epsilon'}\leq 0$. Thus $\frac{d \epsilon'}{d p}\geq 0$ as claimed.}

For $\mathcal{T}\in [2,\infty)$ and all other ranges as before, we can now define the function $\epsilon^{*}(p,\delta_{\rm tot},\mathcal{T})\in \bbr$ as:
\begin{align}\label{eq:eps star as inf}
	\epsilon^{*}(p,\delta_{\rm tot},\mathcal{T})=\underset{\eta, s, L^+}{\inf}~ \epsilon'(p,\delta_{\rm tot},\eta,s,L^+),
\end{align}
where the infimum is over all $\eta\in (0,1)$, $s,L^+ \in \bbn^+$ subject to the constraint:
\begin{align}
	\modeltau(s,L^+ +L_{\min}(\delta_{\rm tot},\eta)) \leq \mathcal{T}.\label{eq:esp inf constraint}
\end{align}
Since the range $\eta\in (0,1)$ is not closed, in principal the function $\epsilon'$ could get arbitrarily close to its infimum $\epsilon^{*}$ without attaining it. However, one can show that for all $p\in [0,1]$, $\delta_{\rm tot}\in (0,1)$ and $\mathcal{T}\geq 2$, there always exists $\eta\in (0,1)$ and $s,L^+ \in \bbn^+$ such that the infimum is attained, i.e. $\epsilon'(p,\delta_{\rm tot},\eta,s,L^+)=\epsilon^{*}(p,\delta_{\rm tot},\mathcal{T})$. To see this, we note that Eqs.~\eqref{eq:Lmin} and \eqref{eq:esp inf constraint} can be used to impose a closed upper bound on $\eta$. Similarly, using the fact that in the limit of $\eta \rightarrow 0$, $\delta'(p,\epsilon_{\rm tot},\eta,s,L)$  becomes greater than $1$, we can impose a closed lower bound on $\eta$. Having restricted the range of $\eta$ in Eq.~\eqref{eq:eps star as inf} to a closed interval contained in $(0,1)$, we can apply the Extreme Value Theorem to prove our claim.

We now show that for $\delta_{\rm tot}\in (0,1]$ and $\mathcal{T}\geq 2$ fixed, $\epsilon^{*}(p,\delta_{\rm tot},\mathcal{T})$ is monotonically increasing in $p$. Let $\eta',s',L'$ and $\eta'',s'',L''$ be such that $\epsilon^{*}(p',\delta_{\rm tot},\mathcal{T})= \epsilon'(p',\delta_{\rm tot},\eta',s',L')$ and $\epsilon^{*}(p'',\delta_{\rm tot},\mathcal{T})= \epsilon'(p'',\delta_{\rm tot},\eta'',s'',L'')$. Then, for $p'\leq p''$, we have:
\begin{align}
	\epsilon^{*}(p',\delta_{\rm tot},\mathcal{T})&= \epsilon'(p',\delta_{\rm tot},\eta',s',L')\nonumber\\
	&\leq \epsilon'(p',\delta_{\rm tot},\eta'',s'',L'')\nonumber\\
	&\leq \epsilon'(p'',\delta_{\rm tot},\eta'',s'',L'')\nonumber\\
	&=\epsilon^{*}(p'',\delta_{\rm tot},\mathcal{T}),
\end{align}
where the first inequality holds by the definition of $\epsilon^*$ and the second inequality holds by the fact that $\epsilon'$ is monotone increasing in $p$.

%------------------------------------------------

Having the rigorous definitions, we now see that from Eq.~\eqref{eq:thm3} and the definitions of $\epsilon{'}$ and $\epsilon^{*}$, we have that for all $\delta_{k}>0$ and $\eta, s, L^+$ satisfying \mbox{$\epsilon^{*}(p,\delta_{k},\mathcal{T}_k)= \epsilon'(p,\delta_{k},\eta,s,L^{+})$}, the output $\hat{p}$ of $\childvar{s}{L_{\min}(\delta_{k},\eta)+L^{+}}$ satisfies Eq.~\eqref{eq: failure prob using optimals}. Also, for $(p,\delta_{\rm tot},\mathcal{T})\in [0,1]\times \bbr^+ \times [2,\infty)$, we define the function $\textsc{OptParams}(p,\delta_{\rm tot},\mathcal{T})\in(0,1) \times \bbn^+ \times \bbn^+$. For a given target cost $\mathcal{T}$, this function numerically optimises the choice of parameters $\eta$, $s$ and $L^{+}$ (subject to the cost budget constraint) such that $\epsilon^{*}(p,\delta_{\rm tot},\mathcal{T})=\epsilon'(p,\delta_{\rm tot},\eta,s,L^+)$. Finally, we note that small improvements in performance can be achieved by using a larger choice of $\mathcal{T}_0$ subject to the requirement that $\mathcal{T}_{0}$ is still too small to satisfy Eq.~\eqref{thm:main} even assuming $p=0$.

\subsection{Upper-bounding the run-time}

We now move on to computing an upper bound for the total modeled run-time cost associated with $\main$ as defined in Eq.~\eqref{eq:child cost}. The run-time cost of $\main$ is probabilistic and dependent on the unknown $p$. Here, we introduce our $\mainruntime$ algorithm which produces run-time cost upper bounds for any given $p$ value. This algorithm can be used to produce a run-time cost upper bound as a function of $p$.  We point out that the actual run-time of $\main$ may differ from the run-time cost for a number of reasons. Firstly, run-time cost is a modelled and/or expected run-time and may differ from actual run-time it aims to predict due to limitations of the model or incorrectly calibrated model parameters $c_1$ and $c_2$. Secondly, the run-time cost only aims to model the total run-time of $\child$ over all calls made in Step~\ref{alg:child call} of Algorithm~\ref{algo:parent}. Thus, it ignores the run-time incurred by $\main$ in all other steps. We justify the choice to not model the run-time cost associated with these other steps since their run-time is insensitive to circuit parameters. 

We now present the pseudo-code for our $\mainruntime$ algorithm. We note that all steps except Steps~\ref{alg: run-time child call step} and~\ref{alg:run-time return step} are identical to the pseudo-code for $\main$.

%XXXXXXXXXXXXXXXXXXXXXXXXXXXXXXXXXXXXXXXXXXXXXXXXXXXXXXXXXXXXX
%run-time
%XXXXXXXXXXXXXXXXXXXXXXXXXXXXXXXXXXXXXXXXXXXXXXXXXXXXXXXXXXXXX
\begin{algorithm}[H]
\caption{$\mainruntime$ returns a probabilistic upper bound of $\mathcal{C}$, the run-time cost defined in Eq.~\eqref{eq:child cost}.}
\label{algo:mainrun-time}
\begin{algorithmic}[1]
\algrenewcommand\algorithmicrequire{\textbf{Input:}}
\algrenewcommand\algorithmicensure{\textbf{Output:}}
\Require  Assumed value of $p$; $\delta_{\rm UB}>0$, the required maximum failure probability of the probabilistic upper bound for $\mathcal{C}$; the output of $\qcompress$; and accuracy parameters $\epsilon_{\rm tot}, \delta_{\rm tot}>0$.
\Ensure The probabilistic upper bound $\mathcal{C}_{\rm UB}=\mathcal{C}_{\rm UB}(p)$.
\Statex
\State $[r, v, W]\gets \qcompress(\mathcal{D})$ \Comment{$\mathcal{D}$ represents the elementary description of $p$.}
\State $p^{*}_0\gets 1$\Comment{This is the running upper bound for the unknown $p$.}
\State ${\rm Exit}\gets 0$, $k\gets 0$
\State $\mathcal{T}_0\gets \modeltau(-\frac{2(\sqrt{\ubex}+1)^2}{\epsilon_{\rm tot}} \ln \frac{\delta_{\rm tot}}{2 e^2},1)$\label{alg:s0 L0 defin run-time}\Comment{$\mathcal{T}_{0}$ is a run-time budget that is too small to satisfy Eq.~\eqref{thm:main} even assuming $p=0$.}
\vspace{0.05cm}
\While{${\rm Exit}=0$}
	\State $k\gets k+1$
	\State $(\eta,s,L^+)\gets \textsc{OptParams}(p^{*}_{k-1},\frac{6}{\pi^2 k^2}\delta_{\rm tot},2^{k}\mathcal{T}_0)$\Comment{In each round, we double the run-time budget.}\label{alg:run-time optc step}
	\State $\epsilon^{*}_k\gets \epsilon'(p^{*}_{k-1},\frac{6}{\pi^2 k^2}\delta_{\rm tot},\eta,s,L^+)$\label{alg:eps dash run-time}\Comment{Equivalently, $\epsilon_k^{*}\gets \epsilon_k^{*}(p^{*}_{k-1},\frac{6}{\pi^2 k^2}\delta_{\rm tot},2^{k}\mathcal{T}_0)$.}
	\If{$\epsilon^{*}_k\leq \epsilon_{\rm tot}$} 
		\State ${\rm Exit}\gets 1$
	\EndIf
	\State $\hat{p}_k\gets p+\epsilon'\left(p,\delta_{\rm UB}/K_{\rm UUB},\tilde{\eta},s,L^+ +L_{\min}\left(\frac{6}{\pi^2 k^2}\delta_{\rm tot},\eta\right)-L_{\rm min}(\delta_{\rm UB}/K_{\rm UUB},\tilde{\eta})\right)$\label{alg: run-time child call step}\Comment{The choice of $K_{\rm UUB}>0$ and $\tilde{\eta}\in (0,1)$ are discussed below.} 
	\State $p^{*}_k\gets \max \set{0,\min \set{1, p^{*}_{k-1},\hat{p}_k+\epsilon^{*}_k}}$
\EndWhile
\State \textbf{return} {$\mathcal{C}_{\rm UB} \gets 2^{k+1}\mathcal{T}_0$}\label{alg:run-time return step}
\end{algorithmic}
\end{algorithm}
%XXXXXXXXXXXXXXXXXXXXXXXXXXXXXXXXXXXXXXXXXXXXXXXXXXXXXXXXXXXXXXXXX

To establish the correctness of our $\mainruntime$ algorithm, we first establish some notation. The $\parent$ algorithm generates the following strings of random variables: $\set{\hat{p}_k}_{k\in [K]}$, $\set{{\epsilon}_k^{*}}_{k\in [K]}$ , $\set{{p}_k^{*}}_{k\in [K]}$ and the string of triples $\set{(\eta_k, s_k, L_k^+)}_{k\in [K]}$ where $K$ is itself a random variable indicating when the exit condition is triggered. The exit condition is triggered when
                                                                                                                                                                   
\begin{align}
	\epsilon^{*}_k:=\epsilon^{*}\left(p^{*}_{k-1},\frac{6}{\pi^2 k^2}\delta_{\rm tot}, 2^{k}\mathcal{T}_0\right)\leq \epsilon_{\rm tot}
\end{align}
for the first time. The lowest value of $k$ for which the exit condition is triggered defines the random variable $K$. 

Since $\modeltau(s_k, L_k)\leq 2^k \mathcal{T}_0$, where $L_k=L_k^+ +L_{\min}(\frac{6}{\pi^2 k^2}\delta_{\rm tot},\eta_k)$, the total cost associated with calls to the $\child$ algorithm as modelled by Eq.~\eqref{eq:child cost} is upper bounded by:
\begin{align}
	\mathcal{C}\leq (2+2^2+\ldots+2^K)\mathcal{T}_0 < 2^{K+1}\mathcal{T}_0.
\end{align}
We note that the run-time cost is a random variable that depends on $K$. We will show that $K_{\rm UB}$, the value of $k$ used in Step~\ref{alg:run-time return step} of the $\mainruntime$ pseudo-code, is a probabilistic upper bound for $K$ and hence $\mathcal{C}\leq \mathcal{C}_{\rm UB}$ with probability larger than $1-\delta_{\rm UB}$.

We note that the $\mainruntime$ algorithm is deterministic. In the $\main$ algorithm, the randomness of the variables $K$, ${\epsilon}_k^{*}$, ${p}_k^{*}$ is entirely due to their functional dependence on $\set{\hat{p}_k}_{k\in [K]}$. The stochastic $\hat{p}_k$ used in Step~\ref{alg:child call} of the $\main$ algorithm are replaced with deterministic $\hat{p}_k$ in Step~\ref{alg: run-time child call step} of the $\mainruntime$ algorithm. Thus, the associated strings of variables generated by the $\mainruntime$ algorithm are all deterministic. 

Let $\v{p}=\set{p_k}_{k\in \bbn^+}$ be a sequence of probabilities $p_k\in [0,1]$. Then we will use $K(\v{p})$ and $\set{{\epsilon^{*}}_k(\v{p})}_{k\in [K(\v{p})]}$ to denote the values computed by $\parent$ in the setting when the $\child$ algorithm's output is forced to be exactly the sequence $\v{p}$. For \mbox{$k=1,2,\ldots$}, the variable $\epsilon_k^{*}$ computed in Line~\ref{alg:eps dash} can be specified by the recursion:
\begin{align}\label{eq:eps star recursive}
	\epsilon_k^{*}(\v{p})=\epsilon^{*}\left(\max \set{0,\min \set{1,p_0+\epsilon_{0}^{*}(\v{p}),\ldots, p_{k-1}+\epsilon_{k-1}^{*}(\v{p})}},\frac{6}{\pi^2 k^2}\delta_{\rm tot},2^k \mathcal{T}_0\right),
\end{align}
where $\epsilon_1^{*}(\v{p}):=\epsilon^{*}(1,6\delta_{\rm tot}/\pi^2,2 \mathcal{T}_0)$. From Eq.~\eqref{eq:eps star recursive} and that $\epsilon^*(p,\delta, \mathcal{T})$ is monotone increasing in $p$, it is clear that higher values of ${p}_{k-1}$ and ${\epsilon}_{k-1}^{*}$ both result in higher values of ${\epsilon}_k^{*}$. Thus, for some fixed $\v{p}$ with $p_1,\ldots, p_{k-1}$ sufficiently large, the deterministic quantity ${\epsilon}_k^{*}(\v{p})$ is a probabilistic upper bounded of the random variable ${\epsilon}_k^{*}={\epsilon}_k^{*}(\hat{p}_1, \hat{p}_2, \ldots)$ computed in the $\main$ algorithm. In particular, let us define $\v{p}=\set{p_k}_{k\in \bbn^+}$ as per Step~\ref{alg:run-time return step} of the $\mainruntime$ algorithm:
\begin{align}
	p_k:=p+\epsilon'\left(p,\delta_{\rm UB}/K_{\rm UUB},\tilde{\eta}_k,s_k,L_k^+ +L_{\min}\left(\frac{6}{\pi^2 k^2}\delta_{\rm tot},\eta_k\right)-L_{\rm min}(\delta_{\rm UB}/K_{\rm UUB},\tilde{\eta}_k)\right).\label{eq:p recur}
\end{align}
Here, $\eta_k, s_k$ and $L_k^+$ are parameters computed on the $k^{\rm th}$ iteration of Step~\ref{alg:run-time optc step} of $\mainruntime$ but $\tilde{\eta}_k$ is a new independent parameter. We will see that any choice of $\tilde{\eta}_k\in (0,1)$ will result in the desired upper bound and hence we can optimize the choice of $\tilde{\eta}_k$ to achieve a tighter upper bound. Although $K_{\rm UUB}$ must be chosen before $K_{\rm UB}$ can be computed, $K_{\rm UUB}$ can be any quantity that satisfies $K_{\rm UUB}\geq K_{\rm UB}$. Due to the weak dependence of $K_{\rm UB}$ on $K_{\rm UUB}$, such a choice is always possible.

We note that the deterministic quantity $p_k$ serves as a probabilistic upper bound of $\hat{p}_k\gets \childvar{s_k}{L_k^+ +L_{\min}(\frac{6}{\pi^2 k^2}\delta_{\rm tot},\eta_k)}$ such that the probability that $p_k$ fails to upper bound $\hat{p}_k$ for any $k\in \bbn^+$ is $\leq \delta_{\rm UB}/K_{\rm UUB}$. Further, since Eq.~\eqref{eq:child concentration} holds for all $\eta$, our statement holds for all choices of $\tilde{\eta}_k\in (0,1)$ subject to $L_k^+ +L_{\min}\left(\frac{6}{\pi^2 k^2}\delta_{\rm tot},\eta_k\right)-L_{\rm min}(\delta_{\rm UB}/K_{\rm UUB},\tilde{\eta}_k)\geq 1$. Thus, $K>K_{\rm UB}$ implies that there is a $\kappa\in [K_{\rm UB}]$ that is the smallest $k\in [K_{\rm UB}]$, such that $\hat{p}_k$ produced in Step~\ref{alg:child call} of $\main$ exceeds $\hat{p}_k$ produced in Step~\ref{alg: run-time child call step} of $\mainruntime$. By the union bound and our choice of $p_k$ the probability of this happening is $\leq K_{\rm UB} \delta_{\rm UB}/K_{\rm UUB}$.

This implies that:
\begin{align}
	\bigpr{K \leq K_{\rm UB} }\geq 1-\delta_{\rm UB}.
\end{align}
We note that before any costly calls to the $\child$ algorithm are made, $\mathcal{C}_{\rm UB}$ can easily by computed and plotted as a function of $p$ thus predicting probabilistic run-time upper bounds conditional on the unknown $p$. A similar plot is presented in Fig.~\ref{fig:run-time-bounds} for the probabilistic upper bound of run-time cost $\mathcal{C}$.

Finally, we note that since the functions $\epsilon'(p,\delta_{\rm tot},\eta,s,L-L_{\min}(\delta_{\rm tot},\eta))$ and $\textsc{OptParams}(\mathcal{T},p,\delta_{\rm tot})$ are not given in a closed form, their evaluation requires using numerical techniques. These will inevitably be subject to small levels of imprecision with run-times that mildly (logarithmically or poly-logarithmically) depend on precision requirements. 

In principle, the run-time for the numerical evaluation of these functions depends on the Born probability estimation problem parameters such as $\epsilon_{\rm tot}$ since the precision parameters must be $\ll \epsilon_{\rm tot}$. In practice, the precision parameters are so small that $\epsilon_{\rm tot}$ values of this order would produce completely infeasible run-times for the $\child$ algorithm. Thus, we ignore such run-time dependencies and treat the evaluation of these functions as having a fixed cost independent of the estimation problem parameters.

%%%%%%%%%% Merge with supplemental materials %%%%%%%%%%
\pagebreak
\widetext
\renewcommand{\appendixname}{Supplemental Material}

\appendix
\begin{center}
  \textbf{\large Supplemental Material}
\end{center}
%%%%%%%%%% Merge with supplemental materials %%%%%%%%%%
%%%%%%%%%% Prefix a "S" to all equations, figures, tables and reset the counter %%%%%%%%%%
\setcounter{equation}{0}
\setcounter{figure}{0}
\setcounter{table}{0}
\setcounter{page}{1}
\makeatletter
\renewcommand{\theequation}{S\arabic{section}.\arabic{equation}}
\renewcommand{\thefigure}{S\arabic{figure}}
\renewcommand{\thefigure}{S\arabic{figure}}
\setcounter{section}{0}
\renewcommand{\thesection}{\arabic{section}} 
\renewcommand{\theHsection}{appendixsection.\Alph{section}}
%%%%%%%%%% Prefix a "S" to all equations, figures, tables and reset the counter %%%%%%%%%%

% -------------------------------------------------------------
% SEC. I - ESTIMATING PAULI EXPECTATIONS VALUES
% -------------------------------------------------------------

\section{Estimating Pauli expectation values}
\label{sec:supp-pauli-expectation-values}

Our algorithms may also be used to compute the expectation value $\bra{0}^{\otimes n} U^\dagger P U \ket{0}^{\otimes n}$ of an $n$ qubit Pauli operator
\begin{align}
    P &= \omega P_1 \otimes P_2 \otimes \hdots \otimes P_n,
\end{align}
where $\omega\in\mathbb{C}$ and each $P_i \in \{I, X, Y, Z\}$. If every $P_i = I$ then the expectation value is $\omega$ so we assume at least one of the $P_i$ is not equal to the identity.

We first show that if there is a qubit $i$ such that $P_i\neq I$ there is a Clifford unitary $C$ such that
\begin{align}
    C^\dagger Z_1 C = P,
\end{align}
where $Z_1$ is the Pauli-$Z$ operator on the first qubit. If $a\neq 0$ we first apply the operator $X_1$ and use the equality $X_1 Z_1 X_1 = -Z_1$. Now let $j$ be the qubit number of the first non-identity $P_i$ making up $P$. If $j\neq 1$ apply a swap gate to qubits $1$ and $j$ to obtain
\begin{align}
    \pm\SWAPgate{1}{j}^\dagger Z_1 \SWAPgate{1}{j} = \pm Z_j.
\end{align}
Now for each $k > j$ such that $P_k\neq I$ we apply a $\CX$ controlled on qubit $k$ and targeted on qubit $j$. We obtain a tensor product operator consisting of single qubit identity and Pauli-$Z$ operators with Pauli-$Z$ operators on exactly the qubits for which $P$ is not the identity
\begin{align}
    \pm\left(\prod_{k> j: P_k\neq I} \CXgate{k}{j}\right)^\dagger Z_j\left(\prod_{k> j: P_k\neq I} \CXgate{k}{j}\right) &= \pm\prod_{k: P_k\neq I} Z_k.
\end{align}
It is easy to transform the $Z_k$ into the required $P_k$ applying single qubit phase and Hadamard gates
\begin{align}
    H_k^\dagger Z_k H_k &= X_k\\
    S_k H_k^\dagger Z_k H_k S_k^\dagger &= Y_k.
\end{align}

The expectation value may be calculated as $\bra{0}^{\otimes n} U^\dagger P U \ket{0}^{\otimes n} =\bra{0}^{\otimes n} (CU)^\dagger Z_1 CU \ket{0}^{\otimes n} = 2 p - 1$ where $p$ is the probability of obtaining the outcome $1$ from a computational basis measurement on the first qubit. 

We note that performance may be improved since the probability distribution is binary. We are free to estimate either $p$ or $1-p$ to obtain an estimate of the expectation value. A modified version of the $\parent$ algorithm may be employed which switches to estimating the probability of the other outcome each time it estimates a probability which is greater than $0.5$.

% -------------------------------------------------------------
% SEC. II - PROOF OF LEMMA 4
% -------------------------------------------------------------

\section{Proof of Lemma~\ref{lem:extract}}
\label{sec:extract detailed}

We define the following useful form for a generating set .
\begin{defn}
\label{def:zx_form}
For $G=\set{g_1,\ldots,g_k}\in \genset{n}{k}$ and $j\in [n]$, we say that $G$ is in $ZX(j)$-form iff at most two generators act non-trivially (i.e. as a Pauli $X$, $Y$ or $Z$) on qubit $j$.  Further, if exactly two generators act non-trivially on the $j^{\rm th}$ qubit then, one acts as a Pauli $X$ and the other as a Pauli $Z$.
\end{defn} 
When $G$ is in $ZX(j)$-form, we call any generator $g\in G$ a \emph{leading} generator if $\pfac{g}{j}\neq I$. We call $g$ the \emph{leading-X} generator if $\pfac{g}{j}= X$ and similarly for $Y$ and $Z$. Two generating sets $G$ and $G'$ are equivalent iff they generate the same stabilizer group i.e. $\gen{G}=\gen{G'}$. If $g_1,g_2 \in G$ are distinct then replacing $g_2$ by $g_1 g_2$ produces an equivalent generating set. By repeatedly using this method, any generating set $G$ can be transformed to an equivalent generating set $G'$ such that $G'$ is in $ZX(j)$-form for any suitable choice of $j$.

We now present a Lemma that will be useful for the proof of Lemma~\ref{lem:extract}.

\begin{lem}\label{lem: independence under fx}
Let $n\in \set{2,3,\ldots}$ and $k\in \set{0, 1,\ldots,n-2}$. Let $G=\set{g_z, g_x, g_1,\ldots,g_k}\in \genset{n}{k+2}$ be a stabilizer generating set in $ZX(j)$-form such that it has both a leading-$Z$ and a leading-$X$ generator, $g_z$ and $g_x$ respectively. For $a\in \set{0,1}$ fixed, let $\tilde{g}_z:=\bra{a}g_z\ket{a}$ and for $i\in [k]$, $\tilde{g}_i:=\bra{a}{g}_i\ket{a}$ where the $\ket{a}$ vectors act on the $j^{\rm th}$ qubit. Then $\tilde{G}:=\set{\tilde{g}_z, \tilde{g}_1,\ldots,\tilde{g}_k}\in \genset{n-1}{k+1}$.
\end{lem}

\begin{proof}
It is clear that $\tilde{G}\subset \pauli{n-1}$ is commuting so we only need to show that the subset is independent. The independence of the set $\set{\tilde{g}_1,\ldots,\tilde{g}_k}$ is inherited from the independence of $G$. Hence we only need to show that $\tilde{g}_z\not \in \gen{\tilde{g}_1,\ldots,\tilde{g}_k}$. For a contradiction, let us assume $\tilde{g}_z \in \gen{\tilde{g}_1,\ldots,\tilde{g}_k}$. Thus, there exists $g\in \gen{g_1,\ldots,g_k}$ such that $\bra{a}{g}_z\ket{a}=\bra{a}{g}\ket{a}$. This implies that $\bra{a}{g}_z g\ket{a}=I^{\otimes n-1}$ and hence ${g}_z g=(-1)^a Z_j$. But $(-1)^a Z_j$ does not commute with $g_x$ and $\gen{G}$ is a commuting group containing ${g}_z g$ and $g_x$ resulting in a contradiction.
\end{proof}

We are ready to prove Lemma~\ref{lem:extract}.

\begin{proof}
	Let us define the generating set $G^{(0)}=\set{g_1,\ldots,g_{n+t}}\in \genset{n+t}{n+t}$ by $g_{j}:=V Z_j V^{\dagger}$ where $Z_j$ is $I^{\otimes n+t}$ with the $j^{\rm th}$ tensor factor replaced by $Z$. Then, from Eq.~\eqref{eq:stab projector defn}, it is evident that $\Pi_{G^{(0)}}=V\density{0}^{\otimes n+t}_{\rega \regb \regc} V^{\dagger}$. Substituting into the LHS of Eq.~\eqref{eq:extract long}, we get:
	\begin{align}
	{\rm Tr}_{\rega \regb}\left(V\density{0}^{\otimes n+t}_{\rega \regb \regc} V^{\dagger} \density{x}_{\rega}\right)&=2^{-(n+t)}\sum_{g\in \gen{G^{(0)}}} \pphase{g}\tr{ \pfac{g}{\rega} \density{x}_{\rega}} \tr{\pfac{g}{\regb}} \pfac{g}{\regc}\label{eq:can see constraints}.
	\end{align}
	
	From the RHS of Eq.~\eqref{eq:can see constraints}, we note that terms associated with $g$ will be zero if certain constraints on $g$ are not satisfied. In particular, for a fixed $g\in \gen{G^{(0)}}$ to produce a non-zero contribution to the sum, it is necessary that:
	\begin{itemize}
		\item Register `$\rega$' constraints: for all $j\in [w]$, $\pfac{g}{j}\in \set{I,Z}$
		\item Register `$\regb$' constraints: for all $j\in [n-w]$, $\pfac{g}{w+j}= I$.
	\end{itemize}
	
	We note that for any $G\in \genset{n}{k}$ and $j\in [n]$ the sets $S_{\rega}(G,j):=\set{g\in \gen{G}|\pfac{g}{j}\in \set{I,Z}}$ and $S_{\regb}(G,j):=\set{g\in \gen{G}|\pfac{g}{j}=I}$ are both subgroups of $\gen{G}$. Hence, these are generated by some generating sets $G_{\rega}(G,j)$ and $G_{\regb}(G,j)$, respectively. Through a procedure similar to performing computational basis measurements on a stabilizer state in the Gottesman-Knill theorem, the $\extract$ algorithm computes the generating set $\tilde{G}$.
	
    Starting from the stabilizer group $\gen{G^{(0)}}$, the $\extract$ algorithm imposes the above constraints to find a stabilizer generating set $\tilde{G}=\set{\tilde{g}_1, \ldots, \tilde{g}_{\tilde{k}}}$ that satisfies the properties:
	\begin{enumerate}
		\item $\gen{\tilde{G}}\subseteq \gen{G^{(0)}}$
		\item for all $g\in \gen{G^{(0)}}$, $g$ satisfies the register `$\rega$' and `$\regb$' constraints if and only if $g\in \gen{\tilde{G}}$.
	\end{enumerate}
	This allows us to replace the sum over $g\in\gen{G^{(0)}}$ by a sum over $g\in \gen{\tilde{G}}$ in the RHS of Eq.~\eqref{eq:can see constraints}. Some further manipulation gives:
	\begin{align}
	2^{-(n+t)}\sum_{g\in \gen{\tilde{G}}} \pphase{g}\tr{ \pfac{g}{a} \density{x}_a} \tr{\pfac{g}{b}} \pfac{g}{c}&=2^{-(n+t)}\abs{\gen{\tilde{G}}}{\rm Tr}_{\rega \regb}\left(\Pi_{\tilde{G}} \density{x}_{\rega} \right)\notag\\
	&=2^{-n-t+\tilde{k}}{\rm Tr}_{\rega \regb}\left(\Pi_{\tilde{G}} \density{x}_{\rega} \right)\label{eq:extract part1},
	\end{align}
	where $\tilde{k}:=\abs{\tilde{G}}$. We will later define $r$ and $v$ such that $\tilde{k}=t-r+v$ thus proving Eq.~\eqref{eqn:defining-r-indirectly}.

	We now define the linear map $f_x:\gen{\tilde{G}}\rightarrow \pauli{t}$ by:
	\begin{align}
		f_x(g)&=2^{-(n-w)}{\rm Tr}_{\rega \regb}\left(g \density{x}_{\rega} \right)\\
		&=\pphase{g}\bra{x}\pfac{g}{\rega}\ket{x} \pfac{g}{\regc}.
	\end{align}
	We show that this is a group homomorphism. Let $g,g' \in \gen{\tilde{G}}$, then:
	\begin{align}
		f_x(g) f_x(g')&=\pphase{g}\pphase{g'}\bra{x}\pfac{g}{\rega}\density{x}\pfac{g'}{\rega}\ket{x} \pfac{g}{\regc} \pfac{g'}{\regc}\nonumber\\
		&=\pphase{g}\pphase{g'}\bra{x}\pfac{g}{\rega} \pfac{g'}{\rega}\ket{x} \pfac{g}{\regc} \pfac{g'}{\regc}\nonumber\\
		&=2^{-(n-w)}{\rm Tr}_{\rega \regb}\left(g g'\density{x}_{\rega} \right)\nonumber\\
		&=f_x(g g')
	\end{align}
	where in the second equality, we used the fact that $\pfac{g'}{a}$ commutes with $\density{x}$ since $g'$ satisfies the register `$\rega$' constraints.
	
	This shows that the image $f_x(\gen{\tilde{G}})$ is an Abelian subgroup of $\pauli{t}$ generated by the set $\set{f_x(g)~|~g\in \tilde{G}}$. Starting from Eq.~\eqref{eq:extract part1}, we now note that:
	\begin{align}\label{eq:projector trace in terms of fx}
		2^{-n-t+\tilde{k}}{\rm Tr}_{\rega \regb}\left(\Pi_{\tilde{G}} \density{x}_{\rega} \right)&=2^{-t+\tilde{k}-w} f_x(\Pi_{\tilde{G}})
	\end{align}
	where we have extended the domain of $f_x$ by linearity, i.e.
	\begin{align}
		f_x(\Pi_{\tilde{G}})=\abs{\gen{\tilde{G}}}^{-1} \sum_{g\in \tilde{G}} f_x (g).
	\end{align}
	
	We note that the list of elements $f_x(\tilde{g}_1), \ldots, f_x(\tilde{g}_{\tilde{k}})$ that generate $f_x(\gen{\tilde{G}})$ may be dependent. This can happen if and only if there is a $g\in \gen{\tilde{G}}\setminus \set{I^{\otimes n+t}}$ such that $f_x(g)=I^{\otimes t}$. Further, we note that the group $f_x(\gen{\tilde{G}})$ may contain the element $-I^{\otimes t}$. If this is the case, it is easy to show that $f_x(\Pi_{\tilde{G}})=0$. 
	
	We now outline a simple procedure that allows us to:
	\begin{enumerate}
		\item Identify constraints on $x$ that are necessary to ensure that $-I^{\otimes t}\not \in f_x(\gen{\tilde{G}})$ and, assuming $x$ satisfies all such constraints,
		\item Identify a subset of $\gen{\tilde{G}}$ such that its image under $f_x$ is independent and generates $f_x(\gen{\tilde{G}})$. 
	\end{enumerate}
	Let us start with the set $\tilde{G}^{(0)}:=\tilde{G}$ and for each $j\in [w]$ we update $\tilde{G}^{(0)} \rightarrow \tilde{G}^{(1)} \rightarrow \ldots \rightarrow \tilde{G}^{(w)}$. On, the $j^{\rm th}$ update procedure we:
	\begin{enumerate}
		\item Put $\tilde{G}^{(j-1)}$ into $ZX(j)$-form. 
		\item If $\tilde{G}^{(j-1)}$ does not have a leading generator, $\tilde{G}^{(j)}\gets \tilde{G}^{(j-1)}$
		\item If $\tilde{G}^{(j-1)}$ has a leading generator, then it must be a leading-$Z$. Call this $g_z$.
		\item Map all elements $g\in \tilde{G}^{(j-1)}\setminus \set{g_z}$ to $\bra{x_1 \ldots x_{j}}g\ket{x_1 \ldots x_{j}}$ where the $\ket{x_1 \ldots x_{j}}$ vectors act on the first $j$ qubits. Check if the group generated by these elements, $\gen{G(x,j)}$, contains either $\bra{0}\bra{x_1 \ldots x_{j-1}} g_z\ket{x_1 \ldots x_{j-1}}\ket{0}$ or $-\bra{0}\bra{x_1 \ldots x_{j-1}} g_z\ket{x_1 \ldots x_{j-1}}\ket{0}$ (where the $\ket{0}$ vectors act on qubit $j$ and the $\ket{x_1 \ldots x_{j-1}}$ vectors act on the first $j-1$ qubits). If neither is true then $\tilde{G}^{(j)}\gets \tilde{G}^{(j-1)}$.
		\item If the group $\gen{G(x,j)}$ contains $(-1)^{a}\bra{0}\bra{x_1 \ldots x_{j-1}} g_z\ket{x_1 \ldots x_{j-1}}\ket{0}$ for either $a=0$ or $a=1$, then we require $x_{j}=a$. In either case ($x_j=0$ or $x_j=1$), we note that $f_x(g_z)$ is dependent on $\set{f_x(g)~|~g\in \tilde{G}^{(j-1)}\setminus \set{g_z}}$ hence we set $\tilde{G}^{(j)}\gets \tilde{G}^{(j-1)}\setminus \set{g_z}$.  
	\end{enumerate}
	It is clear that for all $j$ where Step 5 applies, our constraint on $x_{j}$ is necessary to ensure that $-I^{\otimes t}\not \in f_x(\gen{\tilde{G}})$.
	
	The final output $\tilde{G}^{(w)}$ of this procedure is a subset of $\gen{\tilde{G}}$ such that its image under $f_x$ is independent and generates $f_x(\gen{\tilde{G}})$. To see that $\set{f_x(g)~|~g\in \tilde{G}^{(w)}}$ generates $f_x(\gen{\tilde{G}})$, we note that the above procedure only deletes elements in Step 5. In this case, the deleted element $g_z$ has the property that $f_x(g_z)$ is dependent on the $f_x$ image of the remaining elements. Hence $\set{f_x(g)~|~g\in \tilde{G}^{(w)}}$ generates $f_x(\gen{\tilde{G}})$ since $\set{f_x(g)~|~g\in \tilde{G}^{(0)}}$ generates $f_x(\gen{\tilde{G}})$. We can see that $\set{f_x(g)~|~g\in \tilde{G}^{(w)}}$ is independent by induction. First, we note that $\set{g\in \tilde{G}^{(0)}}$ is independent. Now for the $j^{\rm th}$ induction step, assume that $\set{\bra{x_1 \ldots x_{j-1}}g\ket{x_1 \ldots x_{j-1}}~|~ g\in \tilde{G}^{(j-1)}}$ is independent (where the vector $\ket{x_1 \ldots x_{j-1}}$ acts on the first $j-1$ qubits). Then to see that $\set{\bra{x_1 \ldots x_{j}}g\ket{x_1 \ldots x_{j}}~|~ g\in \tilde{G}^{(j)}}$ is independent let us, without loss of generality, assume $\tilde{G}^{(j-1)}$ is in $\zxform$ with respect to qubit $j$. Then it is clear that $\set{\bra{x_1 \ldots x_{j}}g\ket{x_1 \ldots x_{j}}~|~ g\in \tilde{G}^{(j-1)}, g \text{ not a leading generator}}$ are independent. Thus dependence can only arise if $\bra{x_1 \ldots x_{j}}g_z\ket{x_1 \ldots x_{j}}$ is dependent on $\set{\bra{x_1 \ldots x_{j}}g\ket{x_1 \ldots x_{j}}~|~ g\in \tilde{G}^{(j-1)}, g \text{ not a leading generator}}$. We explicitly check this (in Step 4) and exclude $g_z$ from $\tilde{G}^{(j)}$ if it gives rise to dependence. Hence by induction we have shown that $\set{\bra{x}g\ket{x}~|~g\in \tilde{G}^{(w)}}$ is independent. This immediately leads to the independence of  $\set{f_x(g)~|~g\in \tilde{G}^{(w)}}$.

	We define $G:=f_x(\tilde{G}^{(w)})\in \genset{t}{k}$ for some $k\leq t$. Using Eq.~\eqref{eq:projector trace in terms of fx} and the definition of $\tilde{k}$, Eq.~\eqref{eq:extract long} follows.	We define $r:=t-k$ and define $v$ as the number of rows deleted in Step 5. Hence, $\abs{\tilde{G}}=\abs{\tilde{G}^{(w)}}+v=\abs{G}+v=t-r+v$. Thus, it is clear that $v\in \set{0,1,\ldots,w}$ and $r\leq t$. We now show that $r\leq n-w$. Let us note that to derive $\tilde{G}$ from $G^{(0)}$, we imposed the constraints on registers `$\rega$' and `$\regb$'. Since $G^{(0)}$ is a maximal generating set, just imposing the register `$\regb$' constraints must result in the deletion of between $(n-w)$ and $2(n-w)$ generators. Subsequently imposing register `$\rega$' constraints must result in the further deletion of $d_{\rega}$ generators where, by using Lemma~\ref{lem: independence under fx}, one can show that $d_{\rega}$ is between $0$ and $w-v$. Thus,
	\begin{align}
		\abs{\tilde{G}} &\in \set{\abs{G^{(0)}}-2(n-w)-(w-v), \abs{G^{(0)}}-2(n-w)-w+1 \ldots, \abs{G^{(0)}}-(n-w)}\nonumber\\
		&= \set{(n+t)-2(n-w)-(w-v), (n+t)-2(n-w)-(w-v)+1 \ldots, (n+t)-(n-w)}\nonumber\\
		&= \set{-n+t+w+v, -n+t+w+v+1 \ldots, t+w}
	\end{align}
	Now, $\abs{\tilde{G}}=t-r+v$ by definition of $r$ and $v$. So:
	\begin{align}
		r&\in \set{(t+v)-(t+w), \ldots, (t+v)-(-n+t+w+v)}\nonumber\\
		&=\set{v-w,\ldots,n-w}.
	\end{align}
	This shows that $r\leq n-w$ completing our proof.

\end{proof}

% -------------------------------------------------------------
% SEC. IV - PROOF OF LEMMA 5
% -------------------------------------------------------------

\section{Proof of Lemma~\ref{lem:gateseq}}
\label{app:gate_seq}

\begin{proof}
	We will provide a sequence of gates that forms a unitary $W$ transforming $\Pi_G$ into $\ketbra{0}{0}^{\otimes t-r}\otimes I^{\otimes r}$ for an arbitrary stabilizer generator matrix ${G} \in \genset{t}{t-r}$. Equivalently, it means that $W$ should transform a generator matrix $G$ into a generator matrix corresponding to a projector $\ketbra{0}{0}^{\otimes t-r}\otimes I^{\otimes r}$, i.e.,	
	\begin{equation}	
		G_0=[~X~\|~Z~] \quad\xrightarrow{W}\quad G_{fin}=\left[~0~\|~I~|~0~\right].\label{eq:generator-block-form}
	\end{equation} 
	Here, we will adopt the formulation in terms of tableaux given in Ref.~\cite{aaronson2004improved}. For our algorithm we only require the stabilizers, not the destabilisers so the matrices given in equation~\ref{eq:generator-block-form} (and in the sequel) correspond to the lower half of the destabiliser+stabiliser tableaux of Aaronson and Gottesman.

    Each of $(t-r)$ rows of a generator matrix corresponds to a stabilizer generator given by a Pauli matrix encoded as a binary vector of length $2t$. The first $t$ entries correspond to $X$ stabilizers (i.e. if the entry in column $j\in[t]$ equals one, then there is a Pauli $X$ acting on the $j^{\mathrm{th}}$ qubit), while the remaining $t$ entries correspond to $Z$ stabilizers. If a row $k$ has a qubit with a $1$ in both the $X$ and $Z$ portion then the stabiliser has the operator $Y_k$ on qubit $k$ (this differs from $X_kZ_k$ by a factor of $i$). We use the symbol $\|$ to visually separate the $X$ and $Z$ parts, and the separator $|$ to separate a square block within each part. We also use $X$ and $Z$ to denote arbitrary entries in the corresponding parts, and $I$ to denote a square $(t-r)\times (t-r)$ identity matrix. In addition to the Pauli matrices each stabiliser has a phase $\pm 1$ associated to it. This information is stored in a binary vector $f$ of length $(t-r)$.
	
	First, given $G_0$, we can perform row sums, as they correspond to stabiliser multiplication. We can also swap pairs of columns $j$ and $j+t$ (one in $X$ part, the other in $Z$ part) by applying a Hadamard gate to qubits $j$. Using these two operations, we can bring the $X$ part to the reduced row echelon form. More precisely, using row sums we can perform Gaussian elimination of the $X$ part, and each time we find a column with no leading $1$ in it, we can bring the missing $1$ from the $Z$ part (if one exists) using a Hadamard gate. Since the rows are independent we will obtain exactly $(n-r)$ leading $1$s in this way. Finally, using $\operatorname{SWAP}$ gates (each composed of three $\operatorname{CX}$ gates), we can permute the columns so that the first $(t-r)\times (t-r)$ block in the $X$ part is given by the identity matrix. Thus, after using at most $t-r$ Hadamard gates and $3(t-r)$ instances of $\operatorname{CX}$ gates, we arrive at
	\begin{equation}
	G_1=[~I~|~X~\|~Z~].
	\end{equation}
	
	Next, we can use $\operatorname{CX}$ gates to clear the remaining $(t-r) \times r$ block of the $X$ part. Specifically, if there is a $1$ in row $j$ and column $k$ of this block, the application of $\operatorname{CX}$ controlled on qubit $j$ and targeted on qubit $k$ flips that $1$ to $0$. It also non-trivially affects the entries of the $Z$ part, but we will deal with that in the next step. Thus, after using at most $(t-r)r$ instances of $\operatorname{CX}$ gates, we obtain
	\begin{equation}
	G_2=[~I~|~0~\|~Z~].
	\end{equation}
	
	The third step employs phase gates to ensure that the main diagonal of the $Z$ part, i.e. the elements $Z_{jj}$, are all zero. To achieve this, it is enough to apply $S$ to every qubit $j$ such that $Z_{jj}=1$. This requires the use of at most $(t-r)$ phase gates and results in
	\begin{equation}
	G_3=[~I~|~0~\|~\tilde{Z}~],
	\end{equation}
	with tilde indicating the zero diagonal of the main $(t-r)\times(t-r)$ block in the $Z$ part.
	
	Now, we can use $CZ$ gates to clear the last $(t-r) \times r$ block of the $Z$ part. Specifically, if there is a $1$ in row $j$ and column $k$ of this block, the application of $\operatorname{CZ}$ controlled on qubit $j$ and targeted on qubit $k$ flips that $1$ to $0$. It also does not affect the entries of the $X$ part at all. Thus, after using at most $(t-r)r$ instances of $\operatorname{CZ}$ gates, we obtain
	\begin{equation}
	G_4=[~I~|~0~\|~\tilde{Z}~|~0~].
	\end{equation}
	
	In the fifth step, we employ the fact that stabilisers commute. Assume that there is a non-zero element $\tilde{Z}_{ij} = 1$, $i\neq j$ in the $\tilde{Z}$ matrix, this means there is Pauli $Z$ matrix on qubit $j$ in stabiliser $i$. Stabiliser $i$ has to commute with stabiliser $j$, since there is a $1$ in element $(j,j)$ of the $X$ block there is a Pauli $X$ acting on qubit $j$ in stabiliser $j$ which would lead to stabiliser $i$ anti-commuting with stabiliser $j$. Since the $X$ part of the tableau is $\left[~I~|~0~\right]$ we must also have a Pauli $Z$ in stabiliser $j$ acting on qubit $i$, $\tilde{Z}_{ij} = 1\implies \tilde{Z}_{ji} = 1$. Repeating this argument with the initial assumption $\tilde{Z}_{ij} = 0$ proves that $\tilde{Z}$ is symmetric $\tilde{Z}^T = \tilde{Z}$. 

    This allows us to zero the whole $Z$ part using $\operatorname{CZ}$ gates. More precisely, for each unordered pair $(j,k)$ such that \mbox{$\tilde{Z}_{jk} = \tilde{Z}_{kj} = 1$}, the application of a $\operatorname{CZ}$ gate controlled on qubit $j$ and targeted on qubit $k$ flips both those $1$'s to $0$'s. It also does not affect any other entries. Thus, after using at most $r(r-1)/2$ instances of $\operatorname{CZ}$ gates, we arrive at
	\begin{equation}
	G_5=[~I~|~0~\|~0~].
	\end{equation}
	At this stage it is convenient to zero the phase vector $f$. For each row $k$ if $p_k = 1$ we apply $Z_k = S_k^2$. Due to the form of the $X$ part of the matrix it is easy to see that this will multiply stabiliser $k$ by $-1$ and leave all the others invariant. This requires at most $2(t-r)$ applications of the $\operatorname{S}$ gate.
	
	The final step requires implementing $t-r$ Hadamard gates to transform the above $G_5$ to $G_{fin}$. Summarising, in all steps we used at most $ 2 (4 + r) t-r (17 + 3 r)/2$ Clifford gates including at most $2(t-r)$ Hadamard gates.

\end{proof}

% -------------------------------------------------------------
% SEC. V - PROOF OF LEMMA 6
% -------------------------------------------------------------

\section{Proof of Lemma~\ref{lem:psi(y)_bound}}
\label{app:bound}

\begin{proof}
	From the definition of $\ket{\psi(y)}$, Eq.~\eqref{eq:psi(y)}, we have
	\begin{equation}
		\twonorm{\ket{\psi(y)}}^2:=\cextent\cdot  2^{t-r+v-w} \twonorm{\bra{0}^{\otimes t-r}W\ket{\tilde{y}}}^2,
	\end{equation}
	Now, combining the statements of Lemma~\ref{lem:extract} and Lemma~\ref{lem:gateseq} we have:
	\begin{equation}
		{\rm Tr}_{\rega \regb}\left(V\density{0}^{\otimes n+t}_{\rega \regb\regc} V^{\dagger} \density{x}_{\rega}\right)=2^{-r+v-w}W^\dagger ( \density{0}^{\otimes t-r} \otimes I^{\otimes r}) W.
	\end{equation}
	and so
	\begin{equation}
          2^t\twonorm{\bra{x}_{\rega}\bra{\tilde{y}}_{\regc} V\ket{0}_{\rega\regb\regc}^{\otimes n+t}}^2 = 2^{t-r+v-w}\twonorm{\bra{0}^{\otimes t-r}W\ket{\tilde{y}}}^2. 
	\end{equation}
	Therefore,
	\begin{equation}
		\label{eq:bound-part}
		\twonorm{\ket{\psi(y)}}^2=\cextent\cdot  2^{t} \twonorm{\bra{x}_{\rega}\bra{\tilde{y}}_{\regc} V\ket{0}_{\rega\regb\regc}^{\otimes n+t}}^2\leq \cextent \cdot  2^{t} \twonorm{\bra{\tilde{y}}_{\regc} V\ket{0}_{\rega\regb\regc}^{\otimes n+t}}^2.
	\end{equation}
	Let us take a closer look at $\bra{\tilde{y}}_{\regc} V\ket{0}_{\rega\regb\regc}^{\otimes n+t}$. Referring to Fig.~\ref{fig:circuit}, recall that the unitary $V$ describes a Clifford circuit of the form
	\begin{equation}
		V =C_{t} \prod_{j=1}^{t} CX_{j} C_{j-1}
	\end{equation}
	where $C_j$ is an arbitrary Clifford gate on $n$ qubits in register `$\rega\regb$', $CX_j$ is a CNOT gate controlled on one of the qubits from register `$\rega\regb$' and targeted at the $j$-th qubit in register `$\regc$', and the product is ordered from right to left (i.e., the rightmost term is given by $C_0$). We then have
	\begin{align}
		\bra{\tilde{y}}_{\regc} V\ket{0}_{\rega\regb\regc}^{\otimes n+t}&=\bra{\tilde{y}_1\dots\tilde{y}_t}_{\regc} \left(C_{t} \prod_{j=1}^t CX_{j} C_{j-1}\right) \ket{0}_{\rega\regb}^{\otimes n}\otimes \ket{0}_{\regc}^{\otimes t}\\
		&=\bra{\tilde{y}_2\dots\tilde{y}_t}_{\regc} \left(C_{t} \prod_{j=2}^t CX_{j} C_{j-1}\right) \ket{\Phi_1}_{\rega\regb}\otimes\ket{0}_{\regc}^{\otimes (t-1)}
	\end{align}
	with 
	\begin{align}
		\ket{\Phi_1}_{\rega\regb}=\bra{\tilde{y}_1}_{\regc} CX_1C_0\ket{0}_{\rega\regb}^{\otimes n}\otimes \ket{0}_{\regc}
	\end{align}
	being an $n$-qubit unnormalised stabliser state (note that we could commute the projector $\bra{\tilde{y}_1}_{\regc}$ through the circuit as the first qubit in the register `$\regc$' is never again affected by it). In order to normalize $\ket{\Phi_1}_{\rega\regb}$ we first write
	\begin{equation}
		\label{eq:decompose}
		C_0\ket{0}_{\rega\regb}^{\otimes n}=c_0\ket{0\psi_0}_{\rega\regb}+\ket{1\psi_1}_{\rega\regb},
	\end{equation}
	with $|c_0|^2+|c_1|^2=1$ and the distinguished first qubit corresponding to the control qubit of $CX_1$. We then use the above to obtain
	\begin{align}
		\ket{\Phi_1}_{\rega\regb}&=\bra{\tilde{y}_1}_{\regc} CX_1 (c_0\ket{0\psi_0}_{\rega\regb}+c_1\ket{1\psi_1}_{\rega\regb})\otimes \ket{0}_{\regc}\\
		&=\bra{\tilde{y}_1}_{\regc}(c_0\ket{0\psi_0}_{\rega\regb}\otimes\ket{0}_{\regc}+c_1\ket{1\psi_1}_{\rega\regb}\otimes\ket{1}_{\regc})\\
		&=\frac{1}{\sqrt{2}}(c_0\ket{0\psi_0}_{\rega\regb}+c_1(-i)^{\tilde{y}_1}\ket{1\psi_1}_{\rega\regb}),
	\end{align}
	and so we conclude that the normalised state is given by
	\begin{equation}
		\ket{\Phi_1'}_{\rega\regb}=\frac{1}{\sqrt{2}}\ket{\Phi_1}_{\rega\regb}.
	\end{equation} 
	
	We thus have
	\begin{align}	
		\bra{\tilde{y}}_{\regc} V\ket{0}_{\rega\regb\regc}^{\otimes n+t}&=\frac{1}{\sqrt{2}}\bra{\tilde{y}_2\dots\tilde{y}_t}_{\regc} \left(C_{t} \prod_{j=2}^t CX_{j} C_{j-1}\right) \ket{\Phi_1'}_{\rega\regb}\otimes\ket{0}_{\regc}^{\otimes (t-1)},
	\end{align}
	and we can repeat the whole procedure again. More precisely, we introduce an $n$-qubit unnormalised stabiliser state
	\begin{align}
		\ket{\Phi_2}_{\rega\regb}=\bra{\tilde{y}_2}_{\regc} CX_2C_1\ket{\Phi_1'}_{\rega\regb}\otimes \ket{0}_{\regc},
	\end{align}
	we decompose $C_1\ket{\Phi_1'}_{\rega\regb}$ analogously as we did in Eq.~\eqref{eq:decompose}, and repeating the same reasoning we arrive at
	\begin{align}	
		\bra{\tilde{y}}_{\regc} V\ket{0}_{\rega\regb\regc}^{\otimes n+t}&=\frac{1}{\sqrt{2^2}}\bra{\tilde{y}_3\dots\tilde{y}_t}_{\regc} \left(C_{t} \prod_{j=3}^t CX_{j} C_{j-1}\right) \ket{\Phi_2'}_{\rega\regb}\otimes\ket{0}_{\regc}^{\otimes (t-2)}.
	\end{align}
	Repeating it $t$ times in total we finally arrive at
	\begin{align}		
		\bra{\tilde{y}}_{\regc} V\ket{0}_{\rega\regb\regc}^{\otimes n+t}&=\frac{1}{\sqrt{2^t}} C_{t} \ket{\Phi_t'}_{\rega\regb}.
	\end{align}
	Substituting the above to the inequality from Eq.~\eqref{eq:bound-part} we finally arrive at
	\begin{equation}
		\twonorm{\ket{\psi(y)}}^2\leq \cextent.
	\end{equation}
\end{proof}

% -------------------------------------------------------------
% SEC. VI - PROOF OF LEMMA 7
% -------------------------------------------------------------

\section{Proof of Lemma~\ref{lem:hoeffding}}
\label{app:hoeffding}

In order to prove our result, we will need the definition of a \emph{very-weak martingale} and a theorem from Ref.~\cite{hayes2005large} given below.

\begin{defn}[Very-weak martingale]
	Let $N\in \bbn$, $\Omega$ be a sample space and for all $j\in \bbn$, let $X_j:\Omega \rightarrow \bbr^N$ be a random variable taking values in $\bbr^N$ such that $X_0=0$, $\expval{}{\twonorm{X_j}}< \infty$ and $\expval{}{X_j~|~X_{j-1}}=X_{j-1}$. Then we call the sequence $(X_0, X_1, \ldots)$ a \emph{very-weak martingale in} $\bbr^N$.
\end{defn}

\begin{thm}[Theorem 1.8 of Ref.~\cite{hayes2005large}]
	\label{thm:hayes} 
	Let $X$ ba a very-weak martingale taking values in $\bbr^N$ such that $X_0=0$ and for every $j$, $\twonorm{X_j-X_{j-1}}\leq 1$. Then for every $a>0$:
\begin{align}
	\label{eq: hayes}
	\pr \brackn{\twonorm{X_s}\geq a}\leq 2 e^{1-(a-1)^2/2s}<2 e^2 \exp \brackn{-a^2/2s}.
\end{align}
\end{thm} 

We can now prove Lemma~\ref{lem:hoeffding}.

	% \begin{align}\label{eq: mean phi}
	% 	\ket{\widebar{\psi}}=\frac{1}{s}{\sum_{j=1}^s \ket{\psi_{x_j}}}\text{.}
	% \end{align}
	% Then, for all $\epsilon>0$:
	% \begin{align}\label{hayes hoeffding twonorm}
	% 	\bigpr{\twonorm{\ket{\widebar{\psi}}-\ket{\mu}}\geq \epsilon}\leq 2 e^2 \exp{\left(\frac{-{s}{\epsilon^2}}{2 (\sqrt{m}+\sqrt{p})^2}\right)}
	% \end{align}
	% and
	% \begin{align}\label{hayes hoeffding onenorm}

\begin{proof}
	The proof is a simple application of Theorem~\ref{thm:hayes}. Let us use $R:\bbc^d \rightarrow \bbr^{2 d}$ to denote the two-norm preserving linear map $R(a_1+i b_1, \ldots, a_d+ i b_d)=(a_1, b_1, \ldots, a_d, b_d)$. For $s\in \bbn$, we define the random variable $Y_s\in \bbr^{2d}$ as follows: $Y_0=(0,\ldots,0)$ and for $s>0$:
	\begin{align}
		Y_s:=\frac{s}{\sqrt{m}+\sqrt{p}} R\brackn{\ket{\widebar{\psi}}-\ket{\mu}}\text{,}
	\end{align}
	where we note that $\ket{\widebar{\psi}}$ depends on $s$ as per Eq.~\eqref{eq: mean phi}.
	
	We now note that $Y_s$ is a very-weak martingale since $Y_0=0$ and
	\begin{align}
		\expval{}{\twonorm{Y_s}}&=\frac{s}{\sqrt{m}+\sqrt{p}}\expval{}{\sqrt{\brackn{\bra{\widebar{\psi}}-\bra{\mu}}\brackn{\ket{\widebar{\psi}}-\ket{\mu}}}}< \infty, 
	\end{align}
	 as well as 
	 \begin{align}
		 \expval{}{Y_s~|~Y_{s-1}}&=\expval{}{Y_{s-1}+\frac{1}{\sqrt{m}+\sqrt{p}}R\brackn{\ket{\psi_{x_s}}-\ket{\mu}}~|~Y_{s-1}}=Y_{s-1}\text{.}
	 \end{align}
	Additionally, we note that $\twonorm{Y_s-Y_{s-1}}\leq 1$, since:
	\begin{align}
		\twonorm{Y_s-Y_{s-1}}&=\twonorm{\frac{1}{\sqrt{m}+\sqrt{p}}R\brackn{\ket{\psi_{x_s}}-\ket{\mu}}}\\
		&=\frac{1}{\sqrt{m}+\sqrt{p}}\sqrt{\innerbrakets{\psi_{x_s}}{\psi_{x_s}}-\innerbrakets{\psi_{x_s}}{\mu}-\innerbrakets{\mu}{\psi_{x_s}}+\innerbrakets{\mu}{\mu}}\\
		&\leq \frac{1}{\sqrt{m}+\sqrt{p}}\sqrt{m+2\sqrt{m p}+p}=1\text{.}
	\end{align}
	Hence, by Theorem~\ref{thm:hayes}:
	\begin{align}
		\pr \brackn{\twonorm{Y_s}\geq a} &=\pr \brackn{\twonorm{\ket{\widebar{\psi}}-\ket{\mu}}\geq \frac{a(\sqrt{m}+\sqrt{p})}{s}}<2 e^2 \exp \brackn{-a^2/2s}\text{.}
	\end{align}
	Substituting $\epsilon=a(\sqrt{m}+\sqrt{p})/s$ proves Eq.~\eqref{hayes hoeffding twonorm}. 
	
	To prove Eq.~\eqref{hayes hoeffding onenorm}, we define:
	\begin{align}
		\ket{\Delta}:=\ket{\mu}-\ket{\widebar{\psi}},
	\end{align}
	and note that
	\begin{align}
		\sonenorm{\density{\widebar{\psi}}-\density{\mu}}&=\sonenorm{\ketbra{\mu}{\Delta}+\ketbra{\Delta}{\mu}-\density{\Delta}}\leq\sonenorm{\ketbra{\mu}{\Delta}}+\sonenorm{\ketbra{\Delta}{\mu}}+\sonenorm{\density{\Delta}}\\
		&=2\twonorm{\ket{\mu}}\twonorm{\ket{\Delta}}+\twonorm{\ket{\Delta}}^2=\twonorm{\ket{\Delta}}(\twonorm{\ket{\Delta}}+2\sqrt{p})\text{.}
	\end{align}
	Now, employing the above, if $\twonorm{\ket{\Delta}}\leq \epsilon$ then $\sonenorm{\density{\widebar{\psi}}-\density{\mu}}\leq \epsilon (\epsilon + 2\sqrt{p})$. Applying this observation to the already proven Eq.~\eqref{hayes hoeffding twonorm} yields:
	\begin{align}
		\bigpr{\sonenorm{\density{\widebar{\psi}}-\density{\mu}}\geq \epsilon (\epsilon + 2\sqrt{p})}\leq 2 e^2 \exp{\left(\frac{-{s}{\varepsilon^2}}{2 (\sqrt{m} +\sqrt{p})^2}\right)}\text{.}\notag
	\end{align}
	We can now define a new variable $\varepsilon=\epsilon (\epsilon + 2\sqrt{p})$ and solve this quadratic equation for $\epsilon$. Taking only the positive solution gives $\epsilon=\sqrt{p+\varepsilon}-\sqrt{p}$, which immediately leads to Eq.~\eqref{hayes hoeffding onenorm}.
	
\end{proof}

% -------------------------------------------------------------
% SEC. VII - CH FORM
% -------------------------------------------------------------

\section{CH form}
\label{app:ch_form}

Following the formalism developed in  Ref.~\cite{bravyi2019simulation}, any stabilizer state $\ket{\sigma}$ of $n$ qubits can be written as
\begin{equation}
	\ket{\sigma}=\omega U_C U_H \ket{s},
\end{equation}
where $U_C$ is the control type operator (in our case effectively meaning that it consists of products of $S$, $CX$ and $CZ$ gates), $U_H$ is the Hadamard-type operator (consisting only of products of $H$ gates), $s$ is a bit string of length $n$ representing one of the computational basis states, and $\omega$ is a complex number. The unitary $U_C$ is fully specified by three $n\times n$ matrices $F,G,M$ with entries in $\mathbb{Z}_2$ and a phase vector~$\gamma$ of length $n$ with entries in $\mathbb{Z}_4$. Together, they describe the action of $U_C$ on Pauli matrices:
\begin{align}
	U_C^\dagger	Z_j	U_C = \prod_{k=1}^{n} Z_k^{G_{jk}},\quad 	U_C^\dagger	X_j	U_C = i^{\gamma_j}\prod_{k=1}^{n} X_k^{F_{jk}}Z_k^{M_{jk}},
\end{align}
where $X_j$ and $Z_j$ are Pauli matrices acting on the $j$-th qubit. The unitary $U_H$ is fully specified by a bit string $v$ of length $n$, with $v_j=1$ if $U_H$ acts with a Hadamard on the $j$-th qubit. Thus, a general stabilizer state $\ket{\sigma}$ of $n$ qubits is described by a tuple \mbox{$\{F,G,M,{\gamma},{v},s,\omega\}$}.

The initial state is represented by the following tuple
\begin{equation}
	\ket{0}^{\otimes n} \quad \Longleftrightarrow\quad \{F=\iden,G=\iden,M=0,{\gamma}={0},{v}={0},{s}={0},\omega=1\}.
\end{equation}
The authors of Ref.~\cite{bravyi2019simulation} found an efficient way to find the tuple \mbox{$\{F',G',M',{\gamma}',{v}',{s}',\omega'\}$} representing \mbox{$V\ket{0}^{\otimes n}$} for an arbitrary Clifford unitary $V$. The run-time of this evolution subroutine scales polynomially with the total number of qubits $n$: the ``C-type'' gates ($S$, $CX$, $CZ$) have linear time complexity $O(n)$, while applying a Hadamard gate takes $O(n^2)$ steps. For completeness, we include the update rules for left, $\L[\Gamma]$, and right, $\Rcal[\Gamma]$, multiplication of $U_C$ by a C-type unitary $\Gamma$. All phase vector updates are performed modulo four, and each update containing the symbol $p$ should be read as applying to all $p\in \{1,\ldots, n\}$ in turn.

\begin{subequations}
	\begin{align}	
	\label{eq:update1}
	\Rcal[\Sgate{q}] \, :& \, \left\{
	\begin{array}{rcl}
	M_{p,q} &\gets& M_{p,q} \oplus F_{p,q}  \\
	\gamma_p &\gets&  \gamma_p - F_{p,q} \\
	\end{array} \right.
	\quad & \quad
	\L[\Sgate{q}] \, :& \, \left\{
	\begin{array}{rcl}
	M_{q,p} &\gets& M_{q,p} \oplus G_{q,p}\\
	\gamma_q &\gets&  \gamma_q - 1  \\
	\end{array} \right.\\
	\Rcal[\CZgate{q}{r}] \, :& \, \left\{
	\begin{array}{rcl}
	M_{p,q} &\gets& M_{p,q} \oplus F_{p,r}  \\
	M_{p,r} &\gets& M_{p,r} \oplus F_{p,q}  \\
	\gamma_p &\gets&  \gamma_p + 2F_{p,q}F_{p,r}  \\
	\end{array} \right.
	\quad & \quad
	\L[\CZgate{q}{r}] \, :& \, \left\{
	\begin{array}{rcl}
	M_{q,p} &\gets& M_{q,p} \oplus G_{r,p}  \\
	M_{r,p} &\gets& M_{r,p} \oplus G_{q,p}  \\
	\end{array} \right.\\
	\Rcal[\CXgate{q}{r}] \, :& \, \left\{
	\begin{array}{rcl}
	G_{p,q} &\gets& G_{p,q} \oplus G_{p,r}  \\
	F_{p,r} &\gets& F_{p,r} \oplus F_{p,q}  \\
	M_{p,q} &\gets& M_{p,q} \oplus M_{p,r}\\
	\end{array} \right.
	\quad & \quad
	\L[\CXgate{q}{r}] \, :& \, \left\{
	\begin{array}{rcl}
	\gamma_q &\gets& \gamma_q+ \gamma_r + 2(MF^T)_{q,r} \\
	G_{r,p} &\gets& G_{r,p} \oplus G_{q,p}  \\
	F_{q,p} &\gets& F_{q,p} \oplus F_{r,p}  \\
	M_{q,p} &\gets& M_{q,p} \oplus M_{r,p}\\
	\end{array} \right.
	\label{eq:update3}
	\end{align}
\end{subequations}
We note a slight difference to the update rules as presented by the authors of Ref.~8: in the update rule for $\L[\CXgate{q}{r}]$ we update $\gamma$ before updating $F$ and $M$ to emphasise that $\gamma$ must be updated based on the old values of $F$ and $M$, rather than the new. It will be significant that the action of the operation $\L[\CXgate{q}{r}]$ on the $F$ matrix is column addition and that, since this addition is modulo two, the right-action of the swap gate, $\CXgate{q}{r} \CXgate{r}{q} \CXgate{q}{r}$, on $F$ is just to swap the columns $r$ and $q$. 

We will also employ the equation given in Ref.~\cite{bravyi2019simulation} to compute inner products between CH-form stabiliser states and computational basis states,
\begin{align}
	\bra{x} U_C U_H \ket{s} &= \bra{0}^{\otimes n}\left(\prod_{p=1}^n U_C^{-1} X_p^{x_p}U_C\right) U_H\ket{s}= 2^{-\frac{\abs{v}}{2}} i^{\mu} \prod_{j:~v_j = 1} (-1)^{u_j s_j} \prod_{j:~v_j = 0} \braket{u_j}{s_j},\label{eq:bravyi-inner-product}
\end{align}
where $u_j = x F$, and $\mu = x\cdot\gamma + 2k$ for a constant $k\in\{0,1\}$ which may be computed in quadratic time given $x$ and the CH-form data. Indeed, some algebra demonstrates that $k$ may be determined by the relation
\begin{align}
	\prod_{p=1}^n U_C^{-1} X_p^{x_p}U_C &= i^{x\cdot\gamma}\prod_{p :~ x_p = 1}  \prod_{j=1}^n X_j^{F_{pj}} Z_j^{M_{pj}}=  i^{x\cdot\gamma} (-1)^k \prod_{j=1}^n \prod_{p :~ x_p = 1} Z_j^{M_{pj}} X_j^{F_{pj}}.
\end{align}

% -------------------------------------------------------------
% SEC. VIII - PROOF OF LEMMA 9
% -------------------------------------------------------------

\section{Proof of Lemma~\ref{lem:discard}}
\label{app:discard}

The proof splits into two parts. First, we show in Lemma~\ref{lem:if-f-is-good-we-are-happy} that if the CH-form describing the initial $(n+1)$-qubit state $\ket{0}\otimes\ket{\sigma}$ has a certain form, then a CH-form describing $\ket{\sigma}$ may be obtained by simply deleting the first row and column of each $F$, $G$ and $M$, and the first element of $\gamma$, $v$ and $s$. We assume the deletion operation takes quadratic time to leading order as an implementation is likely to simply allocate a new $O(n^2)$ sized block of memory and then copy the required values across. Second, we give an algorithm which takes an arbitrary CH-form representing $\ket{0}\otimes\ket{\sigma}$ and outputs in quadratic time a CH-form representing the same state, but in the form required by Lemma~\ref{lem:if-f-is-good-we-are-happy}. We show, in Lemmas~\ref{lem:constraining-fcal} and \ref{lem:fixing-fcal-1}, how the CH-form representing such a product state can be brought into a form with at most one qubit $k$ with $v_k = 0$, $s_k=1$. This is important because we can insert $\CX$ gates controlled on any qubit with $v_k = 0$, $s_k=0$ between $U_C$ and $U_H$ in the CH-form without changing the state (a $\CX$ controlled on $\ket{0}$ does nothing). Finally, in Lemma~\ref{lem:first-row-of-g-is-nice}, we show that if the CH-form is in the form produced by Lemmas~\ref{lem:constraining-fcal} and~\ref{lem:fixing-fcal-1}, then inserting $\CX$ gates controlled on $\ket{0}$ qubits can bring the CH-form into the form required by Lemma~\ref{lem:if-f-is-good-we-are-happy}.

We label computational basis vectors with bit-strings, denote bitwise addition modulo-$2$ with the symbol $\oplus$, bitwise multiplication with juxtaposition, and use the operator $\join{}{}$ to denote concatenation of bitstrings, e.g., if $a=01$ then $\join{0}{a} = 001$ and $\join{a}{0}=010$. The first part of the proof is then captured by the following.

\begin{lem}\label{lem:if-f-is-good-we-are-happy}
	Consider a stabiliser state $\ket{0}\otimes\ket{\sigma}$ with CH-form $\F = \left\{F, G, M, \gamma, v, s, \omega\right\}$ such that: the first column of $F$ has a $1$ in the first element and zeros elsewhere, and $s_1 = v_1 = 0$. Then, the CH-form $\F^\prime = \left\{F^\prime, G^\prime, M^\prime, \gamma^\prime, v^\prime, s^\prime ,\omega\right\}$, where $F^\prime$, $G^\prime$ and $M^\prime$ are formed by deleting the first row and column of $F$, $G$ and $M$, respectively, and $\gamma^\prime$ $v^\prime$ and $s^\prime$ are formed by deleting the first element of $\gamma$, $v$ and $s$, respectively, is a CH-form for $\ket{\sigma}$.
\end{lem}

\begin{proof}
	Choose an arbitrary $n$-qubit computational-basis vector $\ket{a}$ and use Eq.~\eqref{eq:bravyi-inner-product} to compute
	\begin{align}
	\braket{a}{\sigma} &= \braket{0}{0} \braket{a}{\sigma}= \bra{\join{0}{a}} \ket{0}\otimes\ket{\sigma}= \omega \bra{\join{0}{a}} U_C U_H\ket{s}= (-1)^k \omega i^{u\cdot \gamma}  \bra{0^n} X(u)  U_H\ket{s}\label{eqn:first-inner-prod-line-with-swaps}\\ 
	&= (-1)^k \omega i^{u\cdot \gamma}  2^{-\abs{v}/2} \prod_{j:~v_j=1} (-1)^{u_j s_j} \prod_{j:~v_j=0} \braket{u_j}{s_j},\label{eqn:inner-prod-last-line}
	\end{align}
	where $k$ is a bit we have introduced to count whether we have swapped an even or odd number of $X_j$ with their $Z_j$ to arrive at the final equality in Eq.~\eqref{eqn:first-inner-prod-line-with-swaps}. We recall that $u = (\join{0}{a})F$. If $F^\prime$ is the matrix obtained from $F$ by removing the first row and column, then one can verify that $u = \join{0}{(aF^\prime)}$; in particular $u_1=0$. In addition, we define $\gamma^\prime$, $v^\prime$ and $s^\prime$ by deleting the first element of $\gamma$, $v$ and $s$, respectively, and we let $u^\prime = aF^\prime$. Ignoring $k$ for the moment we go through the rest of the terms in turn.
	First, since $u_1 = 0$
	\begin{align}
	u^\prime \cdot\gamma^\prime = u\cdot\gamma.
	\end{align}
	Next, since we have $v_1 = 0$,
	\begin{align}
	\abs{v^\prime} = \abs{v}
	\end{align}
	and
	\begin{align}
	\prod_{j:~v_j^\prime=1} (-1)^{u_j^\prime s_j^\prime} &= \prod_{j:~v_j=1} (-1)^{u_j s_j}.
	\end{align}
	Finally, since $\braket{u_1}{s_1} = 1$,
	\begin{align}
	\prod_{j:~v_j^\prime=0} \braket{u_j^\prime}{s_j^\prime} &= \prod_{j:~j>1, v_j=0} \braket{u_j}{s_j}=\prod_{j:~ v_j=0} \braket{u_j}{s_j}.
	\end{align}
	We now turn to the calculation of $k$. For conciseness we write $x = \join{0}{a}$, and want to simplify
	\begin{align}
	X(x) U_C = U_C \prod_{p:~x_p = 1}\left(i^{\gamma_k}\prod_j X_j^{F_{pj}} Z_j^{M_{pj}}\right).\label{eqn:measurement-sign-bit-quadratic-product}
	\end{align}
	In particular, we will commute all the $Z$'s to the back to write
	\begin{align}
	X(x) U_C &= (-1)^k i^{u\cdot \gamma} U_C Z(t) X(u).
	\end{align}
	The bit $k$ may be calculated by the following algorithm. First initialise $k := 0$ and set $t$ to be a length $n$ vector of zeros. Then, for each $p$ with $x_p=1$, we want to compute the product
	\begin{align}
	\left(\prod_j X_j^{F_{pj}} Z_j^{M_{pj}})\right) Z(t) X(u).
	\end{align}
	We update $t$ by adding (mod-$2$) the $p^\text{th}$ row of $M$ (i.e., we combine the adjacent $Z$-type operators), then we commute each $X_j$ through the new $Z(t)$, which gives a $(-1)$ for each $j$ for which both $F_{pj}$ and $t_j$ are non-zero. More explicitly, we update $k$ to be $k + F_p\cdot t (\mathrm{mod}\ 2)$, where $F_p$ is the $p^\text{th}$ row vector of $F$. 
	
	Since the first column of $F$ has a $1$ in the first element and $0$ in all others, each $F_p$ for $p>1$ starts with a zero. Therefore, $k$ is not sensitive to the first bit of $t$ except when $p = 1$ (since the first element of $F_1$ is $1$). However $x = (\join{0}{a})$, so $x_1=0$ and so the case $p=1$ does not appear in the product in Eq.~\eqref{eqn:measurement-sign-bit-quadratic-product}.	We will therefore compute the same $k$ bit with the ``trimmed'' data as we would with the original data.
\end{proof}

In order to present the second part of the proof we will need a few lemmas. First, we will prove a useful constraint on the CH-form of a state in the form $\ket{0}\otimes\ket{\sigma}$.
\begin{lem}\label{lem:constraining-fcal}
	Given a stabiliser state,
	\begin{align}
	\ket{0}\otimes\ket{\sigma} = \omega U_C U_H\ket{s},
	\end{align}
	where the CH-form on the right is defined by the data $\F = \left\{F, G, M, \gamma, v, s ,\omega\right\}$, at least one of the following is true:
	\begin{enumerate}
		\item $\omega = 0$;
		\item $\exists q$ such that $v_q = s_q = 0$;
		\item $\exists q,r$ ($q\neq r$) such that $s_q = s_r = 1$ and $v_q = v_r = 0$.
	\end{enumerate}
\end{lem}

\begin{proof}
	We consider the inner product
	\begin{align}
	\left(\bra{1}\otimes\bra{a} \right)\left(\ket{0}\otimes\ket{\sigma}\right) = \braket{1}{0}\braket{a}{\sigma} = 0,
	\end{align}
	where $a\in\{0,1\}^{n}$ defines a computational basis vector. We thus have that for all $a$:
	\begin{align}
	0 = \omega \bra{\join{1}{a}}U_CU_H\ket{s}.
	\end{align}
	If $\omega = 0$ we are in case 1, otherwise we divide by $\omega$ to get
	\begin{align}
	0 = \bra{\join{1}{a}}U_CU_H\ket{s}.
	\end{align}
	Applying Eq.~\eqref{eq:bravyi-inner-product} we obtain
	\begin{align}
	0 = \bra{0}^{\otimes n} X((\join{1}{a}) F) U_H \ket{s},
	\end{align}
	where $X(b)$ denotes the tensor-product unitary applying $X^{b_i}$ to qubit $i$.The above is equivalent to
	\begin{align}
	0 = \prod_{j:~v_j = 0} \braket{[(\join{1}{a}) F]_j}{s_j}.
	\end{align}
	We therefore have at least one $j$ such that $v_j = 0$. We first consider the case where there is exactly one $j$ such that $v_j = 0$. Then, for this $j$, we have $\forall a\in\{0,1\}^{n}$ 
	\begin{align}
	\braket{[(\join{1}{a})F]_j }{s_j} = 0. \label{eqn:j-bitstring-mul-equals-zero}
	\end{align}
	Choosing $a = 00\hdots 0$ and computing the matrix multiplication $(\join{1}{a})F$ we obtain
	\begin{align}
	F_{1j} \neq s_j.
	\end{align}
	Now choosing (for each $k$ individually) $a_k = \delta_{kj}$, we obtain
	\begin{align}
	F_{1j} + F_{kj} \neq s_j.
	\end{align}
	This implies that for $k > 1$ we have $F_{kj} = 0$, and since the column $F_{:,j}$ cannot consist of entirely zeros as $F$ is invertible (indeed the inverse of $F$ is the transpose of $G$), we have
	\begin{align}
	F_{1j} = 1 \implies s_j = 0.
	\end{align}
	We are therefore in case 2. We note that the assumption that exactly one of the $v_j$ is equal to zero is necessary in the above, to allow us to change the bitstring $a$ without the $j$ in Eq.~\eqref{eqn:j-bitstring-mul-equals-zero} changing.
	
	Finally we consider the case where there exist distinct $j,k$ such that $v_j = v_k = 0$. If either of $s_j$ or $s_k$ equals $0$ we are in case 2, otherwise both are equal to $1$ and we are in case 3.
\end{proof}

In what follows we will neglect case $1$ since if $\omega=0$ the state is independent of the rest of the CH-form and all computations are trivial. We now provide two lemmas that show that any CH-form for a tensor-product state $\ket{0}\otimes\ket{\sigma}$ may be efficiently brought into a convenient form.

\begin{lem}\label{lem:fixing-fcal-1}
	Given $\omega \neq 0$ and 
	\begin{align}
	\ket{0}\otimes \ket{\sigma} = \omega U_C U_H\ket{s},
	\end{align}
	where the CH-form on the right is given by $\F = \left\{F, G, M, \gamma, v, s ,\omega\right\}$, we can compute $\F^\prime = \left\{F^\prime, G^\prime, M^\prime, \gamma^\prime, v^\prime, s^\prime ,\omega^\prime\right\}$ defining the same state such that there is at most one index $j$ with $v^\prime_j = 0$, $s^\prime_j = 1$.
\end{lem}

\begin{proof}
	Assume there are multiple indices $j$ such that $v_j = 0$ and $s_j = 1$. We recall that the controlled $X$ gate, $\CXgate{p}{q}$, is its own inverse, so we have 
	\begin{align}
	\omega U_C U_H \ket{s} = \omega U_C  \CXgate{p}{q} \CXgate{p}{q} U_H \ket{s},
	\end{align}
	for all $p\neq q$. Let $a$ be the least index such that $v_a = 0$ and $s_a = 1$. Then, for all $b > a$ such that $v_b = 0$ and $s_b = 1$, we insert a pair of controlled $X$ gates controlled on $a$ and targeted on $b$. Since $\CXgate{a}{b}$ is its own inverse, this insertion does not change the quantum state we are representing. We then let the left hand $\CXgate{a}{b}$ act on $U_C$ in accordance with the update rules given in Eqs.\eqref{eq:update1}-\eqref{eq:update3}, while the right hand $\CXgate{a}{b}$ acts on the state $U_H \ket{s}$. Since $v_a = v_b = 0$ and $s_a = s_b = 1$, the action of this is to flip $s_b$ to $0$. 
\end{proof}

\begin{lem}\label{lem:first-row-of-g-is-nice}
	Consider $\omega \neq 0$ and 
	\begin{align}
	\ket{0}\otimes\ket{\sigma} = \omega U_C U_H\ket{s},
	\end{align}
	where the CH-form on the right is given by $\F = \left\{F, G, M, \gamma, v, s ,\omega\right\}$, and assume there is at most one $j$ such that $v_j=0$, and $s_j=1$. Then, the first row of $G$ is non-zero only for elements $G_{1p}$ for which $s_p = v_p = 0$.  
\end{lem}

\begin{proof}
	First assume there is no $j$ such that $v_j = 0$ while $s_j = 1$. Let $p$ be an index such that $G_{1p} = 1$ and let $x = e_p G^T$, where $e_p$ is the vector which has $1$ in the $p^\text{th}$ entry and $0$ in all other entries. We consider the inner product of $\bra{x}(\ket{0}\otimes\ket{\sigma})$. Since $G_{1p} = 1$, we have that $x_0 = 1$, so the inner product equals $0$. From Eq.~\eqref{eq:bravyi-inner-product} we read
	\begin{align}
	0 &= \prod_{j:~v_j = 0} \braket{(x F)_j }{0}=\prod_{j:~v_j = 0} \braket{(e_p G^T F)_j }{0}= \prod_{j:~v_j = 0} \braket{(e_p)_j }{0}\implies v_p = 0,
	\end{align}
	since $G^T F$ is the identity matrix and $(e_p)_j = \delta_{pj}$.
	
	Now assume there exists a single index $k$ such that $v_k=0$, $s_k = 1$. Consider the inner product 
	\begin{align}
	\bra{e_k G^T}(\ket{0}\otimes\ket{\sigma}) &= a \prod_{j:~v_j=0} \braket{(e_k G^T F)_j}{s_j}= a \prod_{j:~v_j=0} \braket{(e_k)_j}{s_j}= a \prod_{j\neq k :~v_j=0} \braket{(e_p)_j}{0} \cdot \braket{(e_k)_k}{1}= a,
	\end{align}
	where $a\neq 0$ is a constant given by Eq.~\eqref{eq:bravyi-inner-product}. Since $a\neq 0$ we have $\bra{e_k G^T}\ket{0}\ket{\sigma} \neq 0$, which implies $(e_k G^T)_1 = 0$, and therefore $G_{1k} = 0$.
	
	Finally, for a $p\neq k$ such that $G_{1p} = 1$ consider $x = (e_p + e_k)G^T$. Since $G_{1k} = 0$ and $G_{1p} = 1$, we have $x_1 = 1$ and hence
	
	\begin{align}
	0 &= \bra{x}(\ket{0}\otimes\ket{\sigma})= \prod_{j:~v_j=0} \braket{((e_k + e_p)G^T F)_j}{s_j}= \prod_{j:~v_j=0} \braket{(e_k + e_p)_j}{s_j}\nonumber\\
	&= \prod_{j\neq k :~v_j=0} \braket{(e_p + e_k)_j}{0} \cdot \braket{(e_p + e_k)_k}{1}= \prod_{j\neq k :~v_j=0} \braket{(e_p)_j}{0}\implies v_p = 0.
	\end{align}
\end{proof}

We now have all the ingredients to present the last part of the proof. We first ensure that $G_{11} = 1$. If this not the case, we choose a $q$ such that $G_{1q} = 1$ (such a $q$ exists since $G$ is invertible) and insert a pair of swap gates using the identity
\begin{align}
\omega U_C U_H \ket{s} = \omega U_C \SWAPgate{1}{q}\SWAPgate{1}{q} U_H \ket{s},
\end{align}
multiply the left hand $SWAP$ onto $U_C$ (where it swaps the first and $q^\text{th}$ column of $G$), and apply the right hand one to $U_H\ket{s}$ (where it swaps the first bit of $v$ with the $q^\text{th}$ bit of $v$ and the first bit of $s$ with the $q^\text{th}$ bit of $s$).

The formula $G^T F = I$ implies that (the sums below are $\text{mod~} 2$)
\begin{align}
	\sum_{p} G_{1p}F_{:,p} &= \sum_{p: G_{1p} =1} F_{:,p} = e_1^T,
\end{align}
since $e_1^T$ is the first column of the identity matrix. We now consider all the indices $p > 1$ such that $G_{1p} = 1$., Lemma~\ref{lem:first-row-of-g-is-nice} implies that for such a $p$ the equation
\begin{align}
	\omega U_C U_H \ket{s} = \omega U_C \CXgate{p}{1} U_H \ket{s}
\end{align}
holds, since $v_p = s_p = 0$ implies the $p^\text{th}$ qubit of $U_H \ket{s}$ is in the state $\ket{0}$. Right-multiplying this $\CXgate{p}{1}$ onto $U_C$ causes the $p^\text{th}$ column of the $F$ matrix to be added onto the $1^\text{st}$ column. The right-multiplication does not alter the first row of the $G$ matrix. We therefore have the identity
\begin{align}
\omega U_C U_H \ket{s} = \omega U_C\prod_{p:G_{1p} = 1} \CXgate{p}{1} U_H \ket{s},
\end{align}
resulting in the update to the $F$ matrix
\begin{align}
	F_{:,1} \gets \bigoplus_{p:G_{1p}=1} F_{:,p}   = e_1^T.
\end{align}

\section{Details of hidden-shift circuit data}
\label{sec:supp-hidden-shift-details}
Looking more closely in to the hidden-shift data we notice an interesting fact. One can take each hidden-shift instance and add up the $40$ compressed $T$-counts we get from the single qubit measurements on the $40$ qubits. We performed this calculation for the $10,000$ hidden-shift instances with $8$ CCZ gates summarised in Fig.~\ref{fig:compress-hidden-shift-data} of the main text, and for a subsequent $10,000$ instances with $16$ CCZ gates. We observe that for every one of the $8$ CCZ gate instances the sum of the compressed $T$-counts is $96$, similarly for every one of the $16$ CCZ instances the sum of the compressed $T$-counts is $192$. We do not provide the $800,000$ compressed $T$-counts here, but summarised data is provided in table~\ref{tab:supplemental-material-hidden-shift-data}. Note that for each class of circuits listed in table~\ref{tab:supplemental-material-hidden-shift-data} the total compressed $T$-count is $10,000\times 12\times\text{ the number of CCZ gates}$.

Despite the total compressed $T$-count being the same for every hidden-shift instance the individual compressed $T$-counts are distributed differently for each instance. The first $20$ measurements for each instance are deterministically solved by $\qcompress$ so have compressed $T$-counts of $0$. Focusing on the $8$ CCZ gate data the remaining $20$ qubits have $96$ compressed $T$-gates distributed between them. Since the run-time of our algorithms is exponential in the compressed $T$-count we observe the fastest performance when these $96$ are distributed as evenly as possible. We observe examples where up to $32$ of the $T$ gates are allocated to a single qubit. 

\begin{table}[h]
    \centering
    \newtext{
    \begin{tabular}{c|cc|cc}
         & \multicolumn{2}{c|}{8 CCZ gates} & \multicolumn{2}{c}{16 CCZ gates}\\
         Compressed T count & frequency & \hspace{0.5cm} frequency$\times$count  & frequency  & \hspace{0.5cm} frequency$\times$count \\\hline
         0 & 304,307 &0& 254,420 &0\\
         8 & 73,895&591,160& 77,236 &617,888\\
         16 &19,391 &310,256& 47,187&754,992\\
         24 & 2,305&55,320&16,834 &404,016\\
         32 & 102&3,264& 3770&120,640\\
         40 & 0&0& 511&20,440\\
         48 & 0&0& 41&1,968\\
         56 & 0&0& 1&56\\
         total & 400,000 &960,000 & 400,000 &1,920,000
    \end{tabular}}
    \newtext{\caption{For each listed CCZ count we randomly generate 10,000 hidden-shift circuits, for each of these we run the $\qcompress$ algorithm 40 times, each corresponding to a measurement on one of the 40 qubits in the circuit and report the compressed $T$-counts we observe. The $8$ CCZ count data is plotted in figure~\ref{fig:compress-hidden-shift-data} in the main text.}}
    \label{tab:supplemental-material-hidden-shift-data}
\end{table}

\bibliography{bibRef}

\end{document}